\documentclass[12pt]{article}
\usepackage{color}
\usepackage{bm}
\usepackage{amsfonts}
\usepackage{amsmath}
\usepackage{amssymb}
\usepackage{amsthm}
\usepackage{graphicx}
\usepackage{xfrac}
\usepackage{multirow}
\usepackage{dcolumn}
\usepackage{threeparttable}
\usepackage{threeparttablex}
\usepackage{longtable} 
\usepackage{natbib}
\usepackage{pdflscape}
\usepackage{tabularx}
\usepackage{booktabs}
\usepackage[font=bf,justification=centering]{caption}
\usepackage{siunitx}
\usepackage[top=1.25in, bottom=1.25in, left=1.25in, right=1.25in]{geometry}
\usepackage{setspace}
\usepackage[toc,page]{appendix}
\usepackage[hidelinks,bookmarks]{hyperref}

\newtheorem{theorem}{Theorem}[section]

\theoremstyle{definition}
\newtheorem{assumption}[theorem]{Assumption}
\theoremstyle{definition}

\theoremstyle{plain}

\newtheorem{definition}{Definition}[section]

\theoremstyle{plain}

\newtheorem{proposition}{Proposition}[section]

\hypersetup{pdfstartview={XYZ null null 1.00}} 
\allowdisplaybreaks
\csname @openrightfalse\endcsname 

\title{Identification and Inference Under Narrative Restrictions\thanks{We thank Isaiah Andrews, Sophocles Mavroeidis, Jos\'{e} Luis Montiel-Olea, Mikkel Plagborg-M{\o}ller, Morten Ravn, Christian Wolf and seminar participants at several venues for helpful comments. We gratefully acknowledge financial support from ERC grants (numbers 536284 and 715940) and the ESRC Centre for Microdata Methods and Practice (CeMMAP) (grant number RES-589-28-0001).}\\
}
\date{\today}
\author{Raffaella Giacomini\thanks{University College London, Department of Economics/Cemmap. Email: r.giacomini@ucl.ac.uk}, Toru Kitagawa\thanks{University College London, Department of Economics/Cemmap. Email: t.kitagawa@ucl.ac.uk} and Matthew Read\thanks{University College London, Department of Economics. Email: matthew.read.16@ucl.ac.uk}
}

\begin{document}

\maketitle

\begin{abstract}
\noindent We consider structural vector autoregressions subject to `narrative restrictions', which are inequality restrictions on functions of the structural shocks in specific periods. These restrictions raise novel problems related to identification and inference, and there is currently no frequentist procedure for conducting inference in these models. We propose a solution that is valid from both Bayesian and frequentist perspectives by: 1) formalizing the identification problem under narrative restrictions; 2) correcting a feature of the existing (single-prior) Bayesian approach that can distort inference; 3) proposing a robust (multiple-prior) Bayesian approach that is useful for assessing and eliminating the posterior sensitivity that arises in these models due to the likelihood having flat regions; and 4) showing that the robust Bayesian approach has asymptotic frequentist validity. We illustrate our methods by estimating the effects of US monetary policy under a variety of narrative restrictions.
\end{abstract}

\textbf{JEL classification:} C32, E52

\textbf{Keywords:} Frequentist coverage, global identification, identified set, multiple priors

\newpage

\section{Introduction}
\label{sec:introduction}

Estimating the dynamic causal effects of structural shocks is a key challenge in macroeconomics. A common approach to this problem is to use a structural vector autoregression (SVAR) with sign or zero restrictions on the model's structural parameters. Recently, a number of papers have augmented these restrictions with restrictions that involve the values of the structural shocks in specific periods. For example, \cite{Antolin-Diaz_Rubio-Ramirez_2018} (AR18) propose restricting the signs of structural shocks and their contributions to the change in particular variables in certain historical episodes. Ludvigson, Ma and Ng (2018)\nocite{Ludvigson_Ma_Ng_2018} independently propose restricting the sign or magnitude of the structural shocks in specific periods. A burgeoning empirical literature has adopted similar restrictions, including \cite{BenZeev_2018}, \cite{Furlanetto_Robstad_2019}, \cite{Cheng_Yang_2020}, \cite{Inoue_Kilian_2020}, Kilian and Zhou (2020a, 2020b)\nocite{Kilian_Zhou_2020a}\nocite{Kilian_Zhou_2020b}, \cite{Laumer_2020}, \cite{Redl_2020}, \cite{Zhou_2020} and Ludvigson, Ma and Ng (2020)\nocite{Ludvigson_Ma_Ng_2020}. The fact that these restrictions are placed on the shocks rather than the parameters raises novel problems related to identification, estimation and inference. This paper clarifies the nature of these problems and proposes a solution that is valid from both Bayesian and frequentist perspectives.

Henceforth, we refer to any restrictions that can be written as inequalities involving structural shocks in particular periods as `narrative restrictions' (NR). An example of NR are `shock-sign restrictions', such as the restriction in AR18 that the US economy was hit by a positive monetary policy shock in October 1979. This is when the Federal Reserve markedly increased the federal funds rate following Paul Volcker becoming chairman, and is widely considered an example of a positive monetary policy shock (e.g., \cite{Romer_Romer_1989}). AR18 also consider `historical-decomposition restrictions', such as the restriction that the change in the federal funds rate in October 1979 was overwhelmingly due to a monetary policy shock. This is an inequality restriction that simultaneously constrains the historical decomposition of the federal funds rate with respect to all structural shocks in the SVAR. Other restrictions on the structural shocks also fit into this framework. For example, we additionally consider `shock-rank restrictions', such as the restriction that the monetary policy shock in October 1979 was the largest positive realization of this shock in the sample period.

From a frequentist perspective, NR are fundamentally different from traditional identifying restrictions, such as sign restrictions on impulse responses (e.g., \cite{Uhlig_2005}). Under normally distributed structural shocks, traditional sign restrictions induce set-identification, because they generate a set-valued mapping from the SVAR's reduced-form parameters to its structural parameters that represents observational equivalence (i.e., an identified set). This set-valued mapping corresponds to the flat region of the structural-parameter likelihood and, by the definition of observational equivalence (e.g., \cite{Rothenberg_1971}), does not depend on the realization of the data. NR also result in the structural-parameter likelihood possessing flat regions and hence generate a set-valued mapping from the reduced-form parameters to the structural parameters. Crucially, this mapping depends not only on the reduced-form parameters, but also on the realization of the data. The data-dependence of this mapping implies that the standard concept of an identified set does not apply. In turn, this means that: 1) it is unclear whether NR are point- or set-identifying restrictions; and 2) there is no known valid frequentist procedure to conduct inference in these models.\footnote{Ludvigson et al. (2018, in press)\nocite{Ludvigson_Ma_Ng_2018}\nocite{Ludvigson_Ma_Ng_2020} conduct inference using a bootstrap procedure, but its frequentist validity is unknown.}

From a Bayesian perspective, AR18 and the empirical papers that adopt their approach conduct standard (single-prior) Bayesian inference under NR in much the same way as under traditional sign restrictions. However, we highlight two features of this approach that can spuriously affect inference. First, the \textit{conditional} likelihood used by AR18 to construct the posterior (distribution) implies that, for some types of NR, a component of the prior (distribution) is updated only in the direction that makes the NR unlikely to hold ex ante. This occurs because the numerator of the conditional likelihood -- the likelihood of the reduced-form VAR -- is flat with respect to the orthonormal matrix that maps reduced-form VAR innovations into structural shocks, whereas the denominator -- the ex ante probability that the NR hold -- depends on this matrix. Second, standard Bayesian inference under NR may be sensitive to the choice of prior when the NR yield a likelihood with flat regions. A flat likelihood implies that the conditional posterior of the orthonormal matrix is proportional to its conditional prior whenever the likelihood is nonzero. Posterior inference may therefore be sensitive to the choice of conditional prior for the orthonormal matrix. This is a problem that also occurs in set-identified models under traditional restrictions (e.g., \cite{Poirier_1998}).

To address the above issues, we study identification under NR and propose a framework for conducting estimation and inference that is potentially appealing to both Bayesians and frequentists. We proceed in four main steps. First, we formalize the identification problem under NR. Second, we propose a simple modification of the existing Bayesian approach that eliminates the source of posterior distortion arising under NR. The modification is to use the \textit{unconditional} likelihood, rather than the conditional likelihood, to construct the posterior. Third, as a tool for assessing and/or eliminating posterior sensitivity occurring due to the likelihood having flat regions, we propose a robust (multiple-prior) Bayesian approach to estimation and inference. Finally, we show that the robust Bayesian approach has frequentist validity in large samples.

To the best of our knowledge, this is the first paper to formally study identification under general NR. \cite{Plagborg-Moller_Wolf_2020b} suggest that shock-sign restrictions, in particular, could in principle be recast as an external instrument (or `proxy') and used to point-identify impulse responses in a proxy SVAR or local projection framework. We explore this idea in Appendix~\ref{sec:narrativeproxy} and highlight the potential sensitivity of this approach to the realization of the unrestricted shocks in the time periods that enter the NR. Petterson, Seim and Shapiro (2020)\nocite{Petterson_Seim_Shapiro_2020} derive bounds for a slope parameter in a single equation given restrictions on the plausible magnitude of the residuals, but the restrictions are over the entire sample and the setting is non-probabilistic.

We make two main contributions to the study of identification under NR. First, we provide a necessary and sufficient condition for global identification of an SVAR under NR and show that this condition is satisfied in a simple bivariate example with a single shock-sign restriction. That is, in contrast with traditional sign restrictions, NR may be formally point-identifying despite generating a set-valued mapping from reduced-form to structural parameters in any particular sample. However, this point-identification result does not deliver a point estimator, because the observed likelihood is almost always flat at the maximum. Second, to develop a frequentist-valid procedure for inference, we introduce the notion of a `conditional identified set'. The conditional identified set extends the standard notion of an identified set to a setting where identification is defined in a repeated sampling experiment conditional on the set of observations entering the NR. This provides an interpretation for the set-valued mapping induced by the NR as the set of observationally equivalent structural parameters in such a conditional frequentist experiment.

In terms of inference under NR, this paper makes contributions from both a Bayesian and a frequentist point of view.

The paper's contribution to Bayesian inference is to address the issues associated with the current approach to standard Bayesian inference under NR. First, we advocate using the \textit{unconditional} likelihood -- the joint probability of observing the data and the NR being satisfied -- when constructing the posterior, rather than the conditional likelihood. Regardless of the type of NR imposed, the unconditional likelihood is flat with respect to the orthonormal matrix that maps reduced-form VAR innovations into structural shocks. This removes the source of posterior distortion that arises due to conditioning on the NR holding. Standard Bayesian inference under the unconditional likelihood requires a simple change to existing computational algorithms. Second, to address posterior sensitivity to the choice of prior, we adapt the robust Bayesian approach of \cite{Giacomini_Kitagawa_2020a} (GK) to a setting with NR.

In the context of an SVAR under traditional identifying restrictions, the robust Bayesian approach of GK involves decomposing the prior for the structural parameters into a prior for the reduced-form parameters, which is revised by the data, and a conditional prior for the orthonormal matrix given the reduced-form parameters, which is unrevisable. Considering the class of all conditional priors for the orthonormal matrix that are consistent with the identifying restrictions generates a class of posteriors, which can be summarized by a set of posterior means (an estimator of the identified set) and a robust credible region. This removes the source of posterior sensitivity.\footnote{Giacomini, Kitagawa and Read (2019)\nocite{Giacomini_Kitagawa_Read_2019} extend this approach to proxy SVARs where the parameters of interest are set-identified using external instruments.}

We show that this approach can also be used to summarize posterior sensitivity under NR, since the unconditional likelihood at the realized data possesses flat regions and the posterior can therefore be sensitive to the choice of prior, as in standard set-identified models. There are, however, some modifications needed to account for the novel features of the NR. In particular, one cannot use a conditional prior for the orthonormal matrix to impose the NR due to the data-dependent mapping between reduced-form and structural parameters. However, by considering the class of all conditional priors consistent with any traditional identifying restrictions (if present), one can trace out all possible posteriors that are consistent with the traditional restrictions \textit{and} the NR. This is because traditional restrictions truncate the support of the conditional prior, while NR truncate the support of the likelihood. Consequently, the posterior given any particular conditional prior is only supported on the common support of the conditional prior and the likelihood.

If the researcher has a credible conditional prior, we recommend reporting the standard Bayesian posterior under the unconditional likelihood together with the robust Bayesian output. This allows other researchers to assess the extent to which posterior inference may be driven by prior choice. In the absence of a credible conditional prior, the robust Bayesian output should be reported as an alternative to the standard Bayesian posterior.

The paper's contribution to frequentist inference is to provide an asymptotically valid approach to inference under NR, which, to the best of our knowledge, was not previously available. To explore the asymptotic frequentist properties of our robust Bayesian procedure, we assume a fixed number of NR. This assumption is empirically relevant given that applications typically impose no more than a handful of NR. We provide conditions under which the robust credible region provides asymptotically valid frequentist coverage of the conditional identified set for the impulse response. Since the conditional identified set is guaranteed to include the true impulse response, the robust credible region also provides valid coverage of the true impulse response. Our robust Bayesian approach should therefore appeal to  Bayesians as well as frequentists.

We illustrate our methods by estimating the effects of monetary policy shocks in the United States. We find that posterior inferences about the response of output obtained under restrictions based on the October 1979 episode may be sensitive to the choice of conditional prior for the orthonormal matrix. In contrast, under an extended set of restrictions constructed by AR18 based on multiple historical episodes, output falls with high posterior probability following a positive monetary policy shock regardless of the choice of conditional prior. We also estimate the set of output responses that are consistent with the restriction that the monetary policy shock in October 1979 was the largest positive realization of the shock in the sample period. Compared with the extended set of restrictions, this shock-rank restriction results in broadly similar robust posterior inferences about the output response.

\bigskip

\noindent\textbf{Outline.} The remainder of the paper is structured as follows. Section~\ref{sec:bivariate} highlights the econometric issues that arise when imposing NR using a simple bivariate example. Section~\ref{sec:framework} describes the general SVAR($p$) framework. Section~\ref{sec:identification} formally analyzes identification under NR and introduces the concept of a conditional identified set. Section~\ref{sec:posteriorinferenceunderNR} discusses how to conduct standard and robust Bayesian inference under NR. Section~\ref{sec:fewNR} explores the frequentist properties of the robust Bayesian approach. Section~\ref{sec:empirical} contains the empirical application and Section~\ref{sec:conclusion} concludes. The appendices contain proofs and other supplemental material.

\bigskip

\noindent\textbf{Generic notation:} For the matrix $\mathbf{X}$, $\mathrm{vec}(\mathbf{X})$ is the vectorization of $\mathbf{X}$ and $\mathrm{vech}(\mathbf{X})$ is the half-vectorization of $\mathbf{X}$ (when $\mathbf{X}$ is symmetric). $\mathbf{e}_{i,n}$ is the $i$th column of the $n\times n$ identity matrix, $\mathbf{I}_{n}$. $\mathbf{0}_{n\times m}$ is a $n\times m$ matrix of zeros. $1(.)$ is the indicator function. $\lVert . \rVert$ is the Euclidean norm.

\section{Bivariate example}
\label{sec:bivariate}

This section sets out the econometric issues that arise when imposing NR using the simplest possible SVAR as an example. Consider the SVAR($0$) $\mathbf{A}_{0}\mathbf{y}_{t} = \bm{\varepsilon}_{t}$, for $t=1,\ldots,T$, where $\mathbf{y}_{t} = (y_{1t},y_{2t})'$ and $\bm{\varepsilon}_{t} = (\varepsilon_{1t},\varepsilon_{2t})'$ with $\bm{\varepsilon}_{t} \overset{iid}{\sim} N(\mathbf{0}_{2\times 1},\mathbf{I}_{2})$. We abstract from dynamics for ease of exposition, but this is without loss of generality. The orthogonal reduced form of the model reparameterizes
$\mathbf{A}_{0}$ as $\mathbf{Q}'\bm{\Sigma}_{tr}^{-1}$, where $\bm{\Sigma}_{tr}$ is the lower-triangular Cholesky factor (with positive diagonal elements) of $\bm{\Sigma} = \mathbb{E}(\mathbf{y}_{t}\mathbf{y}_{t}') = \mathbf{A}_{0}^{-1}\left(\mathbf{A}_{0}^{-1}\right)'$. We parameterize $\bm{\Sigma}_{tr}$ directly as
\begin{equation}
  \bm{\Sigma}_{tr} =
  \begin{bmatrix}
    \sigma_{11} & 0 \\
    \sigma_{21} & \sigma_{22}
  \end{bmatrix}
  \quad (\sigma_{11},\sigma_{22} > 0),
\end{equation}
and denote the vector of reduced-form parameters as $\bm{\phi} = \mathrm{vech}(\bm{\Sigma}_{tr})$. $\mathbf{Q}$ is an orthonormal matrix in the space of $2\times 2$ orthonormal matrices, $\mathcal{O}(2)$:
\begin{equation}
  \mathbf{Q} \in \mathcal{O}(2) = \left\{
  \begin{bmatrix}
    \cos \theta & -\sin \theta \\
    \sin \theta & \cos \theta
  \end{bmatrix}
  : \theta \in [-\pi,\pi]\right\} \cup \left\{
  \begin{bmatrix}
    \cos \theta & \sin \theta \\
    \sin \theta & -\cos \theta
  \end{bmatrix}
  : \theta \in [-\pi,\pi]\right\},
\end{equation}
where the first set is the set of `rotation' matrices and the second set is the set of `reflection' matrices.

Given the `sign normalization' $\mathrm{diag}(\mathbf{A}_{0}) \geq \mathbf{0}_{2\times 1}$, the set of values for $\mathbf{A}_{0}$ that are consistent with the reduced-form parameters in the absence of additional restrictions is
\small
\begin{multline}
    \mathbf{A}_{0} \in \left\{\frac{1}{\sigma_{11}\sigma_{22}}
    \begin{bmatrix}
      \sigma_{22}\cos \theta -\sigma_{21}\sin \theta  & \sigma_{11}\sin \theta \\
      -\sigma_{21}\cos \theta -\sigma_{22}\sin \theta & \sigma_{11}\cos \theta
    \end{bmatrix}
    : \sigma_{22}\cos \theta \geq \sigma_{21}\sin \theta, \cos \theta \geq 0, \theta \in [-\pi,\pi]\right\} \\
    \cup \left\{\frac{1}{\sigma_{11}\sigma_{22}}
    \begin{bmatrix}
      \sigma_{22}\cos\theta - \sigma_{21}\sin\theta & \sigma_{11}\sin\theta \\
      \sigma_{22}\sin\theta + \sigma_{21}\cos\theta & -\sigma_{11}\cos\theta
    \end{bmatrix}
    : \sigma_{22}\cos\theta \geq \sigma_{21}\sin\theta, \cos\theta \leq 0, \theta \in [-\pi,\pi]\right\}.
\end{multline}
\normalsize

\subsection{Shock-sign restrictions}
\label{subsec:sign}

Consider the `shock-sign restriction' that $\varepsilon_{1k}$ is nonnegative for some $k \in \left\{1,\ldots,T\right\}$:
\begin{equation}
  \varepsilon_{1k} = \mathbf{e}_{1,2}'\mathbf{A}_{0}\mathbf{y}_{k} = (\sigma_{11}\sigma_{22})^{-1}\left(\sigma_{22}y_{1k}\cos \theta + (\sigma_{11}y_{2k} - \sigma_{21}y_{1k})\sin \theta\right) \geq 0. \label{eq:nsr}
\end{equation}
Given the realization of the data in period $k$, Equation~(\ref{eq:nsr}) implies that the restricted structural shock can be written as a function $\varepsilon_{1k}(\theta,\bm{\phi},\mathbf{y}_{k})$. Under the sign normalization and the shock-sign restriction, $\theta$ is restricted to the set
\begin{multline}
  \theta \in \left\{\theta: \sigma_{21}\sin\theta \leq \sigma_{22}\cos\theta, \cos\theta \geq 0, \sigma_{22}y_{1k}\cos\theta \geq (\sigma_{21}y_{1k}-\sigma_{11}y_{2k})\sin\theta, -\pi \leq \theta \leq \pi \right\} \\
  \cup \left\{\theta: \sigma_{21}\sin\theta \leq \sigma_{22}\cos\theta, \cos\theta \leq 0, \sigma_{22}y_{1k}\cos\theta \geq (\sigma_{21}y_{1k}-\sigma_{11}y_{2k})\sin\theta, -\pi \leq \theta \leq \pi \right\}.
\end{multline}
Since $y_{1k}$ and $y_{2k}$ enter the inequalities characterising this set, the shock-sign restriction induces a set-valued mapping from $\bm{\phi}$ to $\theta$ that depends on the realization of $\mathbf{y}_{k}$. For example,  if $\sigma_{21} < 0$, $\sigma_{21}y_{1k}-\sigma_{11}y_{2k} > 0$ and $y_{1k} > 0$,
\begin{equation}
  \theta \in \left[\arctan\left(\frac{\sigma_{22}}{\sigma_{21}}\right),\arctan\left(\frac{\sigma_{22}y_{1k}}{\sigma_{21}y_{1k}-\sigma_{11}y_{2k}}\right)\right].\footnote{See Appendix~A for the full characterization of this mapping.}
\end{equation}
The direct dependence of this mapping on the realization of the data implies that the standard notion of an identified set -- the set of observationally equivalent structural parameter values given the reduced-form parameters -- does not apply. Consequently, it is not obvious whether existing frequentist procedures for conducting inference in set-identified models are valid under NR. Moreover, it is unclear whether the restrictions are, in fact, set-identifying in a formal frequentist sense. We formally analyze identification under NR in Section~\ref{sec:identification}.

When conducting Bayesian inference, AR18 construct the posterior using the conditional likelihood, which is the likelihood of observing the data conditional on the NR holding. Letting $\mathbf{y}^{T} = (\mathbf{y}_{1}',\ldots,\mathbf{y}_{T}')'$ represent a realization of the random variable $\mathbf{Y}^{T}$, the conditional likelihood is
\begin{equation}
  p\left(\mathbf{y}^{T} | \theta, \bm{\phi}, \varepsilon_{1k}(\theta,\bm{\phi},\mathbf{y}_{k}) \geq 0\right) = \frac{\prod_{t=1}^{T}(2\pi)^{-1}|\bm{\Sigma}|^{-\frac{1}{2}}\exp\left(-\frac{1}{2}\mathbf{y}_{t}'\bm{\Sigma}^{-1}\mathbf{y}_{t}\right)}{\mathrm{Pr}( \varepsilon_{1k} \geq 0 | \theta,\bm{\phi})} 1\left(\varepsilon_{1k}(\theta,\bm{\phi},\mathbf{y}_{k}) \geq 0\right).
\end{equation}
The numerator in the first term is a function of $\bm{\phi}$ and $\mathbf{y}^{T}$, while the denominator is equal to 1/2, because the marginal distribution of $\varepsilon_{1k}$ is standard normal. The conditional likelihood therefore depends on $\theta$ only through the indicator function $1\left(\varepsilon_{1k}(\theta,\bm{\phi},\mathbf{y}_{k}) \geq 0\right)$. This indicator function truncates the likelihood, with the truncation points depending on $\mathbf{y}_{k}$. To illustrate, the left panel of Figure~\ref{fig:LikSign} plots the likelihood given different realizations of the data drawn from a data-generating process with $\sigma_{21} < 0$ and assuming for simplicity that the econometrician knows $\bm{\phi}$.\footnote{The data-generating process assumes $\mathbf{A}_{0} = \begin{bmatrix} 1 & 0.5 \\ 0.2 & 1.2 \end{bmatrix}$, which implies that $\theta = \arcsin(0.5\sigma_{22})$ with $\mathbf{Q}$ equal to the rotation matrix. We assume the time series is of length $T=3$ and draw sequences of structural shocks such that $\varepsilon_{1,1} \geq 0$. $T$ is a small number to control Monte Carlo sampling error in the exercises below. The analysis with known $\bm{\phi}$ replicates the situation with a large sample, where the likelihood for $\bm{\phi}$ concentrates at the truth. The assumption that $\bm{\phi}$ is known also facilitates visualizing the likelihood, which otherwise is a function of four parameters.} The conditional likelihood is flat over the region for $\theta$ satisfying the shock-sign restriction and is zero outside this region. The support of the nonzero region depends on the realization of $\mathbf{y}_{k}$.

\begin{figure}[h]
    \center
    \caption{Shock-sign Restriction} \label{fig:LikSign}
    \begin{tabular}{c c c}
        \includegraphics[scale=0.5]{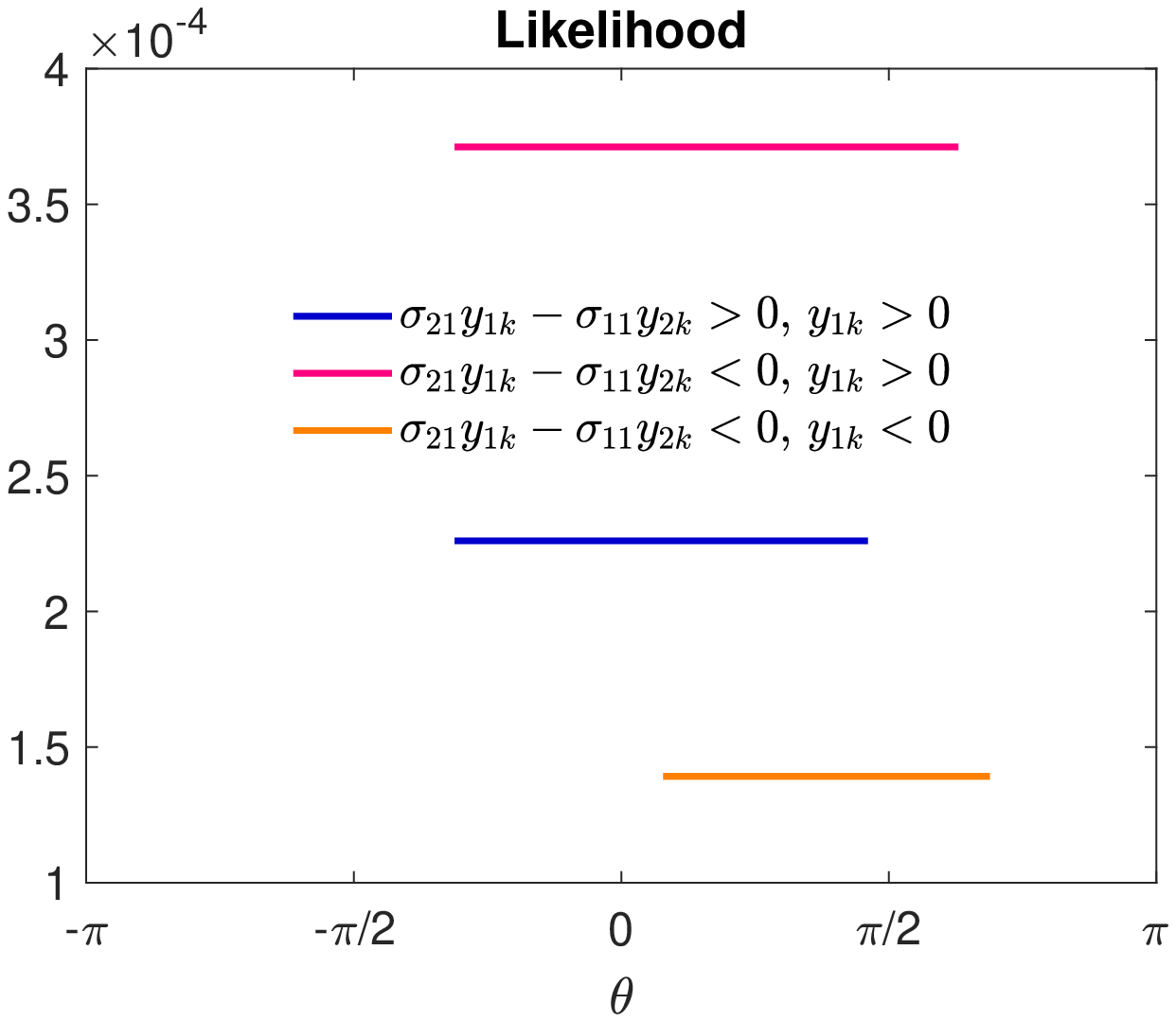} & \phantom{a} & \includegraphics[scale=0.5]{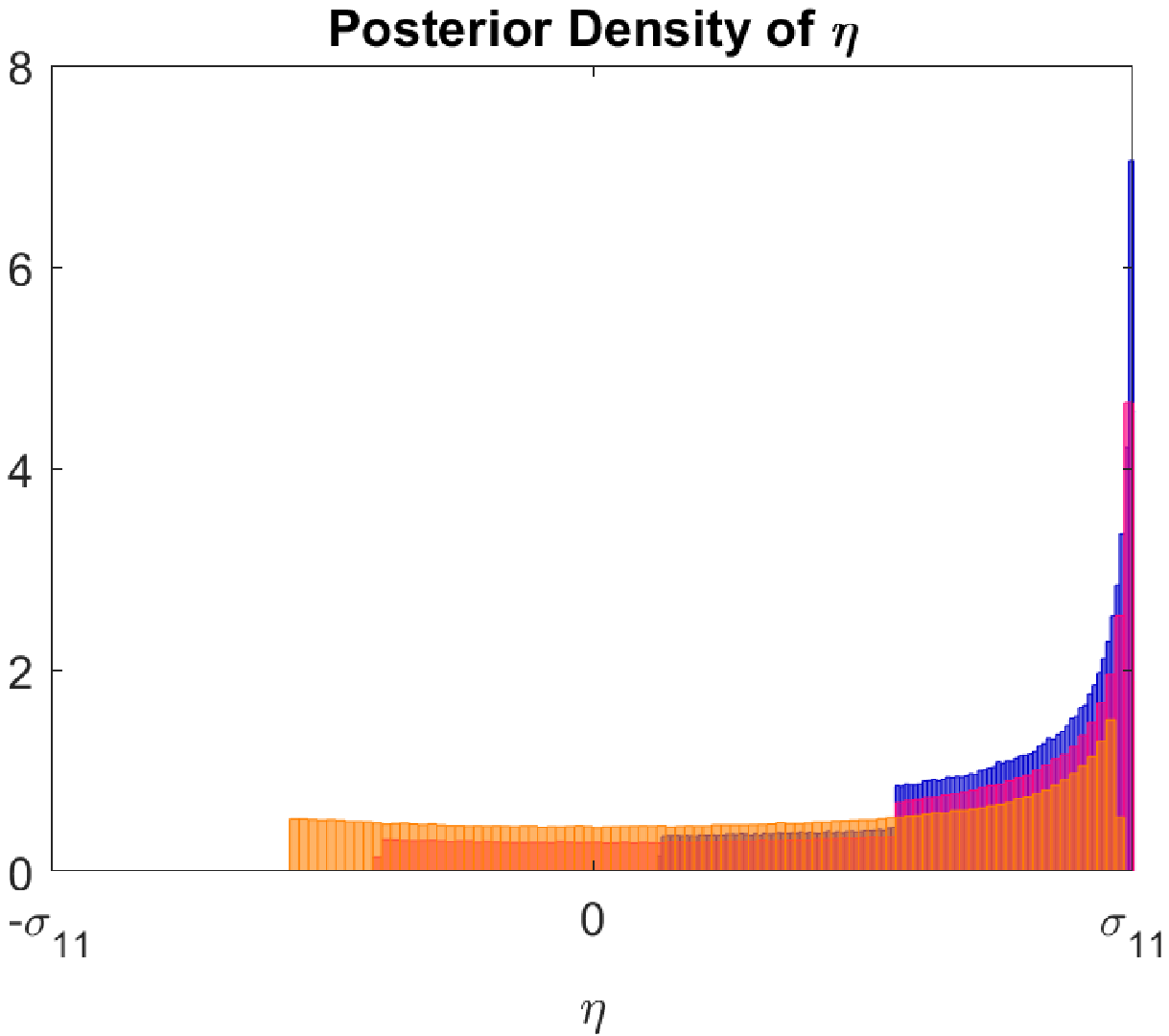} \\
    \end{tabular}
    \footnotesize \parbox[t]{0.65 in}{Notes:}\parbox[t]{5 in}{$T = 3$, $\bm{\phi}$ is known and $\varepsilon_{1k}(\theta,\bm{\phi},\mathbf{y}_{k}) \geq 0$ is the narrative sign restriction; likelihood in top-left panel is zero outside of plotted intervals; posterior density of $\eta = \sigma_{11}\cos\theta$ is approximated using 1,000,000 draws of $\theta$ from its uniform posterior.} \\
\end{figure}

The flat likelihood function implies that the posterior will be proportional to the prior in the region where the likelihood function is nonzero, and it will be zero outside this region. The standard approach to Bayesian inference in SVARs identified via sign restrictions assumes a uniform (or Haar) prior over $\mathbf{Q}$, as does the approach in AR18.\footnote{See, for example, \cite{Uhlig_2005}, Rubio-Ram\'{i}rez, Waggoner and Zha (2010)\nocite{Rubio-Ramirez_Waggoner_Zha_2010}, \cite{Baumeister_Hamilton_2015} and Arias, Rubio-Ram\'{i}rez and Waggoner (2018)\nocite{Arias_Rubio-Ramirez_Waggoner_2018}.} In the bivariate example, this is equivalent to a prior for $\theta$ that is uniform over the interval $[-\pi,\pi]$. This prior implies that the posterior for $\theta$ is also uniform over the interval for $\theta$ where the likelihood function is nonzero.

The impact impulse response of $y_{1t}$ to a positive standard-deviation shock $\varepsilon_{1t}$ is $\eta \equiv \sigma_{11}\cos\theta$. The right panel of Figure~\ref{fig:LikSign} plots the posterior  for $\eta$ induced by a uniform prior over $\theta$ given the same realizations of the data for which the likelihood was plotted in the left panel. The uniform posterior for $\theta$ induces a posterior for $\eta$ that assigns more probability mass to more-extreme values of $\eta$. This highlights that even a `uniform' prior may be informative for parameters of interest, which is also the case under traditional sign restrictions (\cite{Baumeister_Hamilton_2015}). One difference is that the conditional prior under sign restrictions is never updated by the data, whereas the support and shape of the posterior for $\eta$ under NR may depend on the realization of $\mathbf{y}_{k}$ through its effect on the truncation points of the likelihood, so there may be some updating of the conditional prior by the data. For example, when $\sigma_{21} < 0$, $\sigma_{21}y_{1k}-\sigma_{11}y_{2k} > 0$ and $y_{1k} > 0$,
\begin{equation}\label{eq:impulseresponseset}
  \eta \in \left[\sigma_{11}\cos\left(\arctan\left(\max\left\{-\frac{\sigma_{22}}{\sigma_{21}},\frac{\sigma_{22}y_{1k}}{\sigma_{21}y_{1k}-\sigma_{11}y_{2k}}\right\}\right)\right),\sigma_{11}\right].
\end{equation}
However, the conditional prior is not updated at values of $\theta$ corresponding to the flat region of the likelihood. Posterior inference about $\eta$ may therefore still be sensitive to the choice of prior, as in standard set-identified SVARs.

\subsection{Historical-decomposition restrictions}
\label{subsec:histdecomp}

The historical decomposition is the contribution of a particular structural shock to the observed unexpected change in a particular variable over some horizon. The contribution of the first shock to the change in the first variable in the $k$th period is
\begin{equation}
    H_{1,1,k}(\theta,\bm{\phi},\mathbf{y}_{k}) = \sigma_{22}^{-1}\left(\sigma_{22}y_{1k}\cos^{2} \theta + (\sigma_{11}y_{2k} - \sigma_{21}y_{1k})\cos\theta \sin \theta\right),
\end{equation}
while the contribution of the second shock is
\begin{equation}
  H_{1,2,k}(\theta,\bm{\phi},\mathbf{y}_{k}) = \sigma_{22}^{-1}\left(\sigma_{22}y_{1k}\sin^{2}\theta + (\sigma_{21}y_{1k}-\sigma_{11}y_{2k})\cos\theta\sin\theta\right).
\end{equation}
Consider the restriction that the first structural shock in period $k$ was positive and (in the language of AR18) the `most important contributor' to the change in the first variable, which requires that $|H_{1,1,k}(\theta,\bm{\phi},\mathbf{y}_{k})| \geq |H_{1,2,k}(\theta,\bm{\phi},\mathbf{y}_{k})|$. Under these restrictions and the sign normalization, $\theta$ must satisfy a set of inequalities that depends on $\bm{\phi}$ and $\mathbf{y}_{k}$. As in the case of the shock-sign restriction, this set of restrictions generates a set-valued mapping from $\bm{\phi}$ to $\theta$ that depends on $\mathbf{y}_{k}$.\footnote{See Appendix~A for this set of inequalities. It is more difficult to analytically characterize the induced mapping than in the shock-sign example, so we do not pursue this.}

Let $\mathcal{D}(\theta,\bm{\phi},\mathbf{y}_{k}) = 1\{\varepsilon_{1k}(\theta,\bm{\phi},\mathbf{y}_{k}) \geq 0, |H_{1,1,k}(\theta,\bm{\phi},\mathbf{y}_{k})| \geq |H_{1,2,k}(\theta,\bm{\phi},\mathbf{y}_{k})|\}$ represent the indicator function equal to one when the NR are satisfied and equal to zero otherwise, and let $\tilde{\mathcal{D}}(\theta,\bm{\phi},\bm{\varepsilon}_{k}) = 1\{\varepsilon_{1k} \geq 0, |\tilde{H}_{1,1,k}(\theta,\bm{\phi},\varepsilon_{1k})| \geq |\tilde{H}_{1,2,k}(\theta,\bm{\phi},\varepsilon_{2k})|\}$ represent the indicator function for the same event in terms of the structural shocks rather than the data. The conditional likelihood function given the restrictions is then
\begin{equation}
  p\left(\mathbf{y}^{T} | \theta, \bm{\phi}, \mathcal{D}(\theta,\bm{\phi},\mathbf{y}_{k})=1 \right) = \frac{\prod_{t=1}^{T}(2\pi)^{-\frac{n}{2}}|\bm{\Sigma}|^{-\frac{1}{2}}\exp\left(-\frac{1}{2}\mathbf{y}_{t}'\bm{\Sigma}^{-1}\mathbf{y}_{t}\right)}
  {\mathrm{Pr}(\tilde{\mathcal{D}}(\theta,\bm{\phi},\bm{\varepsilon}_{k}) = 1|\theta,\bm{\phi})}\mathcal{D}(\theta,\bm{\phi},\mathbf{y}_{k}).
\end{equation}
As in the case of the shock-sign restriction, the numerator of the first term does not depend on $\theta$. In contrast, the probability in the denominator now depends on $\theta$ through the historical decomposition. Intuitively, changing $\theta$ changes the impulse responses of $y_{1t}$ to the two shocks and thus changes the ex ante probability that $|\tilde{H}_{1,1,k}(\theta,\bm{\phi},\varepsilon_{1k})| \geq |\tilde{H}_{1,2,k}(\theta,\bm{\phi},\varepsilon_{2k})|$. The conditional likelihood therefore depends on $\theta$ both through this probability and through the indicator function determining the truncation points of the likelihood. Consequently, the likelihood function is not necessarily flat when it is nonzero.

To illustrate, the left panel of Figure~\ref{fig:likelihoodHD} plots the conditional likelihood evaluated at a random realization of the data satisfying the restrictions using the same data-generating process as above and assuming that $\bm{\phi}$ is known. The probability in the denominator of the conditional likelihood is approximated by drawing 1,000,000 realizations of $\bm{\varepsilon}_{k}$ and computing the proportion of draws satisfying the restrictions at each value of $\theta$. This probability is plotted in the right panel of Figure~\ref{fig:likelihoodHD}. The likelihood is again truncated according to a set-valued mapping from $\bm{\phi}$ and $\mathbf{y}_{k}$ to $\theta$, but an important difference from the case with the shock-sign restriction is that the likelihood is no longer flat within the region where it is nonzero. In particular, the conditional likelihood has a maximum at the value of $\theta$ that minimizes the ex ante probability that the NR are satisfied (within the set of values of $\theta$ that are consistent with the restrictions). The posterior  for $\theta$ induced by a uniform prior will therefore assign greater posterior probability to values of $\theta$ that yield a lower ex ante probability of satisfying the NR.

\begin{figure}[h]
    \center
    \caption{Historical-decomposition Restriction} \label{fig:likelihoodHD}
    \begin{tabular}{ccc}
        \includegraphics[scale=0.5]{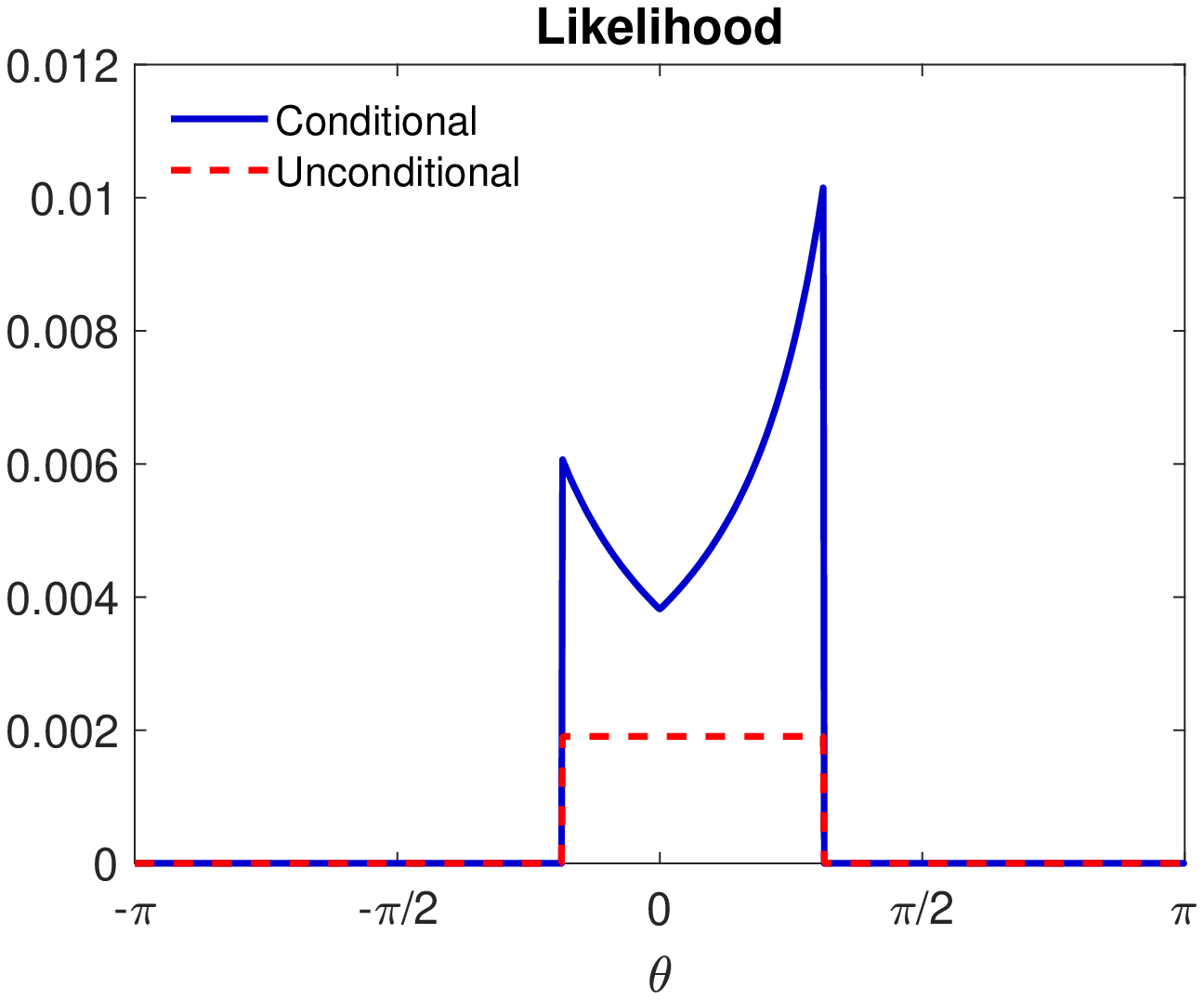} & \phantom{a} & \includegraphics[scale=0.5]{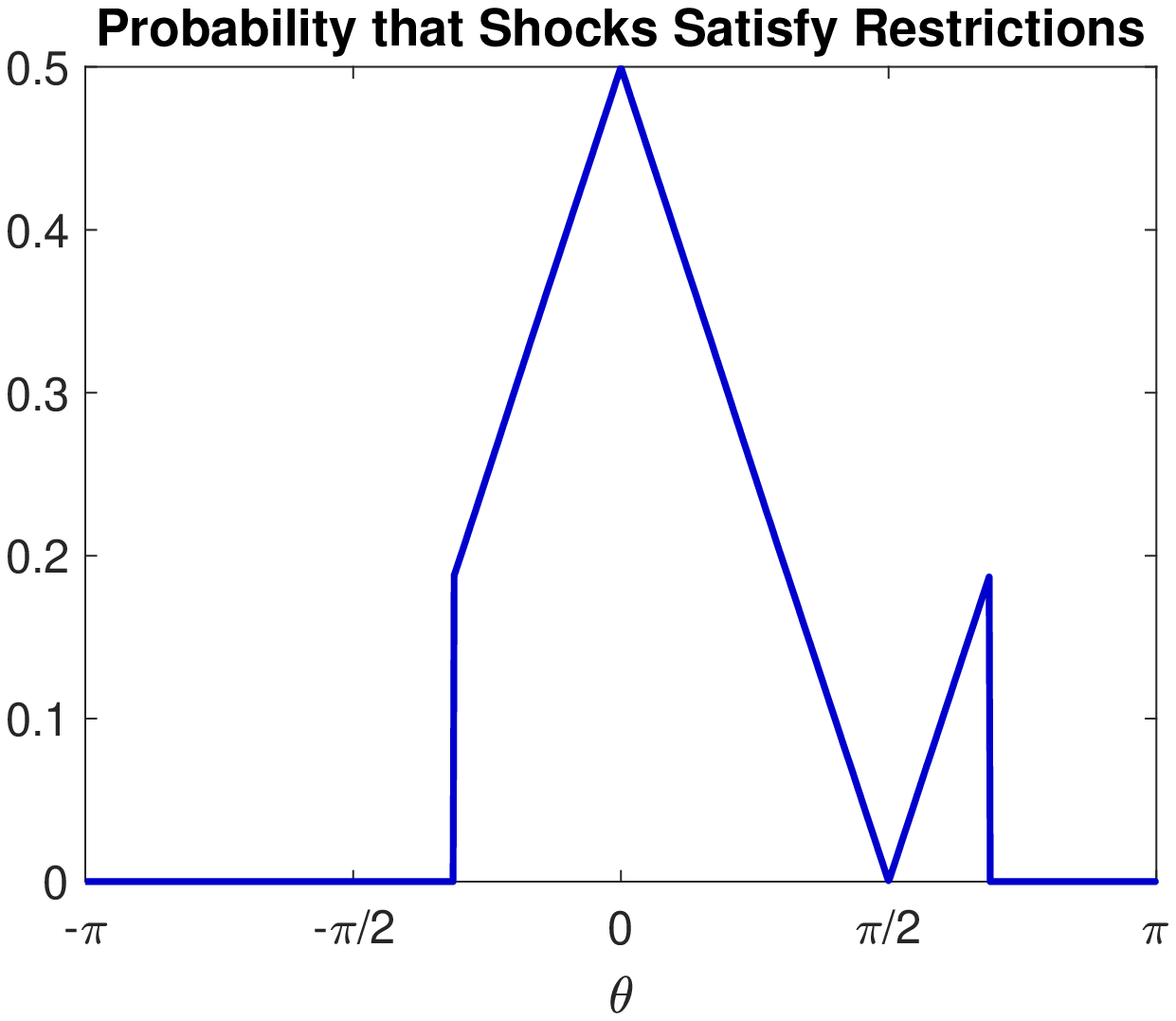} \\
    \end{tabular}
    \footnotesize \parbox[t]{0.65 in}{Notes:}\parbox[t]{5 in}{$T = 3$ and $\bm{\phi}$ is known; $\varepsilon_{1,1}(\bm{\phi},\theta,\mathbf{y}_{k}) \geq 0$ and $|H_{1,1,1}(\bm{\phi},\theta,\mathbf{y}_{k})| \geq |H_{2,1,1}(\bm{\phi},\theta,\mathbf{y}_{k})|$ are the narrative sign restrictions; $\mathrm{Pr}(\tilde{\mathcal{D}}(\theta,\bm{\phi},\bm{\varepsilon}_{k})=1|\theta,\bm{\phi})$ is approximated using 1,000,000 Monte Carlo draws.} \\
\end{figure}

If we view the narrative event as a part of the observables and its probability of occurring depends on the parameter of interest, conditioning on the narrative event implies that we are conditioning on a non-ancillary statistic. When conducting likelihood-based inference, conditioning on a non-ancillary statistic is undesirable, because it represents a loss of information about the parameter of interest. The probability that the shock-sign restriction is satisfied is independent of the parameters, so the event that the restriction is satisfied is ancillary. In the case where there is also a restriction on the historical decomposition, the probability that the NR are satisfied depends on $\theta$, so the event that the NR are satisfied is not ancillary. Conditioning on this non-ancillary event results in the likelihood no longer being flat, but the shape of the likelihood is fully driven by the inverse probability of the conditioning event. That is, the loss of information for $\theta$ can be viewed as distorting the shape of the posterior in the sense that the prior is updated toward values of $\theta$ that make the event that the NR are satisfied less likely ex ante. We therefore advocate forming the likelihood without conditioning on the restrictions holding.

The joint (or unconditional) likelihood of observing the data and the NR holding is obtained by multiplying the conditional likelihood by the probability that the NR are satisfied:
\begin{equation}\label{eq:unlik1}
    p\left(\mathbf{y}^{T}, \tilde{\mathcal{D}}(\theta,\bm{\phi},\bm{\varepsilon}_{k}) = 1 | \theta, \bm{\phi}\right) =
    \prod_{t=1}^{T}(2\pi)^{-\frac{n}{2}}|\bm{\Sigma}|^{-\frac{1}{2}}\exp\left(-\frac{1}{2}\left(\mathbf{y}_{t}'\bm{\Sigma}^{-1}\mathbf{y}_{t}\right)\right)\mathcal{D}(\theta,\bm{\phi},\mathbf{y}_{k}).
\end{equation}
Conditional on being nonzero, the unconditional likelihood is flat with respect to $\theta$. The unconditional likelihood depends on $\theta$ only through the points of truncation. To illustrate, Figure~\ref{fig:likelihoodHD} plots the unconditional likelihood given the same realization of the data used to plot the conditional likelihood. As in the case of the shock-sign restriction, the flat unconditional likelihood implies that posterior inference may be sensitive to the choice of prior. We describe our approach to addressing this posterior sensitivity in Section~\ref{subsec:robustbayes}.

\section{General framework}
\label{sec:framework}

This section describes the general SVAR($p$) and outlines the restrictions that we consider.

\subsection{SVAR($p$)}
\label{subsec:svar}

Let $\mathbf{y}_{t}$ be an $n\times 1$ vector of endogenous variables following the SVAR($p$) process:
\begin{equation}\label{eq:SVARp}
    \mathbf{A}_{0}\mathbf{y}_{t} = \sum_{l=1}^{p}\mathbf{A}_{l}\mathbf{y}_{t-l} + \bm{\varepsilon}_{t}, \quad t=1,...,T,
\end{equation}
where $\mathbf{A}_{0}$ is invertible and $\bm{\varepsilon}_{t}\overset{iid}{\sim} N(\mathbf{0}_{n\times 1},\mathbf{I}_{n})$ are structural shocks. The initial conditions $(\mathbf{y}_{1-p},...,\mathbf{y}_{0})$ are given. We omit exogenous regressors (such as a constant) for simplicity of exposition, but these are straightforward to include. Letting $\mathbf{x}_{t} = (\mathbf{y}_{t-1}',\ldots,\mathbf{y}_{t-p}')'$ and $\mathbf{A}_{+} = (\mathbf{A}_{1},\ldots,\mathbf{A}_{p})$, rewrite the SVAR($p$) as
\begin{equation}\label{eq:SVARpCondensed}
  \mathbf{A}_{0}\mathbf{y}_{t} = \mathbf{A}_{+}\mathbf{x}_{t} + \bm{\varepsilon}_{t}, \quad t=1,...,T.
\end{equation}
$(\mathbf{A}_{0},\mathbf{A}_{+})$ are the structural parameters. The reduced-form VAR($p$) representation is
\begin{equation}\label{eq:VARpCondensed}
  \mathbf{y}_{t} = \mathbf{B}\mathbf{x}_{t} + \mathbf{u}_{t}, \quad t=1,...,T,
\end{equation}
where $\mathbf{B} = (\mathbf{B}_{1},\ldots,\mathbf{B}_{p})$, $\mathbf{B}_{l}=\mathbf{A}_{0}^{-1}\mathbf{A}_{l}$ for $l=1,\ldots,p$, and $\mathbf{u}_{t} = \mathbf{A}_{0}^{-1}\bm{\varepsilon}_{t} \overset{iid}{\sim} N(\mathbf{0}_{n\times 1},\bm{\Sigma})$ with $\bm{\Sigma} = \mathbf{A}_{0}^{-1}(\mathbf{A}_{0}^{-1})'$. $\bm{\phi} = (\mathrm{vec}(\mathbf{B})',\mathrm{vech}(\bm{\Sigma})')' \in \bm{\Phi}$ are the reduced-form parameters. We assume that $\mathbf{B}$ is such that the VAR($p$) can be inverted into an infinite-order vector moving average (VMA($\infty$)) representation.\footnote{The VAR($p$) is invertible into a VMA($\infty$) process when the eigenvalues of the companion matrix lie inside the unit circle. See \cite{Hamilton_1994} or \cite{Kilian_Lutkepohl_2017}.}

As is standard in the literature that considers set-identified SVARs, we reparameterize the model into its orthogonal reduced form (e.g., Arias et al. (2018)\nocite{Arias_Rubio-Ramirez_Waggoner_2018}):
\begin{equation}\label{eq:orthogonalreducedform}
  \mathbf{y}_{t} = \mathbf{B}\mathbf{x}_{t} + \bm{\Sigma}_{tr}\mathbf{Q}\bm{\varepsilon}_{t}, \quad t=1,...,T,
\end{equation}
where $\bm{\Sigma}_{tr}$ is the lower-triangular Cholesky factor of $\bm{\Sigma}$ (i.e. $\bm{\Sigma}_{tr}\bm{\Sigma}_{tr}'=\bm{\Sigma}$) with diagonal elements normalized to be non-negative, $\mathbf{Q}$ is an $n\times n$ orthonormal matrix and $\mathcal{O}(n)$ is the set of all such matrices. The structural and orthogonal reduced-form parameterizations are related through the mapping $\mathbf{B} = \mathbf{A}_{0}^{-1}\mathbf{A}_{+}$, $\bm{\Sigma} = \mathbf{A}_{0}^{-1}(\mathbf{A}_{0}^{-1})'$ and $\mathbf{Q} = \bm{\Sigma}_{tr}^{-1}\mathbf{A}_{0}^{-1}$ with inverse mapping $\mathbf{A}_{0} = \mathbf{Q}'\bm{\Sigma}_{tr}^{-1}$ and $\mathbf{A}_{+} = \mathbf{Q}'\bm{\Sigma}_{tr}^{-1}\mathbf{B}$.

The VMA($\infty$) representation of the model is
\begin{equation}\label{eq:vma}
  \mathbf{y}_{t} = \sum_{h=0}^{\infty}\mathbf{C}_{h}\mathbf{u}_{t-h} = \sum_{h=0}^{\infty}\mathbf{C}_{h}\bm{\Sigma}_{tr}\mathbf{Q}\bm{\varepsilon}_{t}, \quad t=1,...,T,
\end{equation}
where $\mathbf{C}_{h}$ is the $h$th term in $(\mathbf{I}_{n} - \sum_{l=1}^{p}\mathbf{B}_{l}L^{l})^{-1}$ and $L$ is the lag operator. $\mathbf{C}_{h}$ is defined recursively by $\mathbf{C}_{h} = \sum_{l=1}^{\min\{k,p\}}\mathbf{B}_{l}\mathbf{C}_{h-l}$ for $h \geq 1$ with $\mathbf{C}_{0} = \mathbf{I}_{n}$. The $(i,j)$th element of the matrix $\mathbf{C}_{h}\bm{\Sigma}_{tr}\mathbf{Q}$, which we denote by $\eta_{i,j,h}(\bm{\phi},\mathbf{Q})$, is the horizon-$h$ impulse response of the $i$th variable to the $j$th structural shock:
\begin{equation}\label{eq:ir}
  \eta_{i,j,h}(\bm{\phi},\mathbf{Q}) = \mathbf{e}_{i,n}'\mathbf{C}_{h}\bm{\Sigma}_{tr}\mathbf{Q}\mathbf{e}_{j,n} = \mathbf{c}_{i,h}'(\bm{\phi})\mathbf{q}_{j},
\end{equation}
where $\mathbf{c}_{i,h}'(\bm{\phi}) = \mathbf{e}_{i,n}'\mathbf{C}_{h}\bm{\Sigma}_{tr}$ is the $i$th row of $\mathbf{C}_{h}\bm{\Sigma}_{tr}$ and $\mathbf{q}_{j} = \mathbf{Q}\mathbf{e}_{j,n}$ is the $j$th column of $\mathbf{Q}$.

\subsection{Narrative restrictions}
\label{subsec:narrativesignrestrictions}

In the absence of any identifying restrictions, it is well-known that $\mathbf{Q}$ is set-identified. Consequently, functions of $\mathbf{Q}$, such as the impulse responses, are also set-identified. Imposing traditional identifying restrictions on the SVAR is equivalent to restricting $\mathbf{Q}$ to lie in a subspace of $\mathcal{O}(n)$. It is conventional to impose a `sign normalization' on the structural shocks. We normalize the diagonal elements of $\mathbf{A}_{0}$ to be non-negative, so a positive value of $\varepsilon_{it}$ is a positive shock to the $i$th equation in the SVAR at time $t$. The sign normalization implies that $\mathrm{diag}(\mathbf{Q}'\bm{\Sigma}_{tr}^{-1}) \geq \mathbf{0}_{n\times 1}$.

It is common to impose sign restrictions on the impulse responses (e.g., \cite{Uhlig_2005}) or on the structural parameters themselves. For example, the restriction that the horizon-$h$ impulse response of the $i$th variable to the $j$th shock is nonnegative is $c_{i,h}'(\bm{\phi})\mathbf{q}_{j} \geq 0$, which is a linear inequality restriction on a single column of $\mathbf{Q}$ that depends only on the reduced-form parameter $\bm{\phi}$. Restrictions on elements of $\mathbf{A}_{0}$ take a similar form.

In contrast, NR constrain the values of the structural shocks in particular periods. The structural shocks are
\begin{equation}
    \bm{\varepsilon}_{t} = \mathbf{A}_{0}\mathbf{u}_{t} = \mathbf{Q}'\bm{\Sigma}_{tr}^{-1}\mathbf{u}_{t}.
\end{equation}
The shock-sign restriction that the $i$th structural shock at time $k$ is positive is
\begin{equation} \label{eq:structuralshock}
  \varepsilon_{ik}(\bm{\phi},\mathbf{Q},\mathbf{u}_{k}) = \mathbf{e}_{i,n}'\mathbf{Q}'\bm{\Sigma}_{tr}^{-1}\mathbf{u}_{k} = (\bm{\Sigma}_{tr}^{-1}\mathbf{u}_{k})'\mathbf{q}_{i} \geq 0.
\end{equation}
We can treat $\mathbf{u}_{t}$ as observable given $\bm{\phi}$ and the data, so we suppress the dependence of $\mathbf{u}_{t}$ on $\bm{\phi}$ and $(\mathbf{y}_{t}',\mathbf{x}_{t}')'$ for notational convenience. The restriction in (\ref{eq:structuralshock}) is a linear inequality restriction on a single column of $\mathbf{Q}$. In contrast with traditional sign restrictions, the shock-sign restriction depends directly on the data through the reduced-form VAR innovations.

In addition to shock-sign restrictions, AR18 consider restrictions on the historical decomposition, which is the cumulative contribution of the $j$th shock to the observed unexpected change in the $i$th variable between periods $k$ and $k+h$:
\begin{equation}\label{eq:historicaldecomposition}
  H_{i,j,k,k+h}\left(\bm{\phi},\mathbf{Q},\left\{\mathbf{u}_{t}\right\}_{t=k}^{k+h}\right) = \sum_{l=0}^{h} \mathbf{e}_{i,n}'\mathbf{C}_{l}\bm{\Sigma}_{tr}\mathbf{Q}\mathbf{e}_{j,n}\mathbf{e}_{j,n}'\bm{\varepsilon}_{k+h-l} = \sum_{l=0}^{h} \mathbf{c}_{i,l}'(\bm{\phi})\mathbf{q}_{j}\mathbf{q}_{j}'\bm{\Sigma}_{tr}^{-1}\mathbf{u}_{k+h-l}.
\end{equation}
One example of a restriction on the historical decomposition is that the $j$th structural shock was the `most important contributor' to the change in the $i$th variable between periods $k$ and $k+h$, which requires that $|H_{i,j,k,k+h}| \geq \max_{l \neq j} |H_{i,l,k,k+h}|$. Another example is that the $j$th structural shock was the `overwhelming contributor' to the change in the $i$th variable between periods $k$ and $k+h$, which requires that $|H_{i,j,k,k+h}| \geq \sum_{l \neq j} |H_{i,l,k,k+h}|$. From Equation~(\ref{eq:historicaldecomposition}), it is clear that these restrictions are nonlinear inequality constraints that simultaneously constrain every column of $\mathbf{Q}$ and that depend on the realizations of the data in particular periods in addition to the reduced-form parameters.

Other restrictions also naturally fit into this framework. For instance, we can consider restrictions on the relative magnitudes of a particular structural shock in different periods. We refer to these restrictions as `shock-rank restrictions', since they imply a (possibly partial) ordering of the shocks. As an example, one could impose that the $i$th shock in period $k$ was the largest positive realization of this shock in the observed sample. This requires that $\varepsilon_{ik}(\bm{\phi},\mathbf{Q},\mathbf{u}_{k}) \geq \max_{t\neq k}\{\varepsilon_{it}(\bm{\phi},\mathbf{Q},\mathbf{u}_{t})\}$, which can be expressed as a system of $T-1$ linear inequality restrictions on a single column of $\mathbf{Q}$: $(\bm{\Sigma}_{tr}^{-1}(\mathbf{u}_{k}-\mathbf{u}_{t}))'\mathbf{q}_{i} \geq 0$ for $t \neq k$. Alternatively, one could impose that the $i$th shock in period $k$ was the largest-magnitude realization of that shock, or $|\varepsilon_{ik}(\bm{\phi},\mathbf{Q},\mathbf{u}_{k})| \geq \max_{t\neq k}\left\{|\varepsilon_{it}(\bm{\phi},\mathbf{Q},\mathbf{u}_{t})|\right\}$. If $\varepsilon_{ik}(\bm{\phi},\mathbf{Q},\mathbf{u}_{k}) \geq 0$, this would require that $(\bm{\Sigma}_{tr}^{-1}(\mathbf{u}_{k}-\mathbf{u}_{t}))'\mathbf{q}_{i} \geq 0$ and $(\bm{\Sigma}_{tr}^{-1}(\mathbf{u}_{k}+\mathbf{u}_{t}))'\mathbf{q}_{i} \geq 0$ for $t \neq k$, which is a system of $2(T-1)$ linear inequalities constraining $\mathbf{q}_{i}$. These restrictions could also be applied to a subset of the observations rather than the full sample (e.g., $\varepsilon_{ik}(\bm{\phi},\mathbf{Q},\mathbf{u}_{k}) > \varepsilon_{it}(\bm{\phi},\mathbf{Q},\mathbf{u}_{t})$ for some $t \in \{1,\ldots,T\}$).\footnote{Similar to the shock-rank restrictions we describe, \cite{BenZeev_2018} imposes a restriction on the timing of the maximum three-year average of a particular shock, as well as restrictions on the sign and relative magnitudes of this three-year average in specific periods. Restrictions on averages of shocks can also be implemented in the framework we consider.}

The collection of NR can be represented in the general form $N(\bm{\phi},\mathbf{Q},\mathbf{Y}^{T}) \geq \mathbf{0}_{s\times 1}$, where $s$ is the number of restrictions. As an illustration, consider the case where there is a single shock-sign restriction in period $k$, $\varepsilon_{1k}(\bm{\phi},\mathbf{Q},\mathbf{u}_{k}) \geq 0$, as well as the restriction that the first structural shock was the most important contributor to the change in the first variable in period $k$. Then,
\begin{equation}
    N(\bm{\phi},\mathbf{Q},\mathbf{Y}^{T}) =
    \begin{bmatrix}
      (\bm{\Sigma}_{tr}^{-1}\mathbf{u}_{k})'\mathbf{q}_{1} \\
      |\mathbf{e}_{1,n}'\bm{\Sigma}_{tr}\mathbf{q}_{1}\mathbf{q}_{1}'\bm{\Sigma}_{tr}^{-1}\mathbf{u}_{k}| - \max_{j \neq 1}|\mathbf{e}_{1,n}'\bm{\Sigma}_{tr}\mathbf{q}_{j}\mathbf{q}_{j}'\bm{\Sigma}_{tr}^{-1}\mathbf{u}_{k}|
    \end{bmatrix}
    \geq \mathbf{0}_{2\times 1}.
\end{equation}

Traditional sign and zero restrictions can also be applied alongside NR. We follow AR18 by explicitly allowing for sign restrictions on impulse responses and on elements of $\mathbf{A}_{0}$. We denote such sign restrictions by $S(\bm{\phi},\mathbf{Q}) \geq \mathbf{0}_{\tilde{s}\times 1}$, where $\tilde{s}$ is the number of traditional sign restrictions. It is straightforward to additionally allow for zero restrictions, including `short-run' zero restrictions (as in \cite{Sims_1980}), `long-run' zero restrictions (as in \cite{Blanchard_Quah_1989}), or restrictions arising from external instruments (as in \cite{Mertens_Ravn_2013} and \cite{Stock_Watson_2018}); for example, see GK and Giacomini et al. (2019)\nocite{Giacomini_Kitagawa_Read_2019}.

\subsection{Conditional and unconditional likelihoods}
\label{subsec:likelihoods}

When constructing the posterior  of the SVAR's parameters, AR18 use the likelihood conditional on the NR holding. Define
\begin{align*}
D_N &= D_N (\bm{\phi},\mathbf{Q},\mathbf{Y}^T)  \equiv 1\{N(\bm{\phi},\mathbf{Q},\mathbf{Y}^T) \geq \mathbf{0}_{s\times 1} \}, \\
r(\bm{\phi},\mathbf{Q}) & \equiv \Pr(D_{N}(\bm{\phi}, \mathbf{Q}, \mathbf{Y}^T) = 1|\bm{\phi}, \mathbf{Q}), \\
f(\mathbf{y}^T|\bm{\phi}) & \equiv \prod_{t=1}^{T}(2\pi)^{-\frac{n}{2}}|\bm{\Sigma}|^{-\frac{1}{2}}\exp\left(-\frac{1}{2} \left(\mathbf{y}_{t} - \mathbf{Bx}_t \right)'\bm{\Sigma}^{-1} \left(\mathbf{y}_{t} - \mathbf{Bx}_t \right) \right).
\end{align*}
The likelihood conditional on $D_N = 1$ can be written as
\begin{equation}\label{eq:conditional_likelihood}
p(\mathbf{y}^{T}| D_N = 1, \bm{\phi},\mathbf{Q}) = \frac{f(\mathbf{y}^T|\bm{\phi})}{r(\bm{\phi},\mathbf{Q})} \cdot D_N(\bm{\phi},\mathbf{Q},\mathbf{y}^T).
\end{equation}
$f(\mathbf{y}^T|\bm{\phi})$ is the joint density of the data given $\bm{\phi}$ (i.e., the likelihood function of the reduced-form VAR), which depends only on $\bm{\phi}$ and the data. The indicator function $D_N(\bm{\phi},\mathbf{Q},\mathbf{y}^T)$ is equal to one when the NR are satisfied and is equal to zero otherwise. This determines the truncation points of the likelihood. $r(\bm{\phi},\mathbf{Q})$ is the ex ante probability that the NR are satisfied. This will be a constant when there are only shock-sign or shock-rank restrictions; for example, if there are $s$ shock-sign restrictions, $r(\bm{\phi},\mathbf{Q})  = (1/2)^{s}$. In contrast, when there are restrictions on the historical decomposition, this probability will depend on $\bm{\phi}$ and $\mathbf{Q}$.

Consider the case where $\bm{\phi}$ is known, which will be the case asymptotically because $\bm{\phi}$ is point-identified. When $r(\bm{\phi},\mathbf{Q})$ depends on $\mathbf{Q}$, the conditional likelihood will be maximized at the value of $\mathbf{Q}$ that minimizes $r(\bm{\phi},\mathbf{Q})$ (within the set of values of $\mathbf{Q}$ that satisfy the restrictions). The posterior based on this likelihood will therefore place higher posterior probability on values of $\mathbf{Q}$ that result in a lower ex ante probability that the restrictions are satisfied. As discussed in Section~\ref{subsec:histdecomp}, this is an artefact of conditioning on a non-ancillary event, which represents a loss of information about the parameters.

We therefore advocate constructing the likelihood without conditioning on the NR holding. The unconditional likelihood (the joint distribution of the data and $D_N$) can be expressed as
\begin{align}\
p(\mathbf{y}^{T},D_N=d  | \bm{\phi},\mathbf{Q}) &= \left[ f(\mathbf{y}^T|\bm{\phi}) D_N(\bm{\phi},\mathbf{Q},\mathbf{y}^T)\right]^{d} \cdot \left[ f(\mathbf{y}^T|\bm{\phi}) \left( 1-D_N(\bm{\phi},\mathbf{Q},\mathbf{y}^T) \right)\right]^{1-d} \notag \\
& = f(\mathbf{y}^T|\bm{\phi}) \cdot \left[ D_N(\bm{\phi},\mathbf{Q},\mathbf{y}^T)\right]^{d} \cdot \left[ 1-D_N(\bm{\phi},\mathbf{Q},\mathbf{y}^T) \right]^{1-d}. \label{eq:joint_likelihood}
\end{align}
For any value of $\bm{\phi}$ such that $\mathbf{y}^{T}$ is compatible with the NR, there will be a set of values of $\mathbf{Q}$ that satisfy the restrictions, which depend on the data, but the value of the unconditional likelihood will be the same for all values of $\mathbf{Q}$ within this set. The conditional posterior of $\mathbf{Q}|\bm{\phi},\mathbf{y}^{T}$ will therefore be proportional to the conditional prior for $\mathbf{Q}|\bm{\phi}$ in these regions. Given a fixed number of NR, the likelihood will possess flat regions even with a time-series of infinite length, so posterior inference may be sensitive to the choice of conditional prior for $\mathbf{Q}$, even asymptotically (which is also the case for the conditional likelihood when the restrictions are ancillary). This motivates considering Bayesian inferential procedures that are robust to the choice of unrevisable conditional prior for $\mathbf{Q}$, which we explore in Section~\ref{subsec:robustbayes}.

\subsection{Discussion}
\label{subsec:discussion}

In this section, we briefly discuss the distributional assumptions for the structural shocks and the mechanism that generates the NR.

\subsubsection{Distributional assumptions}

Practitioners may be concerned about the robustness of inference with respect to deviations from the assumption of standard normal shocks. For instance, one could worry that the periods in which the NR are imposed are `unusual' in the sense that the structural shocks in these periods were drawn from a distribution with, say, inflated variance or fat tails. The unconditional likelihood depends on the normality assumption only through $f(\mathbf{y}^T|\bm{\phi})$. By omitting terms in $f(\mathbf{y}^T|\bm{\phi})$ corresponding to the periods in which the NR are imposed, one can conduct inference that is robust to the distributional assumption about the shocks in these particular periods. To illustrate, consider the case where NR are imposed in period $k$ only and assume the likelihood function for $\mathbf{y}^{T}$ takes the form
\begin{equation}
  \tilde{f}(\mathbf{y}^T|\bm{\phi}) = v(\left\{\mathbf{y}_{t} - \mathbf{Bx}_t\right\}_{t\neq k} | \bm{\phi}) w(\mathbf{y}_{k} - \mathbf{Bx}_k),
\end{equation}
where
\begin{equation}
  v(\left\{\mathbf{y}_{t}- \mathbf{Bx}_t\right\}_{t\neq k} | \bm{\phi}) = \prod_{t\neq k}(2\pi)^{-\frac{n}{2}}|\bm{\Sigma}|^{-\frac{1}{2}}\exp\left(-\frac{1}{2} \left(\mathbf{y}_{t} - \mathbf{Bx}_t \right)'\bm{\Sigma}^{-1} \left(\mathbf{y}_{t} - \mathbf{Bx}_t \right) \right)
\end{equation}
and $w(\mathbf{y}_{k} - \mathbf{Bx}_k)$ is an unknown, potentially non-normal, density. Replacing $f(\mathbf{y}^T|\bm{\phi})$ in Equation~(\ref{eq:joint_likelihood}) with $v(\left\{\mathbf{y}_{t} - \mathbf{Bx}_t\right\}_{t\neq k} | \bm{\phi})$ yields an `unconditional partial likelihood' that does not depend on the distribution of $\bm{\varepsilon}_{k}$, but that is still truncated by the NR. This would potentially result in a loss of information relative to a likelihood that correctly specifies the distribution of the shocks in period $k$. However, when NR are imposed in only a few periods, this loss of information is likely to be small. In contrast, the conditional likelihood approach cannot leave fully unspecified the distribution of the restricted structural shocks, because computing $r(\bm{\phi},\mathbf{Q})$ requires specifying this distribution.

Concerns about heteroscedasticity or non-normality may also be alleviated by recognizing that the distributional assumption will become irrelevant asymptotically. The set of values of $\mathbf{Q}$ with non-zero unconditional likelihood depends only on $\bm{\phi}$, which summarizes the second moments of the data, and the realization of the data in the periods in which the NR are imposed. Under regularity assumptions, the likelihood (and thus the posterior) of $\bm{\phi}$ will converge to a point at the true value of $\bm{\phi}$ asymptotically regardless of whether the true data-generating process is a VAR with homoscedastic normal shocks.\footnote{See \cite{Plagborg-Moller_2019} for a discussion of this point in the context of a structural VMA model.} The set of values of $\mathbf{Q}$ with non-zero likelihood will therefore converge asymptotically to the same set regardless of whether the distributional assumption is correct.

\subsubsection{Mechanism generating NR}

Note that we do not explicitly model the mechanism responsible for revealing the information underlying the NR (i.e., whether $D_{N} = 1$ or $D_{N} = 0$) or the mechanism determining the periods in which this information is revealed (e.g., the identity of $k$ in examples above), which is consistent with the papers that impose these restrictions. If the revelation of this information depends on the data, the likelihood will be misspecified. The exact implications of this misspecification for estimation or inference will depend on assumptions about the mechanism revealing the narrative information. Exploring the consequences of such misspecification may be an interesting area for further work. In the bivariate example of Section~\ref{sec:bivariate}, if the identity of $k$ is randomly determined independently of $\bm{\varepsilon}_{1},\ldots,\bm{\varepsilon}_{T}$, we can interpret the current analysis conditional on $k$.

\section{Identification under NR}
\label{sec:identification}

This section formally analyzes identification in the SVAR under NR. Section~\ref{subsec:frequentistidentification} considers whether NR are point- or set-identifying in a frequentist sense. Section~\ref{subsec:frequentistproperties} introduces the notion of a `conditional identified set', which extends the standard notion of an identified set to the setting where the mapping from reduced-form to structural parameters depends on the realization of the data. This provides an interpretation of the mapping induced by the NR. Additionally, we make use of this object when showing the frequentist validity of our robust Bayesian procedure in Section~\ref{sec:fewNR}.

\subsection{Point-identification under NR}
\label{subsec:frequentistidentification}

Denoting the true parameter value by $(\bm{\phi}_0, \textbf{Q}_0)$, point-identification for the parametric model (\ref{eq:joint_likelihood}) requires that there is no other parameter value $(\bm{\phi},\mathbf{Q}) \neq (\bm{\phi}_0,\mathbf{Q}_0)$ that is observationally equivalent to $(\bm{\phi}_0,\mathbf{Q}_0)$.\footnote{$(\bm{\phi},\mathbf{Q}) \neq (\bm{\phi}_0,\mathbf{Q}_0)$ is observationally equivalent to $(\bm{\phi}_0,\mathbf{Q}_0)$ if $p(\mathbf{Y}^{T},D_N=d  | \bm{\phi},\mathbf{Q}) = p(\mathbf{Y}^{T},D_N=d  | \bm{\phi}_0,\mathbf{Q}_0)$ holds for all $\mathbf{Y}^T$ and $d \in \{0,1 \}$.}

To assess the existence or non-existence of observationally equivalent parameter points, we analyze a statistical distance between $p(\mathbf{y}^{T},D_N=d  | \bm{\phi},\mathbf{Q})$ and $p(\mathbf{y}^{T},D_N=d  | \bm{\phi}_0,\mathbf{Q}_0)$ that metrizes observation equivalence. Specifically, in the current setting where the support of the distribution of observables can depend on the parameters, it is convenient to work with the Hellinger distance:
\begin{align}
HD(\bm{\phi},\mathbf{Q}) & \equiv \sum_{d=0,1}\int_{\mathbf{Y}} \left(p^{1/2}(\mathbf{y}^{T},D_N=d  | \bm{\phi},\mathbf{Q}) - p^{1/2}(\mathbf{y}^{T},D_N=d  | \bm{\phi}_0,\mathbf{Q}_0) \right)^2 d \mathbf{y}^T \notag \\
&= 2 \left( 1 - \mathcal{H}(\bm{\phi},\mathbf{Q})\right), \mspace{10mu} \text{where} \notag \\
\mathcal{H}(\bm{\phi},\mathbf{Q}) & \equiv \sum_{d=0,1} \int_{\mathbf{Y}} p^{1/2}(\mathbf{y}^{T},D_N=d  | \bm{\phi},\mathbf{Q}) \cdot p^{1/2}(\mathbf{y}^{T},D_N=d  | \bm{\phi}_0,\mathbf{Q}_0) d \mathbf{y}^T, \label{eq:Hellinger}
\end{align}
and $\mathbf{Y}$ is the sample space for $\mathbf{Y}^{T}$. As is known in the literature on minimum distance estimation (see, for example, Basu, Shioya and Park (2011)\nocite{Basu_Shioya_Park_2011}), $(\bm{\phi},\mathbf{Q})$ and $(\bm{\phi}_0,\mathbf{Q}_0)$ are observationally equivalent if and only if $HD(\bm{\phi},\mathbf{Q}) = 0$ or, equivalently, $\mathcal{H}(\bm{\phi},\mathbf{Q}) = 1$.

We similarly define the Hellinger distance for the conditional likelihood as
\begin{align}
HD_{c}(\bm{\phi},\mathbf{Q}) & \equiv 2 \left( 1 - \mathcal{H}_c(\bm{\phi},\mathbf{Q})\right), \mspace{10mu} \text{where} \notag \\
\mathcal{H}_c(\bm{\phi},\mathbf{Q}) & \equiv \int_{\mathbf{Y}} p^{1/2}(\mathbf{y}^{T}  | D_N=1,\bm{\phi},\mathbf{Q}) \cdot p^{1/2}(\mathbf{y}^{T}  | D_N=1, \bm{\phi}_0,\mathbf{Q}_0) d \mathbf{y}^T. \label{eq:Hellinger_conditional}
\end{align}

The next proposition analyzes the conditions for $\mathcal{H}(\bm{\phi}, \mathbf{Q})=1$ and $\mathcal{H}_c(\bm{\phi}, \mathbf{Q})=1$, and shows that observational equivalence of $(\bm{\phi},\mathbf{Q})$ and $(\bm{\phi}_0,\mathbf{Q}_0)$ boils down to geometric equivalence of the set of reduced-form VAR innovations satisfying the NR.

\begin{proposition} \label{prop:point_identification}
Let $(\bm{\phi}_0, \mathbf{Q}_0)$ be the true parameter value and let $\mathbf{U} \equiv \mathbf{U}(\mathbf{y}^{T};\bm{\phi}) = (\mathbf{u}_{1}',\ldots,\mathbf{u}_{T}')'$ collect the reduced-form VAR innovations. Define
\small
\begin{equation}
\mathcal{Q}^{\ast} \equiv \left\{ \begin{matrix} \mathbf{Q} \in \mathcal{O}(n) : \{ \mathbf{U}:N(\bm{\phi}, \mathbf{Q}, \mathbf{Y}^{T}) \geq \mathbf{0}_{s \times 1} \} = \{ \mathbf{U}:N(\bm{\phi}_0, \mathbf{Q}_0, \mathbf{Y}^{T}) \geq \mathbf{0}_{s \times 1} \} \\
 \text{up to $f(\mathbf{Y}^T| \bm{\phi}_0)$-null set}, \ \mathrm{diag}( \mathbf{Q}'\bm{\Sigma}_{tr}^{-1}) \geq \mathbf{0}_{n \times 1} \end{matrix} \right\}. \notag
\end{equation}
\normalsize
The unconditional likelihood model (\ref{eq:joint_likelihood}) and the conditional likelihood model (\ref{eq:conditional_likelihood}) are globally identified (i.e., there are no observationally equivalent parameter points to $(\bm{\phi}_0, \mathbf{Q}_0)$) if and only if $\mathcal{Q}^{\ast}$ is a singleton. If the parameter of interest is an impulse response to the $j$th structural shock, $\eta_{i,j,h}(\bm{\phi}, \mathbf{Q})$, as defined in (\ref{eq:ir}), then $\eta_{i,j,h}(\bm{\phi}, \mathbf{Q})$ is point-identified if the projection of $\mathcal{Q}^{\ast}$ onto its $j$th column vector is a singleton.
\end{proposition}
\begin{proof}
See Appendix B.
\end{proof}
This proposition provides a necessary and sufficient condition for global identification of SVARs by NR. As shown in the proof in Appendix B, $\mathcal{Q}^{\ast}$ defined in this proposition corresponds to the observationally equivalent $\mathbf{Q}$ matrices given $\bm{\phi} = \bm{\phi}_0$, but, importantly, it does not correspond to any flat region of the observed likelihood (the conditional identified set in Definition \ref{def:identifiedset} below).

To illustrate this point, consider the simple bivariate example of Section 2 with the NR (\ref{eq:nsr}), where $\mathbf{y}_t$ itself is the reduced-form error, so $\mathbf{U}$ in Proposition \ref{prop:point_identification} can be set to $\mathbf{y}_k$. Given $\bm{\phi}$, the set of $\mathbf{y}_k \in \mathbb{R}^2$ satisfying the NR is the half-space given by
\begin{equation}
\left\{ \mathbf{y}_k \in \mathbb{R}^2 : (\sigma_{11} \sigma_{22})^{-1} \begin{pmatrix} \sigma_{22} \cos \theta - \sigma_{21} \sin \theta, & \sigma_{11} \sin\theta \end{pmatrix} \mathbf{y}_k \geq 0  \right\}. \label{eq:yk_half_space}
\end{equation}
The condition for point-identification shown in Proposition~\ref{prop:point_identification} is satisfied if no $\theta' \neq \theta$ can generate the half-space of $\mathbf{y}_k$ identical to (\ref{eq:yk_half_space}). Such $\theta'$ cannot exist, since a half-space passing through the origin $(a_1,a_2) \mathbf{y}_k \geq 0$ can be indexed uniquely by the  slope $a_1/a_2$ and (\ref{eq:yk_half_space}) implies the slope $\sigma_{11}^{-1}(\sigma_{22} (\tan \theta)^{-1} - \sigma_{21})$ is a bijective map of $\theta$ on a constrained domain due to the sign normalization. Figure~\ref{fig:hellinger} plots the Hellinger distances in this bivariate example under the shock-sign restriction (\ref{eq:nsr}) and the historical decomposition restriction. For both the conditional and unconditional likelihood, the Hellinger distances are minimized uniquely at the true $\theta$, which is consistent with our point-identification claim for $\theta$.\footnote{Under the restriction on the historical decomposition, a notable difference between the conditional and unconditional likelihood cases is the slope of the Hellinger distance around the minimum. The Hellinger distance of the unconditional likelihood yields a steeper slope than the conditional likelihood. This indicates the loss of information for $\theta$ in the conditional likelihood due to conditioning on the non-ancillary event.}

\begin{figure}[h]
    \center
    \caption{Hellinger Distance} \label{fig:hellinger}
    \begin{tabular}{ccc}
        \includegraphics[scale=0.5]{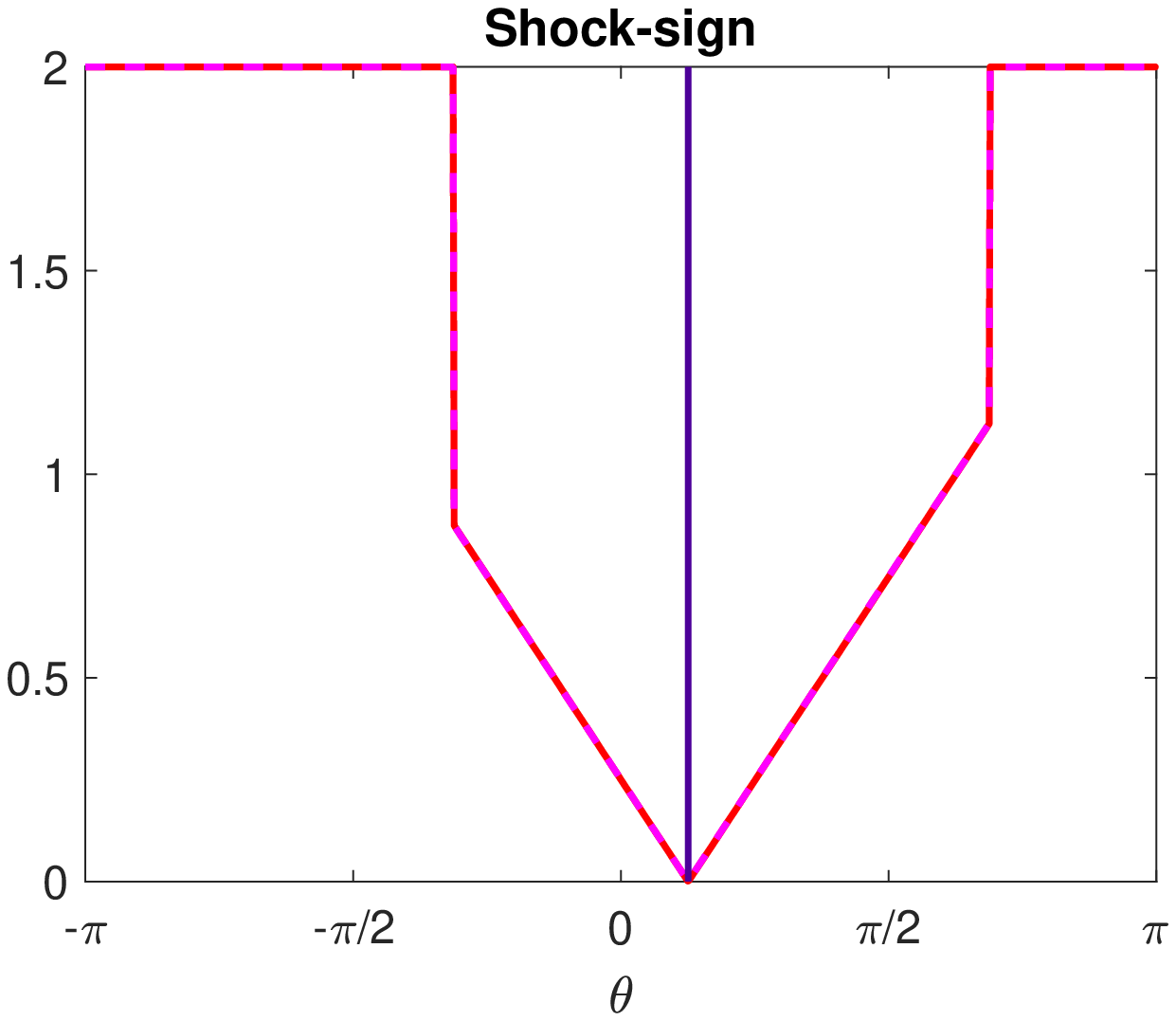} & \phantom{a} & \includegraphics[scale=0.5]{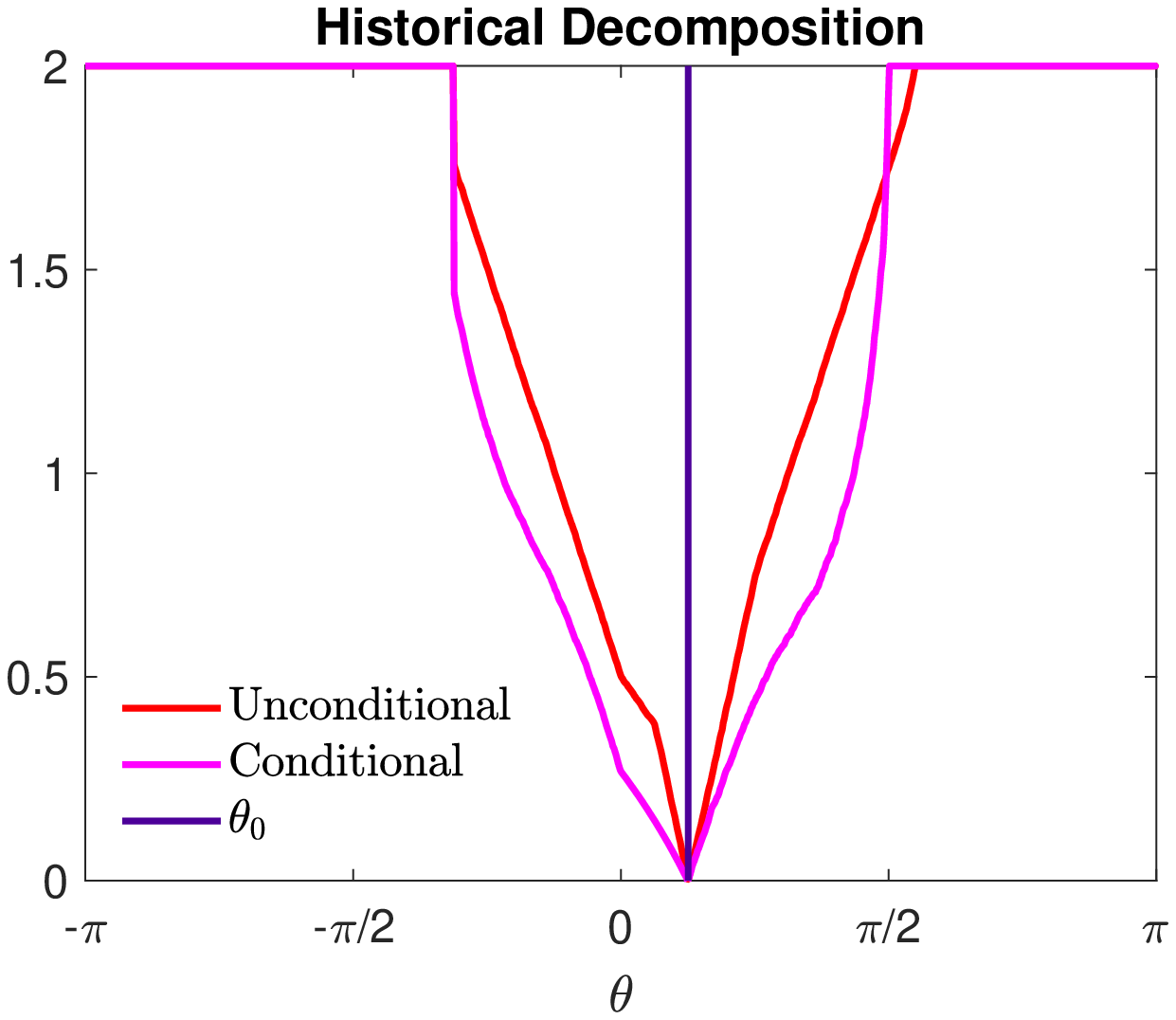} \\
    \end{tabular}
    \footnotesize \parbox[t]{0.65 in}{Notes:}\parbox[t]{5 in}{$T = 3$ and $\bm{\phi}$ is known; Hellinger distances are approximated using Monte Carlo.} \\
\end{figure}

Proposition~\ref{prop:point_identification} also provides conditions under which $(\bm{\phi},\mathbf{Q})$ is not globally identified, but a particular impulse response is. To give an example of this, consider an SVAR with $n > 2$ and with a shock-sign restriction on the first shock in period $k$. Given $\bm{\phi}$, the set of $\mathbf{u}_{k} \in \mathbb{R}^{n}$ satisfying the NR is a half-space defined by $\mathbf{q}_{1}'\bm{\Sigma}_{tr}^{-1}\mathbf{u}_{k} \geq 0$. The set of values of $\mathbf{u}_{k}$ satisfying this inequality is indexed uniquely by $\mathbf{q}_{1}$ given $\bm{\Sigma}_{tr}$ at its true value, so there are no values of $\mathbf{Q}$ that are observationally equivalent to $\mathbf{Q}_{0}$ with $\mathbf{q}_{1} \neq \mathbf{Q}_{0}\mathbf{e}_{1,n}$. Any value for the remaining $n-1$ columns of $\mathbf{Q}$ such that they are orthogonal to $\mathbf{Q}_{0}\mathbf{e}_{1,n}$ will generate the same half-space for $\mathbf{u}_{k}$, so $\mathcal{Q}^{\ast}$ is not a singleton and the SVAR is not globally identified. However, the projection of $\mathcal{Q}^{\ast}$ onto its first column is a singleton, so $\eta_{i,1,h}(\bm{\phi}, \mathbf{Q})$ is globally identified.

Although a single NR can deliver global identification in the frequentist sense, the practical implication of this theoretical claim is not obvious. The observed unconditional likelihood is almost always flat at the maximum, so we cannot obtain a unique maximum likelihood estimator for the structural parameter. As a result, the standard asymptotic approximation of the sampling distribution of the maximum likelihood estimator is not applicable. The SVAR model with NR possesses features of set-identified models from the Bayesian standpoint (i.e., flat regions of the likelihood). However, strictly speaking, it can be classified as a globally identified model in the frequentist sense when the condition of Proposition \ref{prop:point_identification} holds.

\subsection{Conditional identified set}
\label{subsec:frequentistproperties}

It is well-known that traditional sign restrictions deliver set-identification of $\mathbf{Q}$ (or, equivalently, the structural parameters). Given the reduced-form parameter $\bm{\phi}$ -- which is point-identified --  there are multiple observationally equivalent values of $\mathbf{Q}$, in the sense that there exists $\mathbf{Q}$ and $\tilde{\mathbf{Q}} \neq \mathbf{Q}$ such that $p(\mathbf{y}^{T}|\bm{\phi},\mathbf{Q}) = p(\mathbf{y}^{T}|\bm{\phi},\tilde{\mathbf{Q}})$ for every $\mathbf{y}^{T}$ in the sample space. The identified set for $\mathbf{Q}$ given $\bm{\phi}$ contains all such observationally equivalent parameter points, and is defined as
\begin{equation}\label{eq:standardidentifiedset}
  \mathcal{Q}(\bm{\phi}|S) = \left\{\mathbf{Q} \in \mathcal{O}(n): S(\bm{\phi},\mathbf{Q}) \geq \mathbf{0}_{\tilde{s}\times 1}, \mathrm{diag}(\mathbf{Q}'\bm{\Sigma}_{tr}^{-1}) \geq \mathbf{0}_{n\times 1}\right\}.
\end{equation}
The identified set is a set-valued map only of $\bm{\phi}$, which carries all the information about $\mathbf{Q}$ contained in the data.

The complication in applying this definition of the identified set in SVARs when there are NR is that the reduced-form VAR parameters no longer represent all information about $\mathbf{Q}$ contained in the data; by truncating the likelihood, the realizations of the data entering the NR contain additional information about $\mathbf{Q}$. To address this, we introduce a refinement of the definition of an identified set.

\bigskip

\begin{definition}\label{def:identifiedset}
Let $N \equiv N(\bm{\phi},\mathbf{Q},\mathbf{y}^{T}) \geq \mathbf{0}_{s\times 1}$ represent a set of NR in terms of the parameters and the data.

(i) The \textbf{conditional identified set for $\mathbf{Q}$ under NR} is
\begin{equation}
    \mathcal{Q}(\bm{\phi}|\mathbf{y}^{T},N) = \{\mathbf{Q} \in \mathcal{O}(n): N(\bm{\phi},\mathbf{Q},\mathbf{y}^{T}) \geq \mathbf{0}_{s\times 1}\}.
\end{equation}
The conditional identified set for the impulse response $\eta = \eta_{i,j,h}(\bm{\phi},\mathbf{Q})$ under NR is defined by projecting $\mathcal{Q}(\bm{\phi}|\mathbf{y}^{T},N)$ via $\eta_{i,j,h}(\bm{\phi},\mathbf{Q})$:
\begin{equation}
CIS_{\eta}(\bm{\phi}|\mathbf{y}^{T},N) = \{\eta_{i,j,h}(\bm{\phi},\mathbf{Q}) : \mathbf{Q} \in \mathcal{Q}(\bm{\phi}|\mathbf{y}^{T},N) \}.
\end{equation}
(ii) Let $\mathbf{s}: \mathbf{Y} \to \mathbb{R}^S$ be a statistic. We call $\mathbf{s}(\mathbf{Y}^T)$ a \textbf{sufficient statistic for the conditional identified set} $\mathcal{Q}(\bm{\phi}|\mathbf{y}^{T},N)$ if the conditional identified set for $\mathbf{Q}$ depends on the sample $\mathbf{y}^T$ through $\mathbf{s}(\mathbf{y}^T)$; i.e., there exists $\tilde{\mathcal{Q}}(\bm{\phi}|\cdot,N)$ such that
\begin{equation}
\mathcal{Q}(\bm{\phi}|\mathbf{y}^{T},N) = \tilde{\mathcal{Q}}(\bm{\phi}|\mathbf{s}(\mathbf{y}^{T}),N)
\end{equation}
holds for all $\bm{\phi} \in \bm{\Phi}$ and $\mathbf{y}^T \in \mathbf{Y}$.
\end{definition}

\bigskip

Unlike the standard identified set $\mathcal{Q}(\bm{\phi}|S)$, the conditional identified set $\mathcal{Q}(\bm{\phi}|\mathbf{y}^{T},N)$ depends on the sample $\mathbf{y}^T$ because of the aforementioned data-dependent support of the likelihood. In terms of the observed likelihood, however, they share the property that the likelihood is flat on the (conditional) identified set. Hence, given the sample $\mathbf{y}^{T}$ and the reduced-form parameters $\bm{\phi}$, all values of $\mathbf{Q}$ in $\mathcal{Q}(\bm{\phi}|\mathbf{y}^{T},N)$ fit the data equally well and, in this particular sense, they are observationally equivalent.

When the NR concern shocks in only a subset of the time periods in the data, the conditional identified set under these NR depends on the sample only through a few observations entering the NR. The sufficient statistics $\textbf{s}(\mathbf{y}^T)$ defined in Definition \ref{def:identifiedset}(ii) represent such observations. For instance, in the toy example of Section~\ref{subsec:sign}, the conditional identified set depends only on the observations in period $k$, so $\mathbf{s}(\mathbf{y}^T) =\mathbf{y}_k$. If we extend the example of Section~\ref{subsec:sign} to the SVAR($p$), the shock-sign restriction in Equation~(\ref{eq:nsr}) can be expressed as
\begin{equation}
\varepsilon_{1k} = \mathbf{e}_{1,2}^{\prime} \mathbf{A}_0 \mathbf{u}_k=\mathbf{e}_{1,2}^{\prime} \mathbf{Q}^{\prime} \bm{\Sigma}_{tr}^{-1}(\mathbf{y}_k - \mathbf{B} \mathbf{x}_k) \geq 0.
\end{equation}
Hence, the conditional identified set $\mathcal{Q}(\bm{\phi}|\mathbf{y}^{T},N)$ depends on the data only through $(\mathbf{y}_k^{\prime}, \mathbf{x}_k^{\prime})'= (\mathbf{y}_k^{\prime},\mathbf{y}_{k-1}^{\prime}, \cdots, \mathbf{y}_{k-p}^{\prime} )'$, so we can set $\mathbf{s}(\mathbf{y}^T)=(\mathbf{y}_k^{\prime},\mathbf{y}_{k-1}^{\prime}, \cdots, \mathbf{y}_{k-p}^{\prime} )'$.

If the conditional distribution of $\mathbf{Y}^T$ given $\mathbf{s}(\mathbf{Y}^T) = \mathbf{s}(\mathbf{y}^T)$ is nondegenerate, we can consider a frequentist experiment (repeated sampling of $\mathbf{Y}^T$) conditional on the sufficient statistics set to the observed value. In this conditional experiment, we can view the conditional identified set $\mathcal{Q}(\bm{\phi}|\mathbf{y}^{T},N)$ as the standard identified set in set-identified models, since it no longer depends on the data in the conditional experiment where $\mathbf{s}(\mathbf{y}^{T})$ is fixed. This is the reason that we refer to $\mathcal{Q}(\bm{\phi}|\mathbf{y}^{T},N)$ as the conditional identified set. In Section~\ref{sec:fewNR} below, we show the frequentist validity of the robust-Bayes credible region by establishing conditional coverage of the conditional identified set for an impulse response.

\section{Posterior inference under NR}
\label{sec:posteriorinferenceunderNR}

This section presents approaches to conducting posterior inference in SVARs under NR. Section~\ref{subsec:bayesianinference} discusses how to modify the standard Bayesian approach in AR18 to use the unconditional likelihood rather than the conditional likelihood. Section~\ref{subsec:robustbayes} explains how to conduct robust Bayesian inference under NR, which further addresses the issue of posterior sensitivity due to the flat unconditional likelihood. Section~\ref{subsec:numericalimplementation} describes how to numerically implement the robust Bayesian procedure.

\subsection{Standard Bayesian inference}
\label{subsec:bayesianinference}

AR18 propose an algorithm for drawing from the uniform-normal-inverse-Wishart posterior of $(\bm{\phi},\mathbf{Q})$ given a set of traditional sign restrictions and NR. This is the posterior induced by a normal-inverse-Wishart prior over $\bm{\phi}$ and an unconditionally uniform prior over $\mathbf{Q}$. The algorithm proceeds by drawing $\bm{\phi}$ from a normal-inverse-Wishart distribution and $\mathbf{Q}$ from a uniform distribution over $\mathcal{O}(n)$, and checking whether the restrictions are satisfied. If the restrictions are not satisfied, the joint draw is discarded and another draw is made. If the restrictions are satisfied, the ex ante probability that the NR are satisfied at the drawn parameter values is approximated via Monte Carlo simulation. Once the desired number of draws are obtained satisfying the restrictions, the draws are resampled with replacement using as importance weights the inverse of the probability that the NR are satisfied.\footnote{Based on the results in \cite{Arias_Rubio-Ramirez_Waggoner_2018}, AR18 argue that their algorithm draws from a normal-generalized-normal posterior over the SVAR's structural parameters $(\mathbf{A}_{0},\mathbf{A}_{+})$ induced by a conjugate normal-generalized-normal prior, conditional on the restrictions.}

This algorithm essentially draws from the posterior under the unconditional likelihood and then uses importance sampling to transform these draws into draws from the posterior given the conditional likelihood. To draw from the uniform-normal-inverse-Wishart posterior using the unconditional likelihood to construct the posterior, one therefore simply needs to omit the importance-sampling step from this algorithm. Approximating the probability used to construct the importance weights requires Monte Carlo integration, which can be computationally expensive, particularly when the NR constrain the structural shocks in multiple periods. Omitting the importance-sampling step can therefore ease the computational burden of drawing from the posterior. However, as discussed above, standard Bayesian inference under the unconditional likelihood may be sensitive to the choice of conditional prior for $\mathbf{Q}|\bm{\phi}$, because the likelihood possesses flat regions.

By rejecting draws that do not satisfy the restrictions, the algorithm described above places more weight on draws of $\bm{\phi}$ that are less likely to satisfy the restrictions under the uniform distribution over $\mathcal{O}(n)$. As discussed in \cite{Uhlig_2017}, one may instead prefer to use a prior  that is \textit{conditionally} uniform over $\mathbf{Q}|\bm{\phi}$. To draw from the posterior of $(\bm{\phi},\mathbf{Q})$ under the unconditional likelihood given an arbitrary prior over $\bm{\phi}$ and a conditionally uniform prior over $\mathbf{Q}|\bm{\phi}$, one can repeat Step~2 of Algorithm~1 in Section~\ref{subsec:numericalimplementation}.

\subsection{Robust Bayesian inference}
\label{subsec:robustbayes}

This section explains how to conduct robust Bayesian inference about a scalar-valued function of the structural parameters under NR and traditional sign restrictions. The approach can be viewed as performing global sensitivity analysis to assess whether posterior conclusions are robust to the choice of prior on the flat regions of the likelihood. We assume that the object of interest is a particular impulse response $\eta$, although the discussion in this section also applies to any other scalar-valued function of the structural parameters, such as the forecast error variance decomposition or the historical decomposition.

Let $\pi_{\bm{\phi}}$ be a prior over the reduced-form parameter $\bm{\phi} \in \bm{\Phi}$, where $\bm{\Phi}$ is the space of reduced-form parameters such that $\mathcal{Q}(\bm{\phi}|S)$ is non-empty. A joint prior for $(\bm{\phi},\mathbf{Q}) \in \bm{\Phi}\times \mathcal{O}(n)$ can be written as $\pi_{\bm{\phi},\mathbf{Q}} = \pi_{\mathbf{Q}|\bm{\phi}}\pi_{\bm{\phi}}$, where $\pi_{\mathbf{Q}|\bm{\phi}}$ is supported only on $\mathcal{Q}(\bm{\phi}|S)$. When there are only traditional identifying restrictions, $\pi_{\mathbf{Q}|\bm{\phi}}$ is not updated by the data, because the likelihood function is not a function of $\mathbf{Q}$. Posterior inference may therefore be sensitive to the choice of conditional prior, even asymptotically. As discussed above, a similar issue arises under NR. The difference under NR is that $\pi_{\mathbf{Q}|\bm{\phi}}$ is updated by the data through the truncation points of the unconditional likelihood. However, at each value of $\bm{\phi}$, the unconditional likelihood is flat over the set of values of $\mathbf{Q}$ satisfying the NR. Consequently, the conditional posterior for $\mathbf{Q}|\bm{\phi},\mathbf{Y}^{T}$ is proportional to the conditional prior for $\mathbf{Q}|\bm{\phi}$ at each $\bm{\phi}$ whenever the conditional identified set for $\mathbf{Q}$ given $(\bm{\phi}, \mathbf{Y}^T)$ is nonempty.

Rather than specifying a single prior for $\mathbf{Q}|\bm{\phi}$, the robust Bayesian approach of GK considers the class of all priors for $\mathbf{Q}|\bm{\phi}$ that are consistent with the traditional identifying restrictions:
\begin{equation}\label{eq:classofpriors}
  \Pi_{\mathbf{Q}|\bm{\phi}} = \left\{\pi_{\mathbf{Q}|\bm{\phi}}: \pi_{\mathbf{Q}|\bm{\phi}}(\mathcal{Q}(\bm{\phi}|S))=1\right\}.
\end{equation}
Notice that we cannot impose the NR using a particular conditional prior on $\mathbf{Q}|\bm{\phi}$ due to the data-dependent mapping from $\bm{\phi}$ to $\mathbf{Q}$. However, by considering all possible conditional priors for $\mathbf{Q}|\bm{\phi}$ that are consistent with the traditional identifying restrictions, we trace out all possible conditional posteriors for $\mathbf{Q}|\bm{\phi},\mathbf{Y}^{T}$ that are consistent with the traditional identifying restrictions and the NR. This is because the NR truncate the unconditional likelihood function and the traditional identifying restrictions truncate the prior for $\mathbf{Q}|\bm{\phi}$, so the posterior for $\mathbf{Q}|\bm{\phi},\mathbf{Y}^{T}$ is supported only on the values of $\mathbf{Q}$ that satisfy both sets of restrictions.

Given a particular prior for $(\bm{\phi},\mathbf{Q})$ and using the unconditional likelihood, the posterior is
\begin{align}
    \pi_{\bm{\phi},\mathbf{Q}|\mathbf{Y}^{T},D_{N} = 1} &\propto p(\mathbf{Y}^{T},D_{N} = 1|\bm{\phi},\mathbf{Q})\pi_{\mathbf{Q}|\bm{\phi}}\pi_{\bm{\phi}} \notag \\
    &\propto f(\mathbf{Y}^{T}|\bm{\phi})D_N(\bm{\phi},\mathbf{Q},\mathbf{Y}^T)\pi_{\bm{\phi}}\pi_{\mathbf{Q}|\bm{\phi}} \notag \\
    &\propto \pi_{\bm{\phi}|\mathbf{Y}^{T}}\pi_{\mathbf{Q}|\bm{\phi}}D_N(\bm{\phi},\mathbf{Q},\mathbf{Y}^T).
\end{align}
The final expression for the posterior makes it clear that any prior for $\mathbf{Q}|\bm{\phi}$ that is consistent with the traditional identifying restrictions is in effect further truncated by the NR (through the likelihood) once the data are realized. Generating this posterior using every prior within the class of priors for $\mathbf{Q}|\bm{\phi}$ generates a class of posteriors for $(\bm{\phi},\mathbf{Q})$:
\begin{equation}\label{eq:classofposteriors}
  \Pi_{\bm{\phi},\mathbf{Q}|\mathbf{Y}^{T},D_{N} = 1} = \left\{\pi_{\bm{\phi},\mathbf{Q}|\mathbf{Y}^{T},D_{N} = 1} = \pi_{\bm{\phi}|\mathbf{Y}^{T}}\pi_{\mathbf{Q}|\bm{\phi}}D_N(\bm{\phi},\mathbf{Q},\mathbf{Y}^T): \pi_{\mathbf{Q}|\bm{\phi}} \in \Pi_{\mathbf{Q}|\bm{\phi}} \right\}.
\end{equation}
Marginalizing each posterior in this class of posteriors induces a class of posteriors for $\eta$, $\Pi_{\eta|\mathbf{Y}^{T},D_{N} = 1}$. Each prior within the class of priors $\Pi_{\mathbf{Q}|\bm{\phi}}$ therefore induces a posterior for $\eta$. Associated with each of these posteriors are quantities such as the posterior mean, median and other quantiles. For example, as we consider each possible prior within $\Pi_{\mathbf{Q}|\bm{\phi}}$, we can trace out the set of all possible posterior means for $\eta$. This will always be an interval, so we can summarize this `set of posterior means' by its endpoints:
\begin{equation}\label{eq:setofposteriormeans}
  \left[\int_{\bm{\Phi}}l(\bm{\phi},\mathbf{Y}^{T})d\pi_{\bm{\phi}|\mathbf{Y}^{T}},\int_{\bm{\Phi}}u(\bm{\phi},\mathbf{Y}^{T})d\pi_{\bm{\phi}|\mathbf{Y}^{T}}\right],
\end{equation}
where $l(\bm{\phi},\mathbf{Y}^{T}) = \inf \{\eta(\bm{\phi},\mathbf{Q}):\mathbf{Q} \in \mathcal{Q}(\bm{\phi}|\mathbf{Y}^{T},N,S)\}$, $u(\bm{\phi},\mathbf{Y}^{T}) = \sup \{\eta(\bm{\phi},\mathbf{Q}):\mathbf{Q} \in \mathcal{Q}(\bm{\phi}|\mathbf{Y}^{T},N,S)\}$ and
\begin{equation}
  \mathcal{Q}(\bm{\phi}|\mathbf{Y}^{T},N,S) = \left\{\mathcal{Q}(\bm{\phi}|S)  \cap \mathcal{Q}(\bm{\phi}|\mathbf{Y}^{T},N)\right\}
\end{equation}
is the set of values of $\mathbf{Q}$ that are consistent with the traditional identifying restrictions and the NR. In contrast, in GK the set of posterior means is obtained by finding the infimum and supremum of $\eta(\bm{\phi},\mathbf{Q})$ over $\mathcal{Q}(\bm{\phi}|S)$ and averaging these over $\pi_{\bm{\phi}|\mathbf{Y}^{T}}$. The important difference from GK is that the current set of posterior means depends on the data not only through the posterior for $\bm{\phi}$ but also through the set of admissible values of $\mathbf{Q}$ under the NR. As a result, in contrast with GK, we cannot interpret the set of posterior means (\ref{eq:setofposteriormeans}) as a consistent estimator for the identified set for $\eta$ (which is not well-defined, as we discussed above). Nevertheless, the set of posterior means still carries a robust Bayesian interpretation similar to GK in that it clarifies posterior results that are robust to the choice of prior on the non-updated part of the parameter space (i.e., on the flat regions of the likelihood).

As in GK, we can also report a robust credible region with credibility level $\alpha$, which is the shortest interval estimate for $\eta$ such that the posterior probability put on the interval is greater than or equal to $\alpha$ uniformly over the posteriors in $\Pi_{\eta|\mathbf{Y}^{T}, D_{N} = 1}$ (see Proposition~1 of GK). One may also be interested in posterior lower and upper probabilities, which are the infimum and supremum, respectively, of the probability for a hypothesis over all posteriors in the class.

GK provide conditions under which their robust Bayesian approach has a valid frequentist interpretation, in the sense that the robust credible region is an asymptotically valid confidence set for the true identified set. For the same reason as mentioned above, however, frequentist validity of the robust credible region does not immediately extend to the NR case. We provide conditions under which the robust credible region has a valid frequentist interpretation in Section~\ref{sec:fewNR}.

\subsection{Numerical implementation of robust Bayesian approach}
\label{subsec:numericalimplementation}

This section describes a general algorithm to implement our robust Bayesian procedure under NR. GK propose numerical algorithms for conducting robust Bayesian inference in SVARs identified using traditional sign and zero restrictions. Their Algorithm~1 uses a numerical optimization routine to obtain the lower and upper bounds of the identified set at each draw of $\bm{\phi}$. Obtaining the bounds via numerical optimization is not generally applicable under the class of NR considered in AR18, since the constraints on the historical decomposition are not differentiable everywhere in $\mathbf{Q}$. We therefore adapt Algorithm~2 of GK, which approximates the bounds of the identified set at each draw of $\bm{\phi}$ using Monte Carlo simulation.

\bigskip

\noindent \textbf{Algorithm 1.} \textit{Let $N(\bm{\phi},\mathbf{Q},\mathbf{Y}^{T}) \geq \mathbf{0}_{s\times 1}$ be the set of NR and let $S(\bm{\phi},\mathbf{Q}) \geq \mathbf{0}_{\tilde{s}\times 1}$ be the set of traditional sign restrictions (excluding the sign normalization). Assume the object of interest is $\eta_{i,j^{*},h} = c_{i,h}'(\bm{\phi})\mathbf{q}_{j^{*}}$.
\begin{itemize}
  \item \textbf{Step~1}: Specify a prior for $\bm{\phi}$, $\pi_{\bm{\phi}}$, and obtain the posterior $\pi_{\bm{\phi}|\mathbf{Y}^{T}}$.
  \item \textbf{Step~2}: Draw $\bm{\phi}$ from $\pi_{\bm{\phi}|\mathbf{Y}^{T}}$ and check whether $\mathcal{Q}(\bm{\phi}|\mathbf{Y}^{T},N,S)$ is empty using the subroutine below.
  \begin{itemize}
    \item \textbf{Step~2.1}: Draw an $n\times n$ matrix of independent standard normal random variables, $\mathbf{Z}$, and let $\mathbf{Z} = \tilde{\mathbf{Q}}\mathbf{R}$ be the QR decomposition of $\mathbf{Z}$.\footnote{This is the algorithm used by \cite{Rubio-Ramirez_Waggoner_Zha_2010} to draw from the uniform distribution over $\mathcal{O}(n)$, except that we do not normalize the diagonal elements of $\mathbf{R}$ to be positive. This is because we impose a sign normalization based on the diagonal elements of $\mathbf{A}_{0} = \mathbf{Q}'\bm{\Sigma}_{tr}^{-1}$ in Step~2.2.}
    \item \textbf{Step~2.2}: Define
        \begin{equation*}
          \mathbf{Q} = \left[\mathrm{sgn}((\bm{\Sigma}_{tr}^{-1}\mathbf{e}_{1,n})'\tilde{\mathbf{q}}_{1})\frac{\tilde{\mathbf{q}}_{1}}{\lVert \tilde{\mathbf{q}}_{1}\rVert},\ldots, \mathrm{sgn}((\bm{\Sigma}_{tr}^{-1}\mathbf{e}_{n,n})'\tilde{\mathbf{q}}_{n})\frac{\tilde{\mathbf{q}}_{n}}{\lVert \tilde{\mathbf{q}}_{n}\rVert}\right],
        \end{equation*}
        where $\tilde{\mathbf{q}}_{j}$ is the $j$th column of $\tilde{\mathbf{Q}}$.
    \item \textbf{Step~2.3}: Check whether $\mathbf{Q}$ satisfies $\mathbf{N}(\bm{\phi},\mathbf{Q},\mathbf{Y}^{T}) \geq \mathbf{0}_{s\times 1}$ and $S(\bm{\phi},\mathbf{Q}) \geq \mathbf{0}_{\tilde{s}\times 1}$. If so, retain $\mathbf{Q}$ and proceed to Step~3. Otherwise, repeat Steps~2.1 and 2.2 (up to a maximum of $L$ times) until $\mathbf{Q}$ is obtained satisfying the restrictions. If no draws of $\mathbf{Q}$ satisfy the restrictions, approximate $\mathcal{Q}(\bm{\phi}|\mathbf{Y}^{T},N,S)$ as being empty and return to Step~2.
  \end{itemize}
  \item \textbf{Step~3}: Repeat Steps~2.1--2.3 until $K$ draws of $\mathbf{Q}$ are obtained. Let $\{\mathbf{Q}_{k}, k=1,...,K\}$ be the $K$ draws of $\mathbf{Q}$ that satisfy the restrictions and let $\mathbf{q}_{j^{*},k}$ be the $j^{*}$th column of $\mathbf{Q}_{k}$. Approximate $[l(\bm{\phi},\mathbf{Y}^{T}),u(\bm{\phi},\mathbf{Y}^{T})]$ by $[\min_{k}\mathbf{c}_{i,h}'(\bm{\phi})\mathbf{q}_{j^{*},k}$, $\max_{k}\mathbf{c}_{i,h}'(\bm{\phi})\mathbf{q}_{j^{*},k}]$.
  \item \textbf{Step~4}: Repeat Steps 2--3 $M$ times to obtain $[l(\bm{\phi_{m}},\mathbf{Y}^{T}),u(\bm{\phi_{m}},\mathbf{Y}^{T})]$ for $m=1,...,M$. Approximate the set of posterior means using the sample averages of $l(\bm{\phi_{m}},\mathbf{Y}^{T})$ and $u(\bm{\phi_{m}},\mathbf{Y}^{T})$.
  \item \textbf{Step~5}: To obtain an approximation of the smallest robust credible region with credibility $\alpha \in (0,1)$, define $d(\eta,\bm{\phi},\mathbf{Y}^{T}) = \max\{|\eta-l(\bm{\phi},\mathbf{Y}^{T})|,|\eta-u(\bm{\phi},\mathbf{Y}^{T})|\}$ and let $\hat{z}_{\alpha}(\eta)$ be the sample $\alpha$-th quantile of $\{d(\eta,\bm{\phi_{m}},\mathbf{Y}^{T}), m=1,...,M\}$. An approximated smallest robust credible interval for $\eta_{i,j^{*},h}$ is an interval centered at $\arg \min_{\eta} \hat{z}_{\alpha}(\eta)$ with radius $\min_{\eta}\hat{z}_{\alpha}(\eta)$.
\end{itemize}
}

\bigskip

Algorithm~1 approximates $[l(\bm{\phi},\mathbf{Y}^{T}),u(\bm{\phi},\mathbf{Y}^{T})]$ at each draw of $\bm{\phi}$ via Monte Carlo simulation. The approximated set will be too narrow given a finite number of draws of $\mathbf{Q}$, but the approximation error will vanish as the number of draws goes to infinity. The algorithm may be computationally demanding when the restrictions substantially truncate $\mathcal{Q}(\bm{\phi}|\mathbf{Y}^{T},N,S)$, because many draws of $\mathbf{Q}$ from $\mathcal{O}(n)$ may be rejected at each draw of $\bm{\phi}$. However, the same draws of $\mathbf{Q}$ can be used to compute $l(\bm{\phi},\mathbf{Y}^{T})$ and $u(\bm{\phi},\mathbf{Y}^{T})$ for different objects of interest, which cuts down on computation time. For example, the same draws of $\mathbf{Q}$ can be used to compute the impulse responses of all variables to all shocks at all horizons of interest. They can also be used to compute other parameters by replacing $\eta_{i,j^{*},h}$ with some other function, such as the forecast error variance decomposition, an element of $\mathbf{A}_{0}$, the historical decomposition or the structural shocks themselves in particular periods.\footnote{Impulse responses to a unit shock -- rather than a standard-deviation shock -- can be computed as in Algorithm~3 of Giacomini et al. (2019)\nocite{Giacomini_Kitagawa_Read_2019}.} Step~3 is parallelizable, so reductions in computing time are possible by distributing computation across multiple processors. Other algorithms may be computationally more efficient than Algorithm~1 in particular cases. We discuss these in Appendix~\ref{sec:appendixalgorithms}.

\section{Frequentist coverage under a few NR}
\label{sec:fewNR}

In this section, we show that the robust Bayes credible region attains asymptotically valid frequentist coverage in a setting where the number of NR is small relative to the length of the sampled periods in a sense that we make precise in the next assumption. This assumption is empirically relevant given that applications typically impose these restrictions in at most a handful of periods.

\begin{assumption} \label{assump:dim(s)} (fixed-dimensional $\mathbf{s}(\mathbf{Y}^T)$): The conditional identified set under NR has sufficient statistics $\mathbf{s}(\mathbf{Y}^T)$, as defined in Definition \ref{def:identifiedset}(ii), and the dimension of $\mathbf{s}(\mathbf{Y}^T)$ does not depend on $T$.
\end{assumption}

Let $(\bm{\phi}_{0},\mathbf{Q}_{0})$ be the true parameter values. We view the sample $\mathbf{Y}^T$ as being drawn from $p(\mathbf{Y}^T|\bm{\phi}_0)$. Let $p(\mathbf{Y}^T|\bm{\phi}_0, \mathbf{s})$ be the conditional distribution of the sample $\mathbf{Y}^T$ given the sufficient statistics for the conditional identified set $\mathbf{s}= \mathbf{s}(\mathbf{Y}^T)$ at $\bm{\phi}= \bm{\phi}_0$. We denote by $p(\mathbf{s}|\bm{\phi}_0)$ the distribution of the sufficient statistics $\mathbf{s}(\mathbf{Y}^T)$ at $\bm{\phi}=\bm{\phi}_0$. The next assumption assumes that in the conditional experiment given $\mathbf{s}(\mathbf{Y}^T)$, the sampling distribution for the maximum likelihood estimator $\hat{\bm{\phi}} \equiv \arg \max_{\bm{\phi}} p(\mathbf{Y}^T|\bm{\phi})$ centered at $\bm{\phi}_0$ and the posterior for $\bm{\phi}$ centered at $\hat{\bm{\phi}}$ asymptotically coincide.

\begin{assumption} \label{assump:B-vM} (Conditional Bernstein-von Mises property for $\bm{\phi}$): For $p(\mathbf{s}|\bm{\phi}_0)$-almost every $\mathbf{s}$ and $p(\mathbf{Y}^T|\bm{\phi}_0, \mathbf{s})$-almost every sampling sequence $\mathbf{Y}^T$, the posterior  for $\sqrt{T}(\bm{\phi} - \hat{\bm{\phi}})$ asymptotically coincides with the sampling distribution of $\sqrt{T}(\hat{\bm{\phi}} - \bm{\phi}_0)$ with respect to $p(\mathbf{Y}^T|\bm{\phi}_0, \mathbf{s})$, as $T \to \infty$, in the sense stated in Assumption~5(i) in GK.
\end{assumption}

This is a key assumption for establishing the asymptotic frequentist validity of the robust credible region under NR. It holds, for instance, when $\mathbf{s}(\mathbf{y}^T)$ corresponds to one or  a few observations in the whole sample, as we had in the toy example of Section \ref{subsec:sign}. In this case, the influence of $\mathbf{s}(\mathbf{y}^T)$ vanishes in the conditional sampling distribution of $\sqrt{T}(\hat{\bm{\phi}} - \bm{\phi}_0)$ as $T \to \infty$, as the latter asymptotically agrees with the asymptotically normal sampling distribution for the maximum likelihood estimator with variance-covariance matrix given by the inverse of the Fisher information matrix. By the well-known Bernstein-von Mises theorem for regular parametric models, the posterior  for $\sqrt{T}(\bm{\phi} - \hat{\bm{\phi}})$ asymptotically agrees with this sampling distribution.

The last assumption requires convexity and smoothness of the conditional identified set, and is analogous to Assumption~5(ii) of GK for standard set-identified models.

\begin{assumption} \label{assump:as_convex} (Almost-sure convexity and smoothness of the impulse response identified set): Let $\widetilde{CIS}_{\eta}(\bm{\phi}|\mathbf{s}(\mathbf{Y}^T),N)$ be the conditional identified set for $\eta$ with the sufficient statistics $\mathbf{s}(\mathbf{Y}^T)$. For $p(\mathbf{Y}^{T}|\bm{\phi}_0)$-almost every $\mathbf{Y}^T$, $\widetilde{CIS}_{\eta}(\bm{\phi}|\mathbf{s}(\mathbf{y}^T),N)$ is closed and convex, $\widetilde{CIS}_{\eta}(\bm{\phi}|\mathbf{s}(\mathbf{y}^T),N) = [\tilde{\bm{\ell}}(\bm{\phi}, \mathbf{s}(\mathbf{Y}^T)), \tilde{\mathbf{u}}(\bm{\phi}, \mathbf{s}(\mathbf{Y}^T))]$, and its lower and upper bounds are differentiable in $\bm{\phi}$ at $\bm{\phi} = \bm{\phi}_0$ with nonzero derivatives.
\end{assumption}

Propositions~\ref{prop:convexity}--\ref{prop:differentiability} in Appendix~\ref{sec:proofs} provide primitive conditions for Assumption~\ref{assump:as_convex} to hold in the case where there are shock-sign restrictions. Imposing Assumptions~\ref{assump:dim(s)}, \ref{assump:B-vM} and \ref{assump:as_convex}, we obtain the following theorem.

\begin{theorem}\label{thm:coverage} For $\gamma \in (0,1)$, let $\widehat{C}_{\alpha}^{\ast}$ be the volume-minimizing robust credible region for $\eta$ with credibility $\alpha$,\footnote{The volume-minimizing robust credible region $\widehat{C}_{\alpha}^{\ast}$ is defined as a shortest interval among the connected intervals $C_{\alpha}$ satisfying
\begin{equation*}
P_{\mathbf{Y}^T|\mathbf{s},\bm{\phi}}(\widetilde{CIS}_{\eta}(\bm{\phi}_0 |\mathbf{s}(\mathbf{Y}^T),N) \subset C_{\alpha} |\mathbf{s}(\mathbf{Y}^T), \bm{\phi}_0) \geq \alpha.
\end{equation*}
See Proposition 1 in GK for a procedure to compute the volume-minimizing credible region.} which satisfies
\begin{equation}
\inf_{\pi \in  \Pi_{\bm{\phi},\mathbf{Q}|\mathbf{Y}^{T},D_{N} = 1}} \pi( \widehat{C}_{\alpha}^{\ast} ) = \pi_{\bm{\phi}|\mathbf{Y}^{T},D_{N} = 1}(CIS_{\eta}(\bm{\phi}|\mathbf{Y}^{T},N) \subset \hat{C}_{\alpha}^{\ast}|\mathbf{Y}^T, D_{N} = 1) = \alpha.
\end{equation}
Under Assumptions \ref{assump:dim(s)}, \ref{assump:B-vM}, and \ref{assump:as_convex}, $\widehat{C}_{\alpha}^{\ast}$ attains asymptotically valid coverage for the true impulse response, $\eta_0$, conditional on $\mathbf{s}(\mathbf{Y}^T)$.
\begin{multline}
\liminf_{T \to \infty} P_{\mathbf{Y}^T|\mathbf{s},\bm{\phi}}( \eta_0\in \widehat{C}_{\alpha}^{\ast} |\mathbf{s}(\mathbf{Y}^T), \bm{\phi}_0) \geq \\
\lim_{T \to \infty} P_{\mathbf{Y}^T|\mathbf{s},\bm{\phi}}(\widetilde{CIS}_{\eta}(\bm{\phi}_0|\mathbf{s}(\mathbf{Y}^{T}),N ) \subset \widehat{C}_{\alpha}^{\ast} |\mathbf{s}(\mathbf{Y}^T), \bm{\phi}_0) = \alpha.
\end{multline}
Accordingly, $\widehat{C}_{\alpha}^{\ast}$ attains asymptotically valid coverage for $\eta_0$ unconditionally,
\begin{equation}
\liminf_{T \to \infty} P_{\mathbf{Y}^T|\bm{\phi}}( \eta_0\in \widehat{C}_{\alpha}^{\ast} | \bm{\phi}_0) \geq \lim_{T \to \infty} P_{\mathbf{Y}^T|\bm{\phi}}(\widetilde{CIS}_{\eta}(\bm{\phi}_0|\mathbf{s}(\mathbf{Y}^{T}),N ) \subset \widehat{C}_{\alpha}^{\ast} | \bm{\phi}_0) = \alpha.
\end{equation}
\end{theorem}

\begin{proof}
See Appendix B.
\end{proof}

This theorem shows that the robust credible region of GK applied to the SVAR model with NR attains asymptotically valid frequentist coverage for the true impulse response as well as the conditional impulse-response identified set. Even if the point-identification condition of Proposition~\ref{prop:point_identification} holds for the impulse response, it is not obvious if the standard Bayesian credible region can attain frequentist coverage. This is because the Bernstein-von Mises theorem does not seem to hold for the impulse response due to the non-standard features of models with NR.

One could also consider asymptotics under an increasing number of restrictions. We conjecture that, under certain assumptions about how the NR are generated, the class of posteriors for $\eta$ will converge to a point mass at its true value. An implication is that the posterior mean under \textit{any} conditional prior for $\mathbf{Q}$ that places probability one on the identified set would be consistent for the true value. This result would be an interesting contrast to the case under traditional set-identifying restrictions, where the location of the posterior mean within the identified set is determined purely by the conditional prior and the posterior quantiles lie strictly within the identified set (e.g., \cite{Moon_Schorfheide_2012}). We do not analyze this case here, since the assumption that there is a fixed number of restrictions seems to be of primary interest. However, Appendix~\ref{sec:increasingrestrictions} provides numerical evidence in support of our conjecture. We leave formal investigation of this conjecture for future work.

\section{Empirical application: the dynamic effects of a monetary policy shock}
\label{sec:empirical}

AR18 estimate the effects of monetary policy shocks on the US economy using a combination of sign restrictions on impulse responses and NR. The reduced-form VAR is the same as that used in \cite{Uhlig_2005}. The model's endogenous variables are real GDP, the GDP deflator, a commodity price index, total reserves, non-borrowed reserves (all in natural logarithms) and the federal funds rate; see Arias, Caldara and Rubio-Ram\'{i}rez (2019)\nocite{Arias_Caldara_Rubio-Ramirez_2019} for details on the variables. The data are monthly and run from January 1965 to November 2007. The VAR includes 12~lags and we include a constant.

As NR, AR18 impose that the monetary policy shock in October 1979 was positive and that it was the overwhelming contributor to the unexpected change in the federal funds rate in that month. This was the month in which the Federal Reserve markedly and unexpectedly increased the federal funds rate following the appointment of Paul Volcker as chairman of the Federal Reserve, and is widely considered to be an example of a positive monetary policy shock (e.g., \cite{Romer_Romer_1989}). The traditional sign restrictions considered in \cite{Uhlig_2005} are also imposed. Specifically, the response of the federal funds rate is restricted to be non-negative for $h=0,1,\ldots,5$ and the responses of the GDP deflator, the commodity price index and nonborrowed reserves are restricted to be nonpositive for $h=0,1,\ldots,5$.

We assume a Jeffreys' (improper) prior over the reduced-form parameters, $\pi_{\bm{\phi}} = \pi_{\mathbf{B},\bm{\Sigma}} \propto |\bm{\Sigma}|^{-\frac{n+1}{2}}$, which is truncated so that the VAR is stable. The posterior for the reduced-form parameters, $\pi_{\bm{\phi}|\mathbf{Y}^{T}}$, is then a normal-inverse-Wishart distribution, from which it is straightforward to obtain independent draws (for example, see \cite{DelNegro_Schorfheide_2011}). We obtain 1,000 draws from the posterior of $\bm{\phi}$ such that the VAR is stable and $\mathcal{Q}(\bm{\phi}|\mathbf{Y}^{T},N,S)$ is non-empty. We use Algorithm~1 with $K = 10,000$ draws of $\mathbf{Q}$ at each draw of $\bm{\phi}$ to approximate $l(\bm{\phi},\mathbf{Y}^{T})$ and $u(\bm{\phi},\mathbf{Y}^{T})$. If we cannot obtain a draw of $\mathbf{Q}$ satisfying the restrictions after 100,000 draws of $\mathbf{Q}$, we approximate $\mathcal{Q}(\bm{\phi}|\mathbf{Y}^{T},N,S)$ as being empty at that draw of $\bm{\phi}$.

We explore the sensitivity of posterior inference to the choice of prior for $\mathbf{Q}|\bm{\phi}$ when the unconditional likelihood is used to construct the posterior. For brevity, we report only the impulse responses of the federal funds rate and real GDP to a positive standard-deviation monetary policy shock (Figure~\ref{fig:MP_Oct79}). As a point of comparison, we report results obtained using a conditionally uniform prior for $\mathbf{Q}|\bm{\phi}$. Under this prior, the 68~per cent highest posterior density credible intervals for the response of real GDP exclude zero at horizons greater than a year or so.\footnote{The results are not directly comparable to those presented in Figure~6 of AR18. First, we present responses to a standard-deviation shock, whereas AR18 describe their responses as being to a 25~basis point shock (although, from close inspection of their Figure~6, it is evident that this normalization is not imposed correctly, because the impact response of the federal funds rate fans out around zero). Second, we use a prior for $\mathbf{Q}$ that is conditionally uniform given $\bm{\phi}$, whereas AR18 use a prior that is unconditionally uniform.} In contrast, the 68~per cent robust credible intervals include zero at all horizons. Under the single prior, the posterior probability that the output response is negative two years after the shock is 95~per cent. In contrast, the posterior lower probability of this event -- the smallest probability over the class of posteriors generated by the class of priors -- is only 54~per cent. The results suggest that posterior inference about the effect of monetary policy on output can be sensitive to the choice of (unrevisable) prior for $\mathbf{Q}|\bm{\phi}$.

\begin{figure}[h]
    \center
    \caption{Impulse Responses to a Monetary Policy Shock} \label{fig:MP_Oct79}
    \begin{tabular}{ccc}
        \includegraphics[scale=0.5]{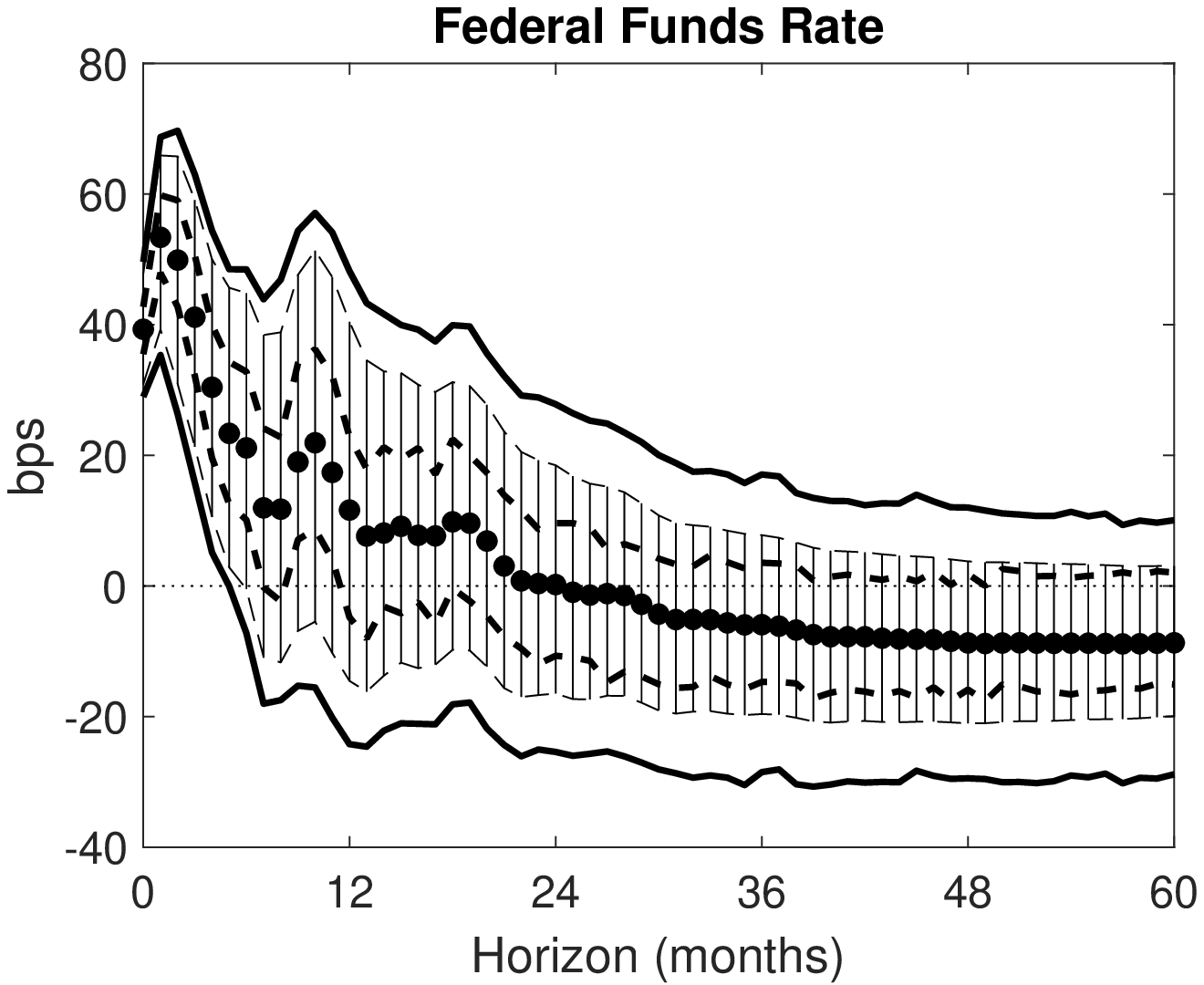} & \phantom{a} & \includegraphics[scale=0.5]{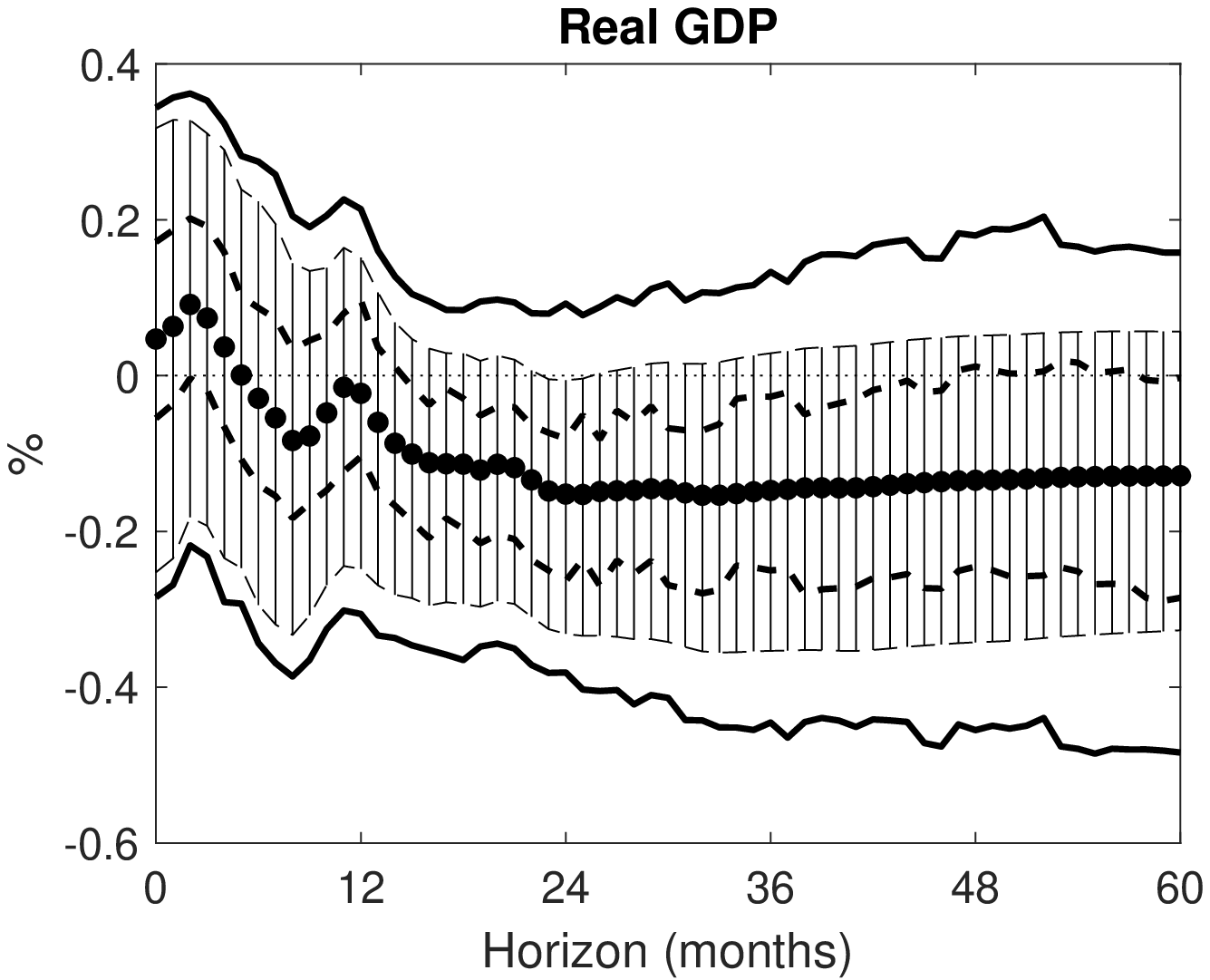} \\
    \end{tabular}
    \footnotesize \parbox[t]{0.65 in}{Notes:}\parbox[t]{5 in}{Circles and dashed lines are, respectively, posterior means and 68~per cent (pointwise) highest posterior density intervals under the uniform prior for $\mathbf{Q}|\bm{\phi}$; vertical bars are sets of posterior means and solid lines are 68~per cent (pointwise) robust credible regions obtained using Algorithm~1 with 10,000 draws from $\mathcal{Q}(\bm{\phi}|\mathbf{Y}^{T},N,S)$; results are based on 1,000 draws from the posterior of $\bm{\phi}$ with nonempty $\mathcal{Q}(\bm{\phi}|\mathbf{Y}^{T},N,S)$; impulse responses are to a standard-deviation shock.} \\
\end{figure}

AR18 also consider an alternative set of restrictions. Specifically, they impose that the monetary policy shock was: positive in April 1974, October 1979, December 1988 and February 1994; negative in December 1990, October 1998, April 2001 and November 2002; and the most important contributor to the observed unexpected change in the federal funds rate in these months. The choice of these dates is based on a synthesis of information from different sources, including the chronology of monetary policy actions from \cite{Romer_Romer_1989}, an updated series of the monetary policy shocks constructed using Greenbook forecasts in \cite{Romer_Romer_2004}, the high-frequency monetary policy surprises from G\"{u}rkaynak, Sack and Swanson (2005)\nocite{Gurkaynak_Sack_Swanson_2005}, and minutes from Federal Open Markets Committee meetings. Under this extended set of restrictions, the set of posterior means and the robust credible interval are tightened noticeably, particularly at shorter horizons (Figure~\ref{fig:MP_Oct79_ShockRank}). The posterior lower probability of a negative output response two years after the shock is now 80~per cent, compared with 54~per cent under the October 1979 restrictions.

Finally, we investigate how posterior inference about the output response is affected by replacing AR18's extended set of restrictions with a shock-rank restriction. Specifically, we estimate the set of output responses that are consistent with the restriction that the monetary policy shock in October 1979 was the largest positive realization of the monetary policy shock in the sample period.\footnote{The large number of inequality constraints and tight conditional identified set induced by the shock-rank restriction poses computational challenges when using Algorithm~1. Accordingly, we use an alternative algorithm to obtain the results. The algorithm adapts an algorithm in \cite{Amir-Ahmadi_Drautzburg_2021} and is described in Appendix~\ref{sec:appendixalgorithms}.} This restriction appears plausible given that the change in the federal funds rate in October 1979 was more positive than the change in the federal funds rate in the other periods identified by AR18 as containing notable monetary policy shocks (Table~\ref{tab:dffr}). The shock-rank restriction somewhat shrinks the set of posterior means and robust credible regions relative to those obtained under the restrictions on the historical decomposition. Nevertheless, the two sets of restrictions lead to similar (robust) posterior inferences about the output response. The 68~per cent robust credible intervals include zero at all horizons under both sets of restrictions. The posterior lower probability that output falls two years after the shock is 73~per cent under the shock-rank restriction, compared with 80~per cent under the restriction on the historical decomposition.

\begin{table}[h]
    \small
    \centering
    \setlength{\tabcolsep}{10pt}
    \begin{threeparttable}
        \captionsetup{justification=centering}
        \caption{Monthly Change in Federal Funds Rate (ppt)} \label{tab:dffr}
        \center
        \begin{tabular}{*{8}{c}}
            \toprule
            Oct 79      &       Apr 74  &       Dec 88  &       Feb 94  &       Dec 90      &       Oct 98  &       Apr 01  &       Nov 02  \\
            2.34        &       1.16    &       0.41    &       0.20     &       --0.50  &       --0.44  &       --0.51  &       --0.41  \\
            \bottomrule
        \end{tabular}
        \begin{tablenotes}[flushleft]
            \footnotesize\item Source: FRED
        \end{tablenotes}
    \end{threeparttable}
\end{table}

\begin{figure}[h]
    \center
    \caption{Impulse Responses to a Monetary Policy Shock -- Extended Restrictions vs Shock-rank Restriction} \label{fig:MP_Oct79_ShockRank}
    \begin{tabular}{ccc}
        \includegraphics[scale=0.5]{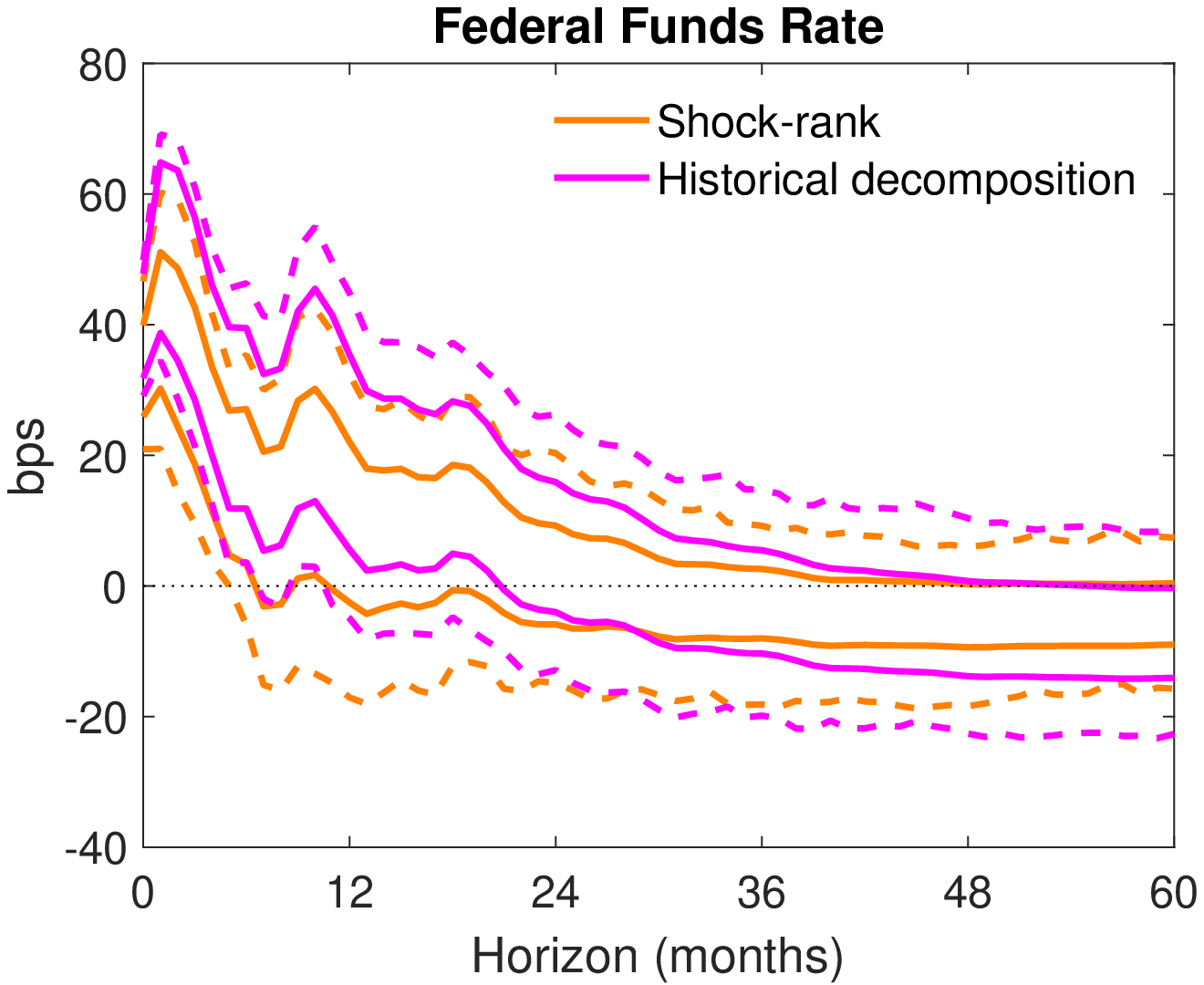} & \phantom{a} & \includegraphics[scale=0.5]{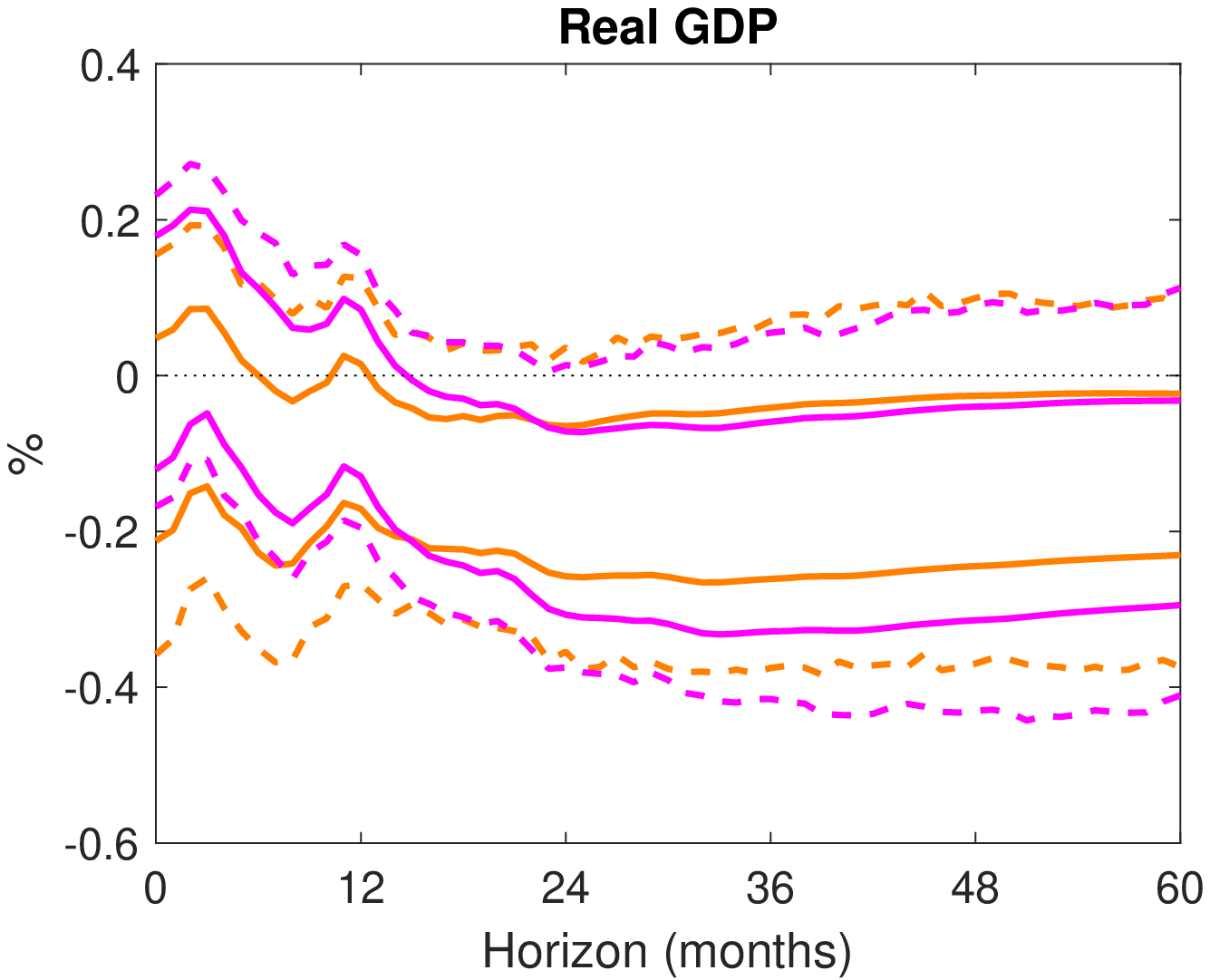} \\
    \end{tabular}
    \footnotesize \parbox[t]{0.65 in}{Notes:}\parbox[t]{5 in}{Solid lines represent set of posterior means and dashed lines represent 68~per cent (pointwise) robust credible regions; results are based on 1,000 draws from the posterior of $\bm{\phi}$ with nonempty $\mathcal{Q}(\bm{\phi}|\mathbf{Y}^{T},N,S)$; results under shock-rank restriction are obtained using Algorithm~D.1; results under restrictions on the historical decomposition are obtained using Algorithm~1 with 1,000 draws from $\mathcal{Q}(\bm{\phi}|\mathbf{Y}^{T},N,S)$; impulse responses are to a standard-deviation shock.} \\
\end{figure}

In general, $\mathcal{Q}(\bm{\phi}|\mathbf{Y}^{T},N,S)$ may be empty at particular values of $\bm{\phi}$. The proportion of draws of $\bm{\phi}$ where $\mathcal{Q}(\bm{\phi}|\mathbf{Y}^{T},N,S)$ is empty can therefore be used to assess the plausibility of the restrictions (see GK). Under the October 1979 restrictions, the posterior plausibility of the restrictions is one (i.e., every draw of $\bm{\phi}$ has a nonempty conditional identified set). In contrast, the posterior plausibility under AR18's extended set of restrictions is 53~per cent, while it is only 17~per cent under the shock-rank restriction.

\section{Conclusion}
\label{sec:conclusion}

Directly restricting the values of structural shocks to be consistent with historical narratives offers a potentially useful approach to disciplining SVARs, but raises novel issues related to identification and inference. These restrictions generate a set-valued mapping from the model's reduced-form parameters to its structural parameters that depends on the realization of the data entering the restrictions. This means that these restrictions do not fit neatly into the existing framework for analyzing identification in SVARs. In particular, we show that these restrictions may be point-identifying in a frequentist sense. We also highlight issues associated with existing standard Bayesian approaches to estimation and inference. Conditioning on the restrictions holding may result in the posterior placing more weight on parameters that yield a lower ex ante probability that the restrictions are satisfied. We therefore advocate using the unconditional likelihood when constructing the posterior. However, the observed unconditional likelihood will almost always possess flat regions, which implies that a component of the prior will not be updated by the data. Posterior inference may therefore be sensitive to the choice of prior. To address this, we provide robust Bayesian tools to assess or eliminate the sensitivity of posterior inference to the choice of prior. We also provide conditions under which these tools have a valid frequentist interpretation, so our approach should appeal to both Bayesians and frequentists.

While we focus on SVARs in the paper, our analysis could be extended to other settings. For example, \cite{Plagborg-Moller_Wolf_2020a} explain how to impose traditional SVAR identifying restrictions in the local projection framework under the assumption that the structural shocks are invertible. In Appendix~\ref{sec:localprojection} we briefly discuss how NR could also be imposed within the local projection framework, but we leave a formal analysis of this problem to future research.

\newpage

\begin{appendices}

\numberwithin{equation}{section}
\setcounter{equation}{0}
\numberwithin{figure}{section}
\setcounter{figure}{0}

\section{Bivariate example derivations}
\label{sec:appendixA}

\noindent\textbf{Set of values of $\theta$ under shock-sign restriction.} This section derives analytical expressions for the set of values of $\theta$ consistent with the shock-sign restriction in the bivariate example of Section~\ref{sec:bivariate}. Throughout, we assume that $\theta \in [-\pi,\pi]$.

Under the shock-sign restriction $\varepsilon_{1k} \geq 0$ and the sign normalization $\mathrm{diag}(\mathbf{A}_{0}) \geq \mathbf{0}_{2\times 1}$, $\theta$ is restricted to lie in the set
\begin{multline}
  \theta \in \left\{\theta: \sigma_{21}\sin\theta \leq \sigma_{22}\cos\theta, \cos\theta \geq 0, \sigma_{22}y_{1k}\cos\theta \geq (\sigma_{21}y_{1k}-\sigma_{11}y_{2k})\sin\theta \right\} \\
  \cup \left\{\theta: \sigma_{21}\sin\theta \leq \sigma_{22}\cos\theta, \cos\theta \leq 0, \sigma_{22}y_{1k}\cos\theta \geq (\sigma_{21}y_{1k}-\sigma_{11}y_{2k})\sin\theta \right\}.
\end{multline}

Consider the case where $\sigma_{21} < 0$ and $\sigma_{21}y_{1k}-\sigma_{11}y_{2k} < 0$. Then $\theta$ is restricted to the set
\begin{multline}
  \theta \in \left\{\theta: \tan\theta \geq \frac{\sigma_{22}}{\sigma_{21}}, \cos\theta >0 , \frac{\sigma_{22}y_{1k}}{\sigma_{21}y_{1k}-\sigma_{11}y_{2k}} \leq \tan\theta \right\} \cup \left\{\frac{\pi}{2}\right\} \\
  \cup \left\{\theta: \tan\theta \leq \frac{\sigma_{22}}{\sigma_{21}}, \cos\theta < 0, \frac{\sigma_{22}y_{1k}}{\sigma_{21}y_{1k}-\sigma_{11}y_{2k}} \geq \tan\theta\right\}.
\end{multline}
The inequalities in the first set hold if and only if $\tan\theta \geq \max\left\{\frac{\sigma_{22}}{\sigma_{21}},\frac{\sigma_{22}y_{1k}}{\sigma_{21}y_{1k}-\sigma_{11}y_{2k}}\right\}$ and $\theta \in (-\frac{\pi}{2},\frac{\pi}{2})$, which implies that
\begin{equation}\label{eq:a1}
  \arctan\left(\max\left\{\frac{\sigma_{22}}{\sigma_{21}},\frac{\sigma_{22}y_{1k}}{\sigma_{21}y_{1k}-\sigma_{11}y_{2k}}\right\}\right)\leq \theta < \frac{\pi}{2}.
\end{equation}
The inequalities on the second line hold if and only if $\tan\theta \leq \min\left\{\frac{\sigma_{22}}{\sigma_{21}}, \frac{\sigma_{22}y_{1k}}{\sigma_{21}y_{1k}-\sigma_{11}y_{2k}}\right\}$ and $\theta \in [-\pi,-\frac{\pi}{2}) \cup (\frac{\pi}{2},\pi]$. Since $\sigma_{21} < 0$, $\tan\theta$ must be negative, which implies that $\theta \in (\frac{\pi}{2},\pi]$. It follows that
\begin{equation}\label{eq:a2}
  \frac{\pi}{2} < \theta \leq \pi + \arctan\left(\min\left\{\frac{\sigma_{22}}{\sigma_{21}},\frac{\sigma_{22}y_{1k}}{\sigma_{21}y_{1k}-\sigma_{11}y_{2k}}\right\}\right).
\end{equation}
Taking the union of (\ref{eq:a1}), (\ref{eq:a2}) and $\left\{\frac{\pi}{2}\right\}$ implies that
\begin{equation}\label{eq:thetaset1}
  \theta \in \left[\arctan\left(\max\left\{\frac{\sigma_{22}}{\sigma_{21}},\frac{\sigma_{22}y_{1k}}{\sigma_{21}y_{1k}-\sigma_{11}y_{2k}}\right\}\right), \pi + \arctan\left(\min\left\{\frac{\sigma_{22}}{\sigma_{21}},\frac{\sigma_{22}y_{1k}}{\sigma_{21}y_{1k}-\sigma_{11}y_{2k}}\right\}\right)\right].
\end{equation}

Next, consider the case where $\sigma_{21} < 0$ and $\sigma_{21}y_{1k}-\sigma_{11}y_{2k} > 0$. Then $\theta$ is restricted to the set
\begin{multline}
  \theta \in \left\{\theta: \tan\theta \geq \frac{\sigma_{22}}{\sigma_{21}}, \cos\theta >0 , \frac{\sigma_{22}y_{1k}}{\sigma_{21}y_{1k}-\sigma_{11}y_{2k}} \geq \tan\theta \right\} \\
  \cup \left\{\theta: \tan\theta \leq \frac{\sigma_{22}}{\sigma_{21}}, \cos\theta < 0, \frac{\sigma_{22}y_{1k}}{\sigma_{21}y_{1k}-\sigma_{11}y_{2k}} \leq \tan\theta\right\}.
\end{multline}
If $y_{1k} > 0$ or if $y_{1k} < 0$ and $\frac{\sigma_{22}}{\sigma_{21}} < \frac{\sigma_{22}y_{1k}}{\sigma_{21}y_{1k}-\sigma_{11}y_{2k}}$, the second set of inequalities is not satisfied for any $\theta$, while the first set of inequalities is satisfied for
\begin{equation}\label{eq:thetaset}
  \theta \in \left[\arctan\left(\frac{\sigma_{22}}{\sigma_{21}}\right),\arctan\left(\frac{\sigma_{22}y_{1k}}{\sigma_{21}y_{1k}-\sigma_{11}y_{2k}}\right)\right].
\end{equation}
If $y_{1k} < 0$ and $\frac{\sigma_{22}}{\sigma_{21}} > \frac{\sigma_{22}y_{1k}}{\sigma_{21}y_{1k}-\sigma_{11}y_{2k}}$, the first set of inequalities has no solution and the second set is satisfied for
\begin{equation}
  \theta \in \left[\pi + \arctan\left(\frac{\sigma_{22}y_{1k}}{\sigma_{21}y_{1k}-\sigma_{11}y_{2k}}\right),\pi + \arctan\left(\frac{\sigma_{22}}{\sigma_{21}}\right)\right].
\end{equation}

In the case where $\sigma_{21} > 0$ and $\sigma_{21}y_{1k}-\sigma_{11}y_{2k} < 0$, $\theta$ is restricted to the set
\begin{multline}
  \theta \in \left\{\theta: \tan\theta \leq \frac{\sigma_{22}}{\sigma_{21}}, \cos\theta >0 , \frac{\sigma_{22}y_{1k}}{\sigma_{21}y_{1k}-\sigma_{11}y_{2k}} \leq \tan\theta \right\} \\
  \cup \left\{\theta: \tan\theta \geq \frac{\sigma_{22}}{\sigma_{21}}, \cos\theta < 0, \frac{\sigma_{22}y_{1k}}{\sigma_{21}y_{1k}-\sigma_{11}y_{2k}} \geq \tan\theta\right\}.
\end{multline}
If $y_{1k} > 0$ or if $y_{1k} < 0$ and $\frac{\sigma_{22}}{\sigma_{21}} > \frac{\sigma_{22}y_{1k}}{\sigma_{21}y_{1k}-\sigma_{11}y_{2k}}$, the second set of inequalities has no solution, while the first is satisfied for
\begin{equation}
  \theta \in \left[\arctan\left(\frac{\sigma_{22}y_{1k}}{\sigma_{21}y_{1k}-\sigma_{11}y_{2k}}\right),\arctan\left(\frac{\sigma_{22}}{\sigma_{21}}\right)\right].
\end{equation}
If $y_{1k} < 0$ and $\frac{\sigma_{22}}{\sigma_{21}} < \frac{\sigma_{22}y_{1k}}{\sigma_{21}y_{1k}-\sigma_{11}y_{2k}}$, the first set of inequalities has no solution and the second set is satisfied for
\begin{equation}
  \theta \in \left[-\pi + \arctan\left(\frac{\sigma_{22}}{\sigma_{21}}\right),-\pi + \arctan\left(\frac{\sigma_{22}y_{1k}}{\sigma_{21}y_{1k}-\sigma_{11}y_{2k}}\right)\right].
\end{equation}

Finally, in the case where $\sigma_{21} > 0$ and $\sigma_{21}y_{1k}-\sigma_{11}y_{2k} > 0$, $\theta$ is restricted to the set
\begin{multline}
  \theta \in \left\{\theta: \tan\theta \leq \frac{\sigma_{22}}{\sigma_{21}}, \cos\theta >0 , \frac{\sigma_{22}y_{1k}}{\sigma_{21}y_{1k}-\sigma_{11}y_{2k}} \geq \tan\theta \right\} \cup \left\{-\frac{\pi}{2}\right\} \\
  \cup \left\{\theta: \tan\theta \geq \frac{\sigma_{22}}{\sigma_{21}}, \cos\theta < 0, \frac{\sigma_{22}y_{1k}}{\sigma_{21}y_{1k}-\sigma_{11}y_{2k}} \leq \tan\theta\right\}.
\end{multline}
The first set of inequalities holds if and only if $\tan\theta \leq \min\left\{\frac{\sigma_{22}}{\sigma_{21}},\frac{\sigma_{22}y_{1k}}{\sigma_{21}y_{1k}-\sigma_{11}y_{2k}}\right\}$ and $\theta \in (-\frac{\pi}{2},\frac{\pi}{2})$, which implies that
\begin{equation}\label{eq:a3}
  -\frac{\pi}{2} < \theta \leq \arctan\left(\min\left\{\frac{\sigma_{22}}{\sigma_{21}},\frac{\sigma_{22}y_{1k}}{\sigma_{21}y_{1k}-\sigma_{11}y_{2k}}\right\}\right).
\end{equation}
The second set of inequalities holds if and only if $\tan\theta \geq \max\left\{\frac{\sigma_{22}}{\sigma_{21}},\frac{\sigma_{22}y_{1k}}{\sigma_{21}y_{1k}-\sigma_{11}y_{2k}}\right\}$ and $\theta \in [-\pi,-\frac{\pi}{2}) \cup (\frac{\pi}{2},\pi]$. Since $\sigma_{21} > 0$, $\tan\theta$ must be positive, which implies that $\theta \in [-\pi,-\frac{\pi}{2})$. It follows that
\begin{equation}\label{eq:a4}
  -\pi + \arctan\left(\max\left\{\frac{\sigma_{22}}{\sigma_{21}},\frac{\sigma_{22}y_{1k}}{\sigma_{21}y_{1k}-\sigma_{11}y_{2k}}\right\}\right) \leq \theta < -\frac{\pi}{2}.
\end{equation}
Taking the union of (\ref{eq:a3}), (\ref{eq:a4}) and $\left\{-\frac{\pi}{2}\right\}$ implies that
\begin{equation}
  \theta \in \left[-\pi + \arctan\left(\max\left\{\frac{\sigma_{22}}{\sigma_{21}},\frac{\sigma_{22}y_{1k}}{\sigma_{21}y_{1k}-\sigma_{11}y_{2k}}\right\}\right), \arctan\left(\min\left\{\frac{\sigma_{22}}{\sigma_{21}},\frac{\sigma_{22}y_{1k}}{\sigma_{21}y_{1k}-\sigma_{11}y_{2k}}\right\}\right)\right].
\end{equation}

\bigskip

\noindent\textbf{Set of values of $\eta$ under shock-sign restriction.} Here we derive the expression for the set of impulse responses $\eta \equiv \sigma_{11}\cos\theta$ consistent with the shock-sign restriction (i.e., (\ref{eq:impulseresponseset}) in Section~\ref{sec:bivariate}).

In the absence of restrictions, the set of admissible values for the matrix of contemporaneous impulse responses is
\begin{multline}
    \mathbf{A}_{0}^{-1} \in \left\{
    \begin{bmatrix}
      \sigma_{11}\cos\theta & -\sigma_{11}\sin\theta \\
      \sigma_{21}\cos\theta + \sigma_{22}\sin\theta & \sigma_{22}\cos\theta - \sigma_{21}\sin\theta
    \end{bmatrix}
     : \theta \in [-\pi,\pi]\right\} \\
     \cup \left\{
    \begin{bmatrix}
      \sigma_{11}\cos\theta & \sigma_{11}\sin\theta \\
      \sigma_{21}\cos\theta + \sigma_{22}\sin\theta & \sigma_{21}\sin\theta - \sigma_{22}\cos\theta
    \end{bmatrix}
    : \theta \in [-\pi,\pi]\right\}.
\end{multline}

Assume that $\sigma_{21} < 0$, $\sigma_{21}y_{1k}-\sigma_{11}y_{2k} > 0$ and $y_{1k} > 0$. Within the interval for $\theta$ defined in (\ref{eq:thetaset}), $\eta$ is maximized at $\theta = 0$, so $\eta_{ub} = \sigma_{11}$. The lower bound $\eta_{lb}$ occurs at one of the endpoints of the interval for $\theta$, so it satisfies
\begin{align}
  \eta_{lb} &= \min\left\{\sigma_{11}\cos\left(\arctan\left(\frac{\sigma_{22}}{\sigma_{21}}\right)\right),\sigma_{11}\cos\left(\arctan\left(\frac{\sigma_{22}y_{1k}}{\sigma_{21}y_{1k}-\sigma_{11}y_{2k}}\right)\right)\right\} \notag \\
  &= \min\left\{\sigma_{11}\cos\left(-\arctan\left(\frac{\sigma_{22}}{\sigma_{21}}\right)\right),\sigma_{11}\cos\left(\arctan\left(\frac{\sigma_{22}y_{1k}}{\sigma_{21}y_{1k}-\sigma_{11}y_{2k}}\right)\right)\right\} \notag \\
  &= \min\left\{\sigma_{11}\cos\left(\arctan\left(-\frac{\sigma_{22}}{\sigma_{21}}\right)\right),\sigma_{11}\cos\left(\arctan\left(\frac{\sigma_{22}y_{1k}}{\sigma_{21}y_{1k}-\sigma_{11}y_{2k}}\right)\right)\right\} \notag \\
  &= \sigma_{11}\cos\left(\max\left\{\arctan\left(-\frac{\sigma_{22}}{\sigma_{21}}\right),\arctan\left(\frac{\sigma_{22}y_{1k}}{\sigma_{21}y_{1k}-\sigma_{11}y_{2k}}\right)\right\}\right) \notag \\
  &= \sigma_{11}\cos\left(\arctan\left(\max\left\{-\frac{\sigma_{22}}{\sigma_{21}},\frac{\sigma_{22}y_{1k}}{\sigma_{21}y_{1k}-\sigma_{11}y_{2k}}\right\}\right)\right).
\end{align}
The second line follows from the fact that $\cos(.)$ is an even function and the third line follows from the fact that $\arctan(.)$ is an odd function. The arguments entering the $\cos(.)$ functions on the third line are both in the interval $[0,\frac{\pi}{2})$, so the fourth line follows from the fact that $\cos(.)$ is a decreasing function over this domain. The final line follow from the fact that $\arctan(.)$ is an increasing function.

\bigskip

\noindent\textbf{Restriction on the historical decomposition.} Under the restrictions that the first structural shock is positive in period $k$ and was the most important (or overwhelming) contributor to the change in the first variable, $\theta$ is restricted to lie in the set
\begin{multline}
  \theta \in \Big\{\theta: \sigma_{21}\sin\theta \leq \sigma_{22}\cos\theta, \cos\theta \geq 0, \sigma_{22}y_{1k}\cos\theta \geq (\sigma_{21}y_{1k}-\sigma_{11}y_{2k})\sin\theta, \\
    |\sigma_{22}y_{1k}\cos^{2} \theta + (\sigma_{11}y_{2k} - \sigma_{21}y_{1k})\cos\theta \sin \theta| \geq |\sigma_{22}y_{1k}\sin^{2}\theta + (\sigma_{21}y_{1k}-\sigma_{11}y_{2k})\cos\theta\sin\theta| \Big\} \\
    \cup \Big\{\theta: \sigma_{21}\sin\theta \leq \sigma_{22}\cos\theta, \cos\theta \leq 0, \sigma_{22}y_{1k}\cos\theta \geq (\sigma_{21}y_{1k}-\sigma_{11}y_{2k})\sin\theta, \\
    |\sigma_{22}y_{1k}\cos^{2} \theta + (\sigma_{11}y_{2k} - \sigma_{21}y_{1k})\cos\theta \sin \theta| \geq |\sigma_{22}y_{1k}\sin^{2}\theta + (\sigma_{21}y_{1k}-\sigma_{11}y_{2k})\cos\theta\sin\theta| \Big\}.
\end{multline}
As in the case of the shock-sign restriction, this set also depends on the data $\mathbf{y}_{k}$ independently of the reduced-form parameters.

\numberwithin{equation}{section}
\setcounter{equation}{0}
\section{Omitted proofs}
\label{sec:proofs}
\noindent\textbf{Proof of Proposition~\ref{prop:point_identification}.}
\begin{proof}
$\mathcal{H}(\bm{\phi}, \mathbf{Q})$ can be written as
\begin{align*}
\mathcal{H}(\bm{\phi}, \mathbf{Q}) = & \int_{\mathbf{Y}} f^{1/2}(\mathbf{y}^T|\bm{\phi})f^{1/2}(\mathbf{y}^T|\bm{\phi}_0) \cdot  D_N(\bm{\phi},\mathbf{Q},\mathbf{y}^T) D_N(\bm{\phi}_0,\mathbf{Q}_0,\mathbf{y}^T) d \mathbf{y}^T \\
&+ \int_{\mathbf{Y}} f^{1/2}(\mathbf{y}^T|\bm{\phi})f^{1/2}(\mathbf{y}^T|\bm{\phi}_0) \cdot (1- D_N(\bm{\phi},\mathbf{Q},\mathbf{y}^T)) (1-D_N(\bm{\phi}_0,\mathbf{Q}_0,\mathbf{y}^T)) d \mathbf{y}^T.
\end{align*}
Note that the likelihood for the reduced-form parameters $f(\mathbf{y}^T|\bm{\phi})$ point-identifies $\bm{\phi}$, so $f(\cdot|\bm{\phi})=f(\cdot|\bm{\phi}_0)$ holds only at $\bm{\phi} = \bm{\phi}_0$. Hence, we set $\bm{\phi}=\bm{\phi}_0$ and consider $\mathcal{H}(\bm{\phi}_0, \mathbf{Q})$,
\begin{equation*}
\mathcal{H}(\bm{\phi}_0, \mathbf{Q}) =\int_{\{ \mathbf{y}^T:D_N(\bm{\phi}_0,\mathbf{Q}, \mathbf{y}^T) = D_N(\bm{\phi}_0,\mathbf{Q}_0, \mathbf{y}^T) \} } f(\mathbf{y}^T| \bm{\phi}_0) d \mathbf{y}^T.
\end{equation*}
Hence, $\mathcal{H}(\bm{\phi}_0, \mathbf{Q}) = 1$ if and only if $D_N(\bm{\phi}_0,\mathbf{Q}, \mathbf{y}^T) = D_N(\bm{\phi}_0,\mathbf{Q}_0, \mathbf{y}^T)$ holds $f(\mathbf{Y}^T| \bm{\phi}_0)$-a.s. In terms of the reduced-form residuals entering the NR, the latter condition is equivalent to $\{\mathbf{U}: N(\bm{\phi}_0, \mathbf{Q}, \mathbf{Y}^{T}) \geq \mathbf{0}_{s \times 1} \} = \{ \mathbf{U}: N(\bm{\phi}_0, \mathbf{Q}_0, \mathbf{Y}^{T}) \geq \mathbf{0}_{s \times 1} \}$ up to $f(\mathbf{Y}^T| \bm{\phi}_0)$-null set. Hence, $\mathcal{Q}^{\ast}$ defined in the proposition collects observationally equivalent values of $\mathbf{Q}$ at $\bm{\phi} = \bm{\phi}_0$ in terms of the unconditional likelihood.

Next, consider the conditional likelihood and consider
\begin{align*}
\mathcal{H}_c(\bm{\phi}_0, \mathbf{Q}) & = \frac{1}{r^{1/2}(\bm{\phi}, \mathbf{Q})r^{1/2} (\bm{\phi}_0, \mathbf{Q}_0) } \int_{\mathbf{Y}} f(\mathbf{y}^T|\bm{\phi}_0) \cdot  D_N(\bm{\phi},\mathbf{Q},\mathbf{y}^T) D_N(\bm{\phi}_0,\mathbf{Q}_0,\mathbf{y}^T) d \mathbf{y}^T \\
& = \frac{ E_{\mathbf{Y}^T| \bm{\phi}_0} \left[ D_N(\bm{\phi}_0,\mathbf{Q},\mathbf{Y}^T) D_N(\bm{\phi}_0,\mathbf{Q}_0,\mathbf{Y}^T)\right] } {r^{1/2}(\bm{\phi}, \mathbf{Q})r^{1/2} (\bm{\phi}_0, \mathbf{Q}_0) } \notag \\
& \leq 1,
\end{align*}
where the inequality follows by the Cauchy-Schwartz inequality, and it holds with equality if and only if $D_N(\bm{\phi}_0,\mathbf{Q},\mathbf{Y}^T) = D_N(\bm{\phi}_0,\mathbf{Q}_0,\mathbf{Y}^T)$ holds $f(\mathbf{Y}^T| \bm{\phi}_0)$-a.s. Hence, by repeating the argument for the unconditional likelihood case, we conclude that $\mathcal{Q}^{\ast}$ consists of observationally equivalent values of $\mathbf{Q}$ at $\bm{\phi} = \bm{\phi}_0$ in terms of the conditional likelihood.
\end{proof}

\bigskip

\noindent\textbf{Proof of Theorem~\ref{thm:coverage}.} Since $(\bm{\phi}_{0},\mathbf{Q}_{0})$ satisfies the imposed NR $N(\bm{\phi}_{0},\mathbf{Q}_{0},\mathbf{y}^T) \geq \mathbf{0}_{s\times 1}$ and the other sign restrictions (if any imposed), $\eta_0 \in \widetilde{CIS}_{\eta}(\bm{\phi}_0|\mathbf{s}(\mathbf{y}^{T}),N)$ holds for any $\mathbf{y}^T$. Hence, for all $T$,
\begin{equation}
P_{\mathbf{Y}^T|\mathbf{s},\bm{\phi}}( \eta_0\in \widehat{C}_{\alpha}^{\ast} | \mathbf{s}(\mathbf{Y}^T), \bm{\phi}_0) \geq P_{\mathbf{Y}^T|\bm{\phi}}( \widetilde{CIS}_{\eta} (\bm{\phi}_0|\mathbf{s}(\mathbf{Y}^{T}),N) \subset \widehat{C}_{\alpha}^{\ast} |\mathbf{s}(\mathbf{Y}^T), \bm{\phi}_0).
\end{equation}
Hence, to prove the claim, it suffices to focus on the asymptotic behavior of the coverage probability for the conditional identified set shown in the right-hand side.

Under Assumption \ref{assump:B-vM} and \ref{assump:as_convex}, the asymptotically correct coverage for the conditional identified set can be obtained by applying Proposition~2 in GK. \qedsymbol

\bigskip

\noindent\textbf{Primitive Conditions for Assumption~\ref{assump:as_convex}.} In what follows, we present sufficient conditions for convexity, continuity and differentiability (both in $\bm{\phi}$) of the conditional impulse-response identified set under the assumption that there is a fixed number of shock-sign restrictions constraining the first structural shock only (possibly in multiple periods).

\bigskip

\begin{proposition}\label{prop:convexity}
\textbf{Convexity.} Let the parameter of interest be $\eta_{i,1,h}$, the impulse response of the $i$th variable at the $h$th horizon to the first structural shock. Assume that there are shock-sign restrictions on $\varepsilon_{1,t}$ for $t=t_{1},\ldots,t_{K}$, so $N(\bm{\phi},\mathbf{Q},\mathbf{Y}^{T}) = (\bm{\Sigma}_{tr}^{-1}\mathbf{u}_{t_{1}}, \ldots, \bm{\Sigma}_{tr}^{-1}\mathbf{u}_{t_{K}})'\mathbf{q}_{1} \geq \mathbf{0}_{K\times 1}$. Then the set of values of $\eta_{i,1,h}$ satisfying the shock-sign restrictions and sign normalization, $\{\eta_{i,1,h}(\bm{\phi},\mathbf{Q}) = \mathbf{c}_{i,h}(\bm{\phi})\mathbf{q}_{1}: N(\bm{\phi},\mathbf{Q},\mathbf{Y}^{T}) \geq \mathbf{0}_{K\times 1}, \mathrm{diag}(\mathbf{Q}'\bm{\Sigma}_{tr}^{-1}) \geq \mathbf{0}_{n\times 1}, \mathbf{Q} \in \mathcal{O}(n)\}$ is convex for all $i$ and $h$ if there exists a unit-length vector $\mathbf{q} \in \mathbb{R}^{n}$ satisfying
\begin{equation}\label{eq:propconvexityinequality}
\begin{bmatrix}
 (\bm{\Sigma}_{tr}^{-1}\mathbf{u}_{t_{1}}, \ldots, \bm{\Sigma}_{tr}^{-1}\mathbf{u}_{t_{K}})' \\
 (\bm{\Sigma}_{tr}^{-1}\mathbf{e}_{1,n})'
 \end{bmatrix}
 \mathbf{q} \geq \mathbf{0}_{(K+1)\times 1}.
\end{equation}
\end{proposition}

\bigskip

\noindent\textbf{Proof of Proposition~\ref{prop:convexity}.} If there exists a unit-length vector $\mathbf{q}$ satisfying the inequality in (\ref{eq:propconvexityinequality}), it must lie within the intersection of the $K$ half-spaces defined by the inequalities $(\bm{\Sigma}_{tr}^{-1}\mathbf{u}_{t_{k}})'\mathbf{q} \geq 0$, $k=1,\ldots,K$, the half-space defined by the sign normalization, $(\bm{\Sigma}_{tr}^{-1}\mathbf{e}_{1,n})'\mathbf{q} \geq 0$, and the unit sphere in $\mathbb{R}^{n}$. The intersection of these $K+1$ half-spaces and the unit sphere is a path-connected set. Since $\eta_{i,1,h}(\bm{\phi},\mathbf{Q})$ is a continuous function of $\mathbf{q}_{1}$, the set of values of $\eta_{i,1,h}$ satisfying the restrictions is an interval and is thus convex, because the set of a continuous function with a path-connected domain is always an interval. $\qedsymbol$

\bigskip

\begin{proposition}\label{prop:continuity}
\textbf{Continuity.} Let the parameter of interest and restrictions be as in Proposition~\ref{prop:convexity}, and assume that the conditions in the proposition are satisfied. If there exists a unit-length vector $\mathbf{q} \in \mathbb{R}^{n}$ such that, at $\bm{\phi} = \bm{\phi}_{0}$,
\begin{equation}\label{eq:propcontinuityinequality}
\begin{bmatrix}
 (\bm{\Sigma}_{tr}^{-1}\mathbf{u}_{t_{1}}, \ldots, \bm{\Sigma}_{tr}^{-1}\mathbf{u}_{t_{K}})' \\
 (\bm{\Sigma}_{tr}^{-1}\mathbf{e}_{1,n})'
 \end{bmatrix}
 \mathbf{q} >> \mathbf{0}_{(K+1)\times 1},
\end{equation}
then $u(\bm{\phi},\mathbf{Y}^{T})$ and $l(\bm{\phi},\mathbf{Y}^{T})$ are continuous at $\bm{\phi} = \bm{\phi}_{0}$ for all $i$ and $h$.\footnote{For a vector $\mathbf{x} = (x_{1},\ldots,x_{m})'$, $\mathbf{x} >> \mathbf{0}_{m\times 1}$ means that $x_{i} > 0$ for all $i=1,\ldots,m$.}
\end{proposition}

\bigskip

\noindent\textbf{Proof of Proposition~\ref{prop:continuity}.} $\mathbf{Y}^{T}$ enters the NR through the reduced-form VAR innovations, $\mathbf{u}_{t}$. After noting that the reduced-form VAR innovations are (implicitly) continuous in $\bm{\phi}$, continuity of $u(\bm{\phi},\mathbf{Y}^{T})$ and $l(\bm{\phi},\mathbf{Y}^{T})$ follows by the same logic as in the proof of Proposition~B.2 of \cite{Giacomini_Kitagawa_2020b}. We omit the detail for brevity. $\qedsymbol$

\bigskip

\begin{proposition}\label{prop:differentiability}
\textbf{Differentiability.} Let the parameter of interest and restrictions be as in Proposition~\ref{prop:convexity}, and assume that the conditions in the proposition are satisfied. Denote the unit sphere in $\mathbb{R}^{n}$ by $\mathcal{S}^{n-1}$. If, at $\bm{\phi} = \bm{\phi}_{0}$, the set of solutions to the optimization problem
\begin{equation}\label{eq:optimisationproblem}
  \max_{\mathbf{q} \in \mathcal{S}^{n-1}} \quad \left(\min_{\mathbf{q} \in \mathcal{S}^{n-1}}\right) \quad \mathbf{c}_{i,h}'(\bm{\phi})\mathbf{q} \quad \textrm{s.t.} \quad  \begin{bmatrix}(\bm{\Sigma}_{tr}^{-1}\mathbf{u}_{t_{1}}, \ldots, \bm{\Sigma}_{tr}^{-1}\mathbf{u}_{t_{K}}), & \bm{\Sigma}_{tr}^{-1}\mathbf{e}_{1,n}\end{bmatrix}'\mathbf{q} \geq \mathbf{0}_{(K+1)\times 1}
\end{equation}
is singleton, the optimized value $u(\bm{\phi},\mathbf{Y}^{T})$ ($l(\bm{\phi},\mathbf{Y}^{T})$) is nonzero, and the number of binding inequality restrictions at the optimum is at most $n-1$, then $u(\bm{\phi},\mathbf{Y}^{T})$ ($l(\bm{\phi},\mathbf{Y}^{T})$) is almost-surely differentiable at $\bm{\phi} = \bm{\phi}_{0}$.
\end{proposition}

\bigskip

\noindent\textbf{Proof of Proposition~\ref{prop:differentiability}.} One-to-one differentiable reparameterization of the optimization problem in Equation~(\ref{eq:optimisationproblem}) using $\mathbf{x} = \bm{\Sigma}_{tr}\mathbf{q}$ yields the optimization problem in Equation~(2.5) of Gafarov, Meier and Montiel-Olea (2018)\nocite{Gafarov_Meier_Montiel-Olea_2018} with a set of inequality restrictions that are now a function of the data through the reduced-form VAR innovations entering the NR. Noting that $\mathbf{u}_{t}$ is (implicitly) differentiable in $\bm{\phi}$, differentiability of $u(\bm{\phi},\mathbf{Y}^{T})$ at $\bm{\phi} = \bm{\phi}_{0}$ follows from their Theorem~2 under the assumptions that, at $\bm{\phi} = \bm{\phi}_{0}$, the set of solutions to the optimization problem is singleton, the optimized value $u(\bm{\phi},\mathbf{Y}^{T})$ is nonzero, and the number of binding sign restrictions at the optimum is at most $n-1$. Differentiability of $l(\bm{\phi},\mathbf{Y}^{T})$ follows similarly. Note that Theorem~2 of \cite{Gafarov_Meier_Montiel-Olea_2018} additionally requires that the column vectors of $\begin{bmatrix}(\bm{\Sigma}_{tr}^{-1}\mathbf{u}_{t_{1}}, \ldots, \bm{\Sigma}_{tr}^{-1}\mathbf{u}_{t_{K}}), & \bm{\Sigma}_{tr}^{-1}\mathbf{e}_{1,n}\end{bmatrix}$ are linearly independent, but this occurs almost-surely under the probability law for $\mathbf{Y}^{T}$. $\qedsymbol$

\section{Asymptotics with increasing number of NR}
\label{sec:increasingrestrictions}

What happens if the number of NR increases with the sample size? We conjecture that, under some assumptions about how the NR are generated, the class of posteriors for particular parameters will converge to a point mass at the truth. Intuitively, as the number of restrictions increases, the likelihood tends to be truncated to an increasing extent until the only point with positive likelihood is the true value of the parameter.

To provide some numerical evidence for this conjecture, we return to the bivariate example from Section~\ref{sec:bivariate}. Assume that $\bm{\phi}=\bm{\phi}_{0}$ is known (which will be the case asymptotically under regularity conditions), so $\pi_{\bm{\phi}|\mathbf{Y}^{T}}$ is a point mass at $\bm{\phi} = \bm{\phi}_{0}$. Consider the case where the econometrician observes $\mathrm{sgn}(\varepsilon_{1t})$ for all $t$ and imposes the shock-sign restriction $\mathrm{sgn}(\varepsilon_{1t})\varepsilon_{1t}(\bm{\phi},\mathbf{Q},\mathbf{y}_{t}) \geq 0$ for $t=1,\ldots,T$.\footnote{This assumption can be relaxed, so that the econometrician observes $\mathrm{sgn}(\varepsilon_{1t})$ each period with some probability.}

Figure~\ref{fig:posteriorConsistency} plots the conditional identified set for $\mathbf{q}_{1}$ for different numbers of restrictions given random realizations of a time series drawn from the data-generating process. As the sample size increases, the boundaries of the half-spaces generated by the binding shock-sign restrictions converge towards the true value of $\mathbf{q}_{1}$, $\mathbf{q}_{1,0}$. Additionally, the conditional identified set for $\mathbf{q}_{2}$ will converge to its true value, since $\mathbf{q}_{2}$ is orthogonal to $\mathbf{q}_{1}$ and satisfies a sign normalization.\footnote{Equivalently, we could show convergence of the conditional identified set for $\theta$ to $\theta_{0}$, which pins down all impulse responses.} In other words, imposing a growing number of shock-sign restrictions on a single shock is sufficient for the posterior of all impulse responses to converge to a point mass at the true value. Note, however, that this will not be the case in higher-dimensional VARs, since the collapse of the conditional identified set for $\mathbf{q}_{1}$ does not pin down values for $\mathbf{q}_{j}$, $j=2,\ldots,n$.

\begin{figure}[h]
    \center
    \caption{Illustration of Posterior Consistency} \label{fig:posteriorConsistency}
    \begin{tabular}{c c c}
        \includegraphics[scale=0.35]{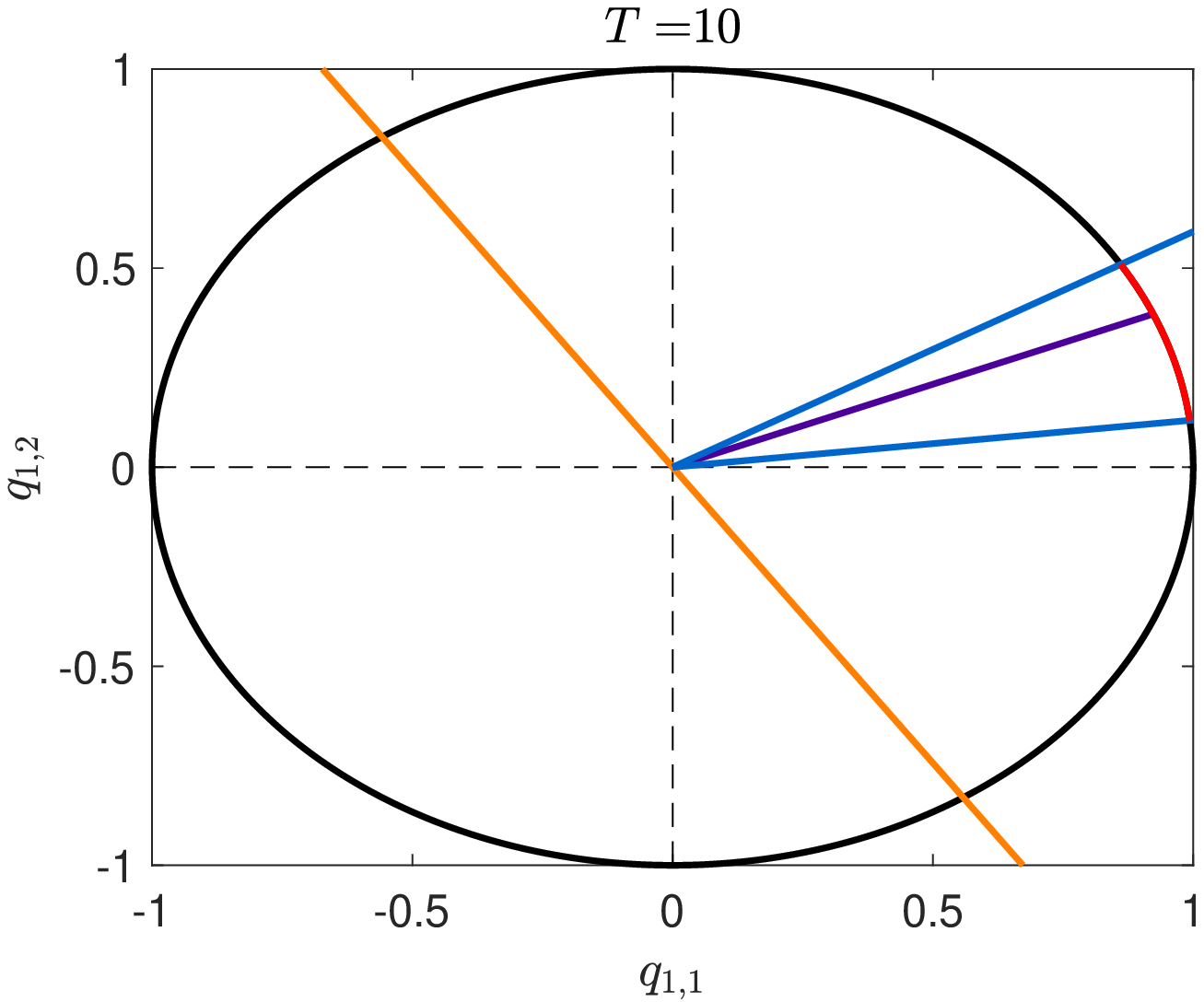} & \includegraphics[scale=0.35]{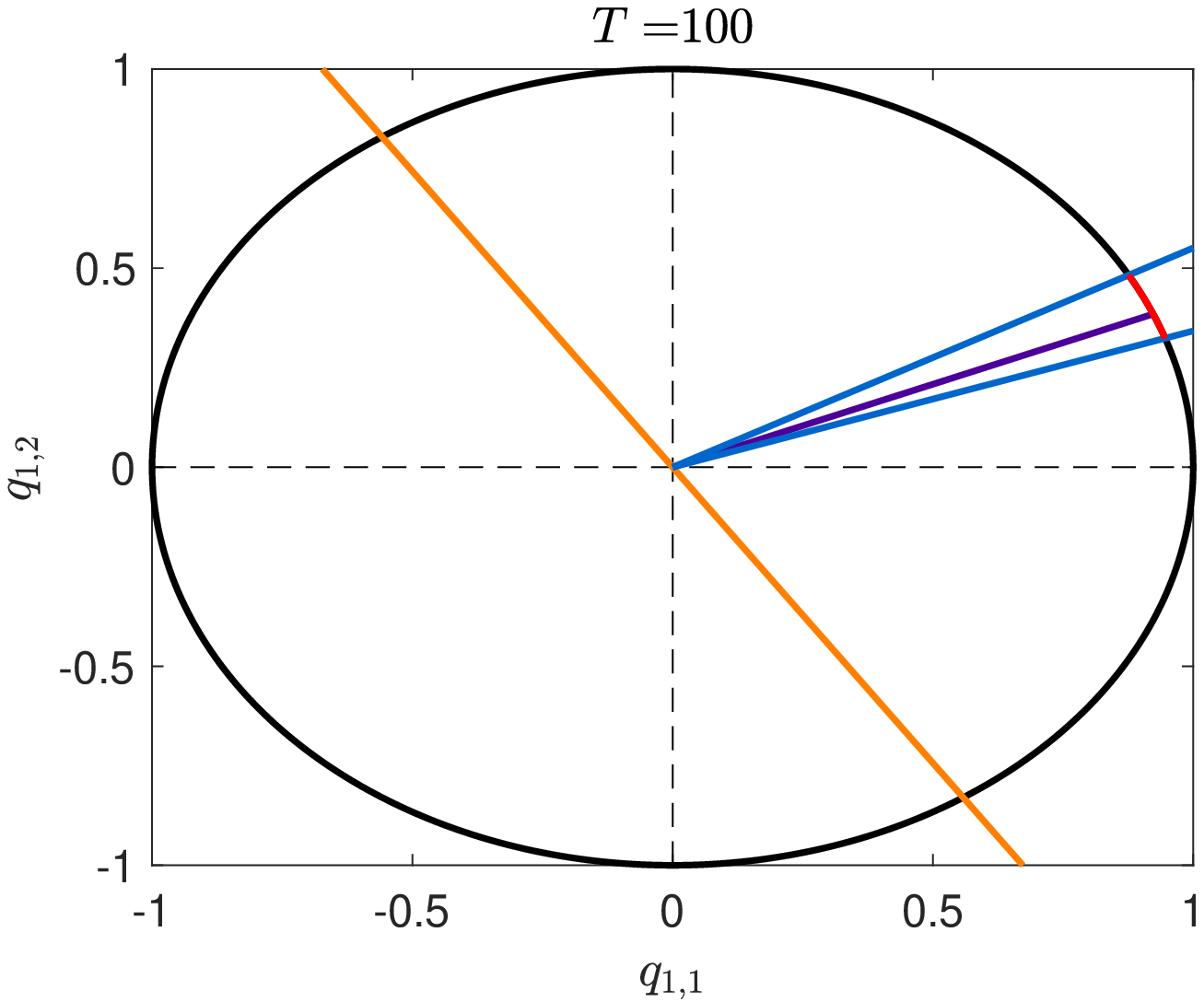} & \includegraphics[scale=0.35]{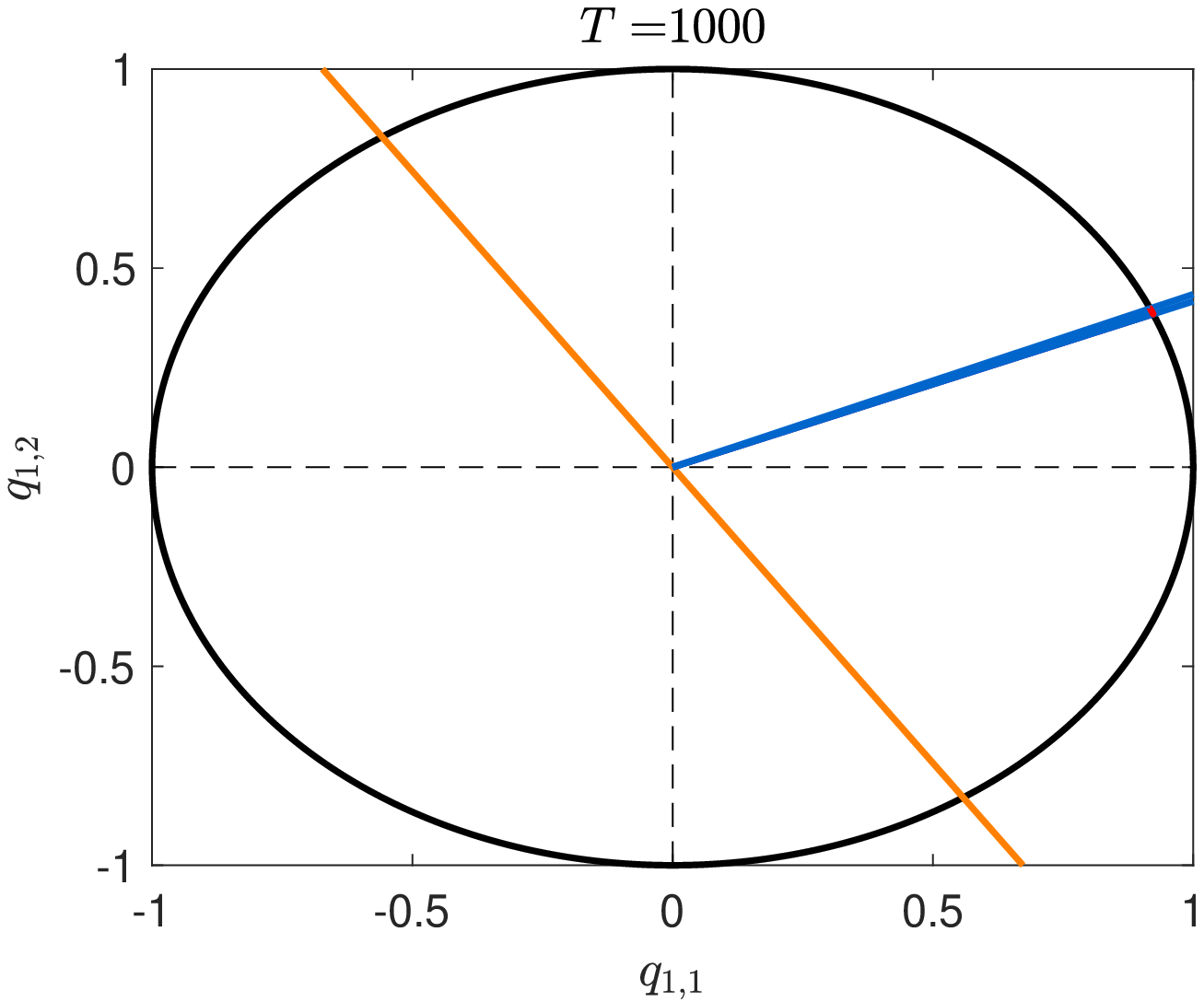}\\
    \end{tabular}
    \footnotesize \parbox[t]{0.65 in}{Notes:}\parbox[t]{5 in}{Purple line is true value of $\mathbf{q}_{1} = (q_{1,1},q_{1,2})'$; orange line is boundary of half-space generated by the sign normalization; blue lines are boundaries of half-spaces generated by `binding' shock-sign restrictions intersected with half-space generated by the sign normalisation; red line is intersection of all half-spaces with unit circle.} \\
\end{figure}

\section{Alternative algorithms for robust Bayesian inference}
\label{sec:appendixalgorithms}

Assume that the object of interest is an impulse response to the first structural shock. The upper bound of the conditional identified set for the horizon-$h$ impulse response of the $i$th variable to this shock given $\bm{\phi}$ and $\mathbf{Y}^{T}$ is the value function associated with the optimization problem
\begin{equation}\label{eq:objectivefunction}
  u(\bm{\phi},\mathbf{Y}^{T}) = \max_{\mathbf{Q}\in \mathcal{Q}(\bm{\phi}|\mathbf{Y}^{T},N,S)} \mathbf{c}_{i,h}'(\bm{\phi})\mathbf{q}_{1}.
\end{equation}
$l(\bm{\phi},\mathbf{Y}^{T})$ is obtained by minimising the same objective function subject to the same constraints. When $N(\bm{\phi},\mathbf{Q},\mathbf{Y}^{T})$ and $S(\bm{\phi},\mathbf{Q})$ only constrain $\mathbf{q}_{1}$, applying the change of variables $\mathbf{x} = \bm{\Sigma}_{tr}\mathbf{q}_{1}$ yields the optimization problem in \cite{Gafarov_Meier_Montiel-Olea_2018} with additional inequality restrictions that are functions of $\mathbf{Y}^{T}$.

Given a set of active inequality restrictions, \cite{Gafarov_Meier_Montiel-Olea_2018} provide an analytical expression for the value function and solution of this optimization problem. To find the bounds of the identified set, they compute these quantities for every possible combination of active restrictions and check which pair solves the optimization problem. Since the bounds are computed analytically at each set of active restrictions, this algorithm is computationally inexpensive as long as there is not a very large number of inequality restrictions. However, if $N(\bm{\phi},\mathbf{Q},\mathbf{Y}^{T})$ contains restrictions on the historical decomposition, all columns of $\mathbf{Q}$ are (nonlinearly) constrained and the analytical results are longer applicable. Similarly, the approach is not applicable when there are shock-sign or shock-rank restrictions on different structural shocks, or traditional sign restrictions on multiple columns of $\mathbf{Q}$. This approach may also be prohibitively slow when there is a large number of restrictions, which may be the case when there are shock-rank restrictions.\footnote{At most $n-1$ inequality constraints may be active at an optimum of the program in (\ref{eq:objectivefunction}), so the number of combinations of active constraints that must be checked when there are $s$ NR and $\tilde{s}$ traditional sign restrictions is $\sum_{k=0}^{n-1} {s+\tilde{s}+1\choose k}$. For example, in the empirical application below, when we consider a shock-rank restriction alongside traditional sign restrictions, there are $\sum_{k=0}^{n-1} {T+\tilde{s}+1\choose k} = 3.3876\times 10^{11}$~combinations of active restrictions to check.}

\cite{Amir-Ahmadi_Drautzburg_2021} propose an algorithm to determine whether the set of admissible values for $\mathbf{Q}$ is nonempty without recourse to random sampling from $\mathcal{O}(n)$. This algorithm can be more accurate and efficient than the simulation-based approach used in Algorithm~1, but it is applicable only when the columns of $\mathbf{Q}$ are subject to linear inequality restrictions, which is not the case when there are restrictions on the historical decomposition. However, practitioners may not always wish to impose restrictions on the historical decomposition. Accordingly, we describe an algorithm that can be used to conduct robust Bayesian inference without recourse to rejection sampling when there are shock-rank, shock-sign and/or traditional sign restrictions on a single column of $\mathbf{Q}$. The algorithm uses the approach in \cite{Amir-Ahmadi_Drautzburg_2021} to determine whether the conditional identified set for $\mathbf{q}_{1}$ is nonempty and replaces the Monte Carlo approximation of $[l(\bm{\phi},\mathbf{Y}^{T}),u(\bm{\phi},\mathbf{Y}^{T})]$ in Algorithm~1 with a numerical optimization step.

\bigskip

\noindent \textbf{Algorithm D.1.} \textit{Let $N(\bm{\phi},\mathbf{Y}^{T})\mathbf{q}_{1} \geq \mathbf{0}_{s\times 1}$ be the set of NR and let $S(\bm{\phi})\mathbf{q}_{1} \geq \mathbf{0}_{(\tilde{s}+1)\times 1}$ be the set of traditional sign restrictions (including the sign normalization). Assume the object of interest is $\eta_{i,1,h} = c_{i,h}'(\bm{\phi})\mathbf{q}_{1}$. Replace Steps~2 and 3 of Algorithm 1 with the following.
\begin{itemize}
  \item \textbf{Step~2}: Draw $\bm{\phi}$ from $\pi_{\bm{\phi}|\mathbf{Y}^{T}}$ and check whether the conditional identified set for $\mathbf{q}_{1}$ is empty by using the following subroutine.
      \begin{itemize}
      \item \textbf{2.1} Solve for the Chebyshev center $\{R,\tilde{\mathbf{q}}\}$ of the set
      \begin{equation}
        \{\tilde{\mathbf{q}}: (N(\bm{\phi},\mathbf{Y}^{T})', S(\bm{\phi})')'\tilde{\mathbf{q}} \geq \mathbf{0}_{(s+\tilde{s}+1)\times 1}, |\tilde{q}_{i}| \leq 1, i=1,\ldots,n\}.
       \end{equation}
       If $R > 0$, the conditional identified set is nonempty, so proceed to Step~3. Otherwise, repeat Step~2.
      \end{itemize}
  \item \textbf{Step~3}: Compute $l(\bm{\phi},\mathbf{Y}^{T})$ by solving the following constrained optimization problem with initial value $\mathbf{q}^{0} = \tilde{\mathbf{q}}/\lVert \tilde{\mathbf{q}} \rVert$:
    \begin{equation}
      l(\bm{\phi},\mathbf{Y}^{T}) = \min_{\mathbf{q}} c_{i,h}'(\bm{\phi})\mathbf{q} \quad \text{s.t.} \quad   (N(\bm{\phi},\mathbf{Y}^{T})', S(\bm{\phi})')'\tilde{\mathbf{q}} \geq \mathbf{0}_{(s+\tilde{s}+1)\times 1}, \mathbf{q}'\mathbf{q} = 1.
    \end{equation}
    Similarly, obtain $u(\bm{\phi},\mathbf{Y}^{T})$ by maximising $c_{i,h}'(\bm{\phi})\mathbf{q}$ subject to the same set of constraints.
\end{itemize}
}

\bigskip

Step~2.1 requires solving for the Chebyshev center of the set satisfying the narrative and traditional sign restrictions. The Chebyshev center $\tilde{\mathbf{q}}$ is the center of the largest ball with radius $R$ that can be inscribed within the set $\{\tilde{\mathbf{q}}: (N(\bm{\phi},\mathbf{Y}^{T})', S(\bm{\phi})')'\tilde{\mathbf{q}} \geq \mathbf{0}_{(s+\tilde{s}+1)\times 1}, |\tilde{q}_{i}| \leq 1, i=1,\ldots,n\}$, which is the intersection of the half-spaces generated by the inequality restrictions and the unit $n$-cube.\footnote{The restriction that $\tilde{\mathbf{q}}$ lies within the unit $n$-cube ensures that the problem is well-defined.} Letting $\mathbf{Z}_{k}'$ be the $k$th row of $(N(\bm{\phi},\mathbf{Y}^{T})', S(\bm{\phi})')'$, the Chebyshev center and radius can be obtained as the solution to the following problem (see, for example, \cite{Boyd_Vandenberghe_2004}):
\begin{equation*}
  \max_{\{R \geq 0,\tilde{\mathbf{q}}\}} R
\end{equation*}
subject to
\begin{align*}
    \mathbf{Z}_{k}'\tilde{\mathbf{q}} + R\lVert\mathbf{Z}_{k}\rVert &\geq 0, \quad k=1,\ldots,s+\tilde{s}+1 \\
    \tilde{q}_{i} + R &\leq 1, \quad i=1,\ldots,n.\\
    \tilde{q}_{i} - R &\geq -1, \quad i=1,\ldots,n.
\end{align*}
This is a linear program, which can be solved efficiently. If $R > 0$, then the conditional identified set for $\mathbf{q}_{1}$ is nonempty. If $\tilde{\mathbf{q}}$ is a Chebyshev center with $R > 0$, then $\tilde{\mathbf{q}}$ satisfies the inequality restrictions and $\lVert \tilde{\mathbf{q}} \rVert > 0$. $\mathbf{q}_{1}^{0} = \tilde{\mathbf{q}}/\lVert \tilde{\mathbf{q}} \rVert$ then has unit norm and satisfies the sign restrictions, so it can be used to initialize the optimization problem of Step~3. In practice, we solve this optimization problem using an interior-point algorithm within Matlab's `fmincon' optimizer.

\section{NR in the local projection framework}
\label{sec:localprojection}

\cite{Plagborg-Moller_Wolf_2020a} explain how to impose typical SVAR identifying restrictions in the local projection framework. This appendix explains how to impose NR in this framework.

\bigskip

\noindent\textbf{Local projection framework.} Assume the $n\times 1$ vector of data $\mathbf{y}_t$ is driven by an $n \times 1$ vector $\bm{\varepsilon}_t = (\varepsilon_{1t},\ldots,\varepsilon_{nt})'$ of structural shocks:
\begin{equation}
  \mathbf{y}_t = \mathbf{\mu}_{y} + \bm{\Theta}(L)\varepsilon_t, \quad \bm{\Theta}(L) \equiv \sum_{l=0}^{\infty}\bm{\Theta}_{l}L^{l},
\end{equation}
where $\{\bm{\Theta}_{l}\}_{l=0}^{\infty}$ is absolutely summable, $\bm{\Theta}(x)$ has full row rank for all complex scalars $x$ on the unit circle, and $\bm{\varepsilon}_t$ is independently and identically distributed with $\mathbb{E}(\bm{\varepsilon}_t) = \mathbf{0}_{n\times1}$ and $\mathbb{E}(\bm{\varepsilon}_{t}\bm{\varepsilon}_{t}') = \mathbf{I}_{n}$.

Consider the coefficient vectors $\{\bm{\beta}_{i,h}\}$ obtained from the $n\times(H+1)$ local projections
\begin{equation}
  y_{i,t+h} = \mu_{i,h} + \bm{\beta}_{i,h}'\mathbf{y}_t + \sum_{l=1}^{\infty}\bm{\delta}_{i,h,l}'\mathbf{y}_{t-l} + u_{i,h,t},
\end{equation}
where $i=1,\ldots,n$ and $h = 0,1,\ldots,H$. Let $\mathbf{C}_h = (\bm{\beta}_{1,h},\ldots,\bm{\beta}_{n,h})'$ denote the $n\times n$ matrix of horizon-$h$ projection coefficients. \cite{Plagborg-Moller_Wolf_2020a} show that the elements of $\mathbf{C}_h$ are the impulse responses of $\mathbf{y}_t$ at horizon $h$ to the Wold innovations $\mathbf{u}_t = \mathbf{y}_t - \mathrm{Proj}(\mathbf{y}_t|\{\mathbf{y}_{t-l}\}_{l=1}^{\infty})$. The Wold innovations are equal to the residuals of the local projection at $h=1$, so $\mathbf{u}_t = (u_{1,1,t},\ldots,u_{n,1,t})'$. Let $\mathrm{Var}(\mathbf{u}_t) = \bm{\Sigma} = \bm{\Sigma}_{tr}\bm{\Sigma}_{tr}'$, where $\bm{\Sigma}_{tr}$ is the lower-triangular Cholesky factor of $\bm{\Sigma}$ with strictly positive diagonal elements.

Assume that the structural shocks are invertible, in the sense that they can be recovered as a linear combination of $\mathbf{u}_t$: $\bm{\varepsilon}_t = \mathbf{G}\mathbf{u}_t$, where $\mathbf{G}$ is a full-rank $n\times n$ matrix. Reparameterize $\mathbf{G}$ as $\mathbf{G}^{-1} = \bm{\Sigma}_{tr}\mathbf{Q}$, where $\mathbf{Q} \in \mathcal{O}(n)$. The diagonal elements of $\mathbf{G}$ are normalized to be nonnegative, so $\mathrm{diag}(\mathbf{Q}'\bm{\Sigma}_{tr}^{-1}) \geq \mathbf{0}_{n\times 1}$.

Let $\bm{\phi} = (\mathrm{vec}(\mathbf{C}_0)',\ldots,\mathrm{vec}(\mathbf{C}_H)',\mathrm{vech}(\bm{\Sigma}_{tr})')'$. The object of interest is the (structural) impulse response, which is an element of $\bm{\Theta}_{l}$ . Invertibility of the shocks implies that the structural impulse responses can be obtained as rotations of the reduced-form impulse responses $\mathbf{C}_h$. In particular, given that $\mathbf{u}_t = \bm{\Sigma}_{tr}\mathbf{Q}\bm{\varepsilon}_t$, $\bm{\Theta}_{h} = \mathbf{C}_{h}\bm{\Sigma}_{tr}\mathbf{Q}$.

\bigskip

\noindent\textbf{Imposing NR.} The invertibility assumption implies that
\begin{equation}
    \bm{\varepsilon}_{t} = \mathbf{Q}'\bm{\Sigma}_{tr}^{-1}\mathbf{u}_{t}.
\end{equation}
The $i$th structural shock at time $t$ is therefore
\begin{equation}
  \varepsilon_{it}(\bm{\phi},\mathbf{Q},\mathbf{u}_{t}) = \mathbf{e}_{i}'\mathbf{Q}'\bm{\Sigma}_{tr}^{-1}\mathbf{u}_{t} = (\bm{\Sigma}_{tr}^{-1}\mathbf{u}_{t})'\mathbf{q}_{i}.
\end{equation}
A shock-sign restriction is therefore a linear inequality restriction on $\mathbf{q}_{i}$ that depends on the reduced-form parameter $\bm{\Sigma}_{tr}$ and the Wold innovations $\mathbf{u}_{t}$.

The historical decomposition is the cumulative contribution of the $j$th shock to the observed unexpected change in the $i$th variable between periods $t$ and $t+h$:
\begin{equation}
  H_{i,j,t,t+h} = \sum_{l=0}^{h} \mathbf{e}_{i}'\mathbf{C}_{l}\bm{\Sigma}_{tr}\mathbf{Q}\mathbf{e}_{j}\mathbf{e}_{j}'\bm{\varepsilon}_{t+h-l} = \sum_{l=0}^{h} \mathbf{C}_{l}\bm{\Sigma}_{tr}\mathbf{q}_{j}\mathbf{q}_{j}'\bm{\Sigma}_{tr}^{-1}\mathbf{u}_{t+h-l}.
\end{equation}
This is a function of the reduced-form parameter $\bm{\Sigma}_{tr}$, the reduced-form impulse responses $\{\mathbf{C}_{l}\}_{l=0}^{h}$ and the Wold innovations $\mathbf{u}_{t}$.

Given a set of traditional and narrative sign restrictions, the upper bound of the conditional identified set for a particular impulse response of interest, $\mathbf{e}_{i,n}'\bm{\Theta}_{h}\mathbf{e}_{j^{*},n}$, is the solution $u(\bm{\phi},\mathbf{Y}^{T})$ of the following constrained optimization problem:
\begin{equation}
  u(\bm{\phi},\mathbf{Y}^{T}) = \max_{\mathbf{Q}} \mathbf{e}_{i,n}'\mathbf{C}_{h}\bm{\Sigma}_{tr}\mathbf{Q}\mathbf{e}_{j^{*},n}
\end{equation}
subject to
\begin{equation*}
    N(\bm{\phi},\mathbf{Q},\mathbf{Y}^{T}) \geq \mathbf{0}_{s\times 1}, \quad \mathbf{Q}'\mathbf{Q} = \mathbf{I}_n, \quad \mathrm{diag}(\mathbf{Q}'\bm{\Sigma}_{tr}^{-1}) \geq \mathbf{0}_{n\times 1}.
\end{equation*}
The lower bound of the conditional identified set is the solution of the corresponding minimization problem.

\subsection*{Remarks}

\begin{itemize}
  \item The key difference between this framework and the SVAR framework is that, in the SVAR framework, the reduced-form impulse responses would be obtained from the VMA representation of the reduced-form VAR. In contrast, here they are obtained directly from local projections. Otherwise, structural impulse responses are obtained by rotating reduced-form impulse responses in exactly the same way as in the SVAR.
  \item As in \cite{Plagborg-Moller_Wolf_2020a}, we have assumed that there is an infinite number of lags appearing as controls in the local projections. Under this assumption, the reduced-form impulse responses will coincide with those from a VAR($\infty$) at all horizons. The horizon-1 local projection innovations will also coincide with the one-step-ahead forecast errors from the VAR, so the covariance matrix of these innovations will coincide. Consequently, the conditional identified set for the structural impulse responses will also coincide at all horizons. See \cite{Plagborg-Moller_Wolf_2020a} for discussions of the finite-lag case and the choice between VARs and local projections.
  \item Given posterior draws of $\bm{\phi}$, one could conduct robust Bayesian inference in exactly the same way as in the SVAR case. However, obtaining the posterior of $\bm{\phi}$ requires specifying a joint prior over the parameters in $\bm{\phi}$ and the parameters governing the system of local projection residuals, which are in general serially correlated (for example, see \cite{Lusompa_2020}).
\end{itemize}

\section{NR as proxy variables}
\label{sec:narrativeproxy}

\cite{Plagborg-Moller_Wolf_2020b} point out that information about the sign of a particular structural shock can be recast as a variable that can be used to point-identify impulse responses in a proxy SVAR or local projection framework.\footnote{For a related approach, see \cite{Budnik_Runstler_2020}.} Specifically, consider the variable that takes value one when the structural shock is known to be positive, minus one when it is known to be negative and zero otherwise. This `narrative proxy' will clearly be positively correlated with the structural shock of interest. Since the proxy depends only on the structural shock of interest, it will also be contemporaneously uncorrelated with the other structural shocks. It can therefore be used to point-identify the impulse responses to the shock of interest in a proxy SVAR (e.g., \cite{Mertens_Ravn_2013} and Montiel-Olea, Stock and Watson (2020)\nocite{Montiel-Olea_Stock_Watson_2020}). Since the instrument is additionally uncorrelated with leads and lags of all structural shocks, it could alternatively be used as an instrument in a local projection, which does not require assuming invertibility (e.g., \cite{Stock_Watson_2018}).\footnote{Note that the covariance between the narrative proxy and the structural shock of interest will converge to zero asymptotically when there is a fixed number of shock-sign restrictions used to generate the proxy. In this case, frequentist inference could be conducted using weak-instrument robust methods (e.g., Montiel-Olea et al. (2020)\nocite{Montiel-Olea_Stock_Watson_2020}).}

This approach is valid when there are shock-sign restrictions only, but more generally it is unclear how one would encode the information underlying richer sets of NR (e.g., restrictions on the historical decomposition) as an instrument without discarding potentially useful identifying information. Additionally, when there are only a small number of shock-sign restrictions used to generate the instrument, the point estimator of the impulse response will be sensitive to the realization of the data in the periods corresponding to the shock-sign restrictions. We illustrate this point below using the bivariate example of Section~\ref{sec:bivariate}.

\bigskip

\noindent\textbf{Proxy variables in the bivariate example.} Assume there is a variable $Z_{t}$ satisfying $\mathbb{E}(Z_{t}\varepsilon_{1t}) \neq 0$ and $\mathbb{E}(Z_{t}\varepsilon_{2t}) = 0$. After expressing $\varepsilon_{2t}$ in terms of $\mathbf{y}_{t}$ and the parameters, the exogeneity condition implies that
\begin{equation}
  \sigma_{22}\mathbb{E}(Z_{t}y_{1t})\sin\theta = \mathbb{E}(Z_{t}(\sigma_{11}y_{2t} - \sigma_{21}y_{1t}))\cos\theta.
\end{equation}
If the instrument is not relevant, so that $\mathbb{E}(Z_{t}\varepsilon_{1t}) = 0$, the restriction $\mathbb{E}(Z_{t}\varepsilon_{2t}) = 0$ carries no information about $\theta$, since $\mathbb{E}(Z_{t}y_{1t}) = 0$ and $\mathbb{E}(Z_{t}(\sigma_{11}y_{2t} - \sigma_{21}y_{1t})) = 0$. Otherwise,
\begin{equation}\label{eq:tantheta}
  \tan\theta = \frac{\mathbb{E}(Z_{t}(\sigma_{11}y_{2t} - \sigma_{21}y_{1t}))}{\sigma_{22}\mathbb{E}(Z_{t}y_{1t})}.
\end{equation}
This equation has two solutions in $[-\pi,\pi]$, one of which will be ruled out by the sign normalization restrictions. For example, if $\sigma_{21} < 0$ and the term on the right-hand side of Equation~(\ref{eq:tantheta}) (henceforth denoted by $C$) is positive, then $\theta$ is either equal to $\arctan(C) - \pi$ or $\arctan(C)$. The sign normalization implies that $\theta \in [\arctan(\sigma_{22}/\sigma_{21}),\arctan(\sigma_{22}/\sigma_{21}) + \pi]$, which rules out the first solution, so $\theta = \arctan(C)$. If $C$ is negative, then $\theta$ is either equal to $\arctan(C)$ or $\arctan(C) + \pi$. If $C > \sigma_{22}/\sigma_{21}$, then the sign normalization selects the first solution, otherwise it selects the second solution. Similar arguments apply when $\sigma_{21} > 0$.\footnote{When $\sigma_{21} > 0$, the sign normalization restricts $\theta$ to lie in $\left[-\pi + \arctan\left(\frac{\sigma_{22}}{\sigma_{21}}\right),\arctan\left(\frac{\sigma_{22}}{\sigma_{21}}\right)\right]$. If $C > \frac{\sigma_{22}}{\sigma_{21}}$, the sign normalization implies that $\theta = \arctan(C)-\pi$. Otherwise, the sign normalization implies that $\theta = \arctan(C)$.}

Consider the case where information about the sign of the first structural shock is recast as a binary variable. Specifically, as in the shock-sign example, assume the econometrician knows that $\varepsilon_{1k} \geq 0$ for some $k \in \left\{1,\ldots,T\right\}$, and let $Z_{k} = \mathrm{sgn}(\varepsilon_{1k})$ with $Z_{t} = 0$ for $t \neq k$. What happens if the econometrician imposes the identifying restriction that $\mathbb{E}(Z_{t}\varepsilon_{2t}) = 0$?

Maintaining the assumption that $\bm{\phi}$ is known with $\sigma_{21} < 0$, in the case where $(\sigma_{11}y_{2k} - \sigma_{21}y_{1k}) > 0$ and $y_{1k} > 0$, an analogue estimator of $\theta$ is
\begin{align}
  \hat{\theta} &= \arctan\left(\frac{\frac{1}{T}\sum_{t=1}^{T}Z_{t}(\sigma_{11}y_{2t} - \sigma_{21}y_{1t})}{\sigma_{22}\frac{1}{T}\sum_{t=1}^{T}Z_{t}y_{1t}}\right) \notag \\
         &= \arctan\left(\frac{\sigma_{11}y_{2k} - \sigma_{21}y_{1k}}{\sigma_{22}y_{1k}}\right). \label{eq:thetahat}
\end{align}
Note that this is equal to the estimator that would be obtained if one were to impose the `narrative zero restriction' $\varepsilon_{2k} = 0$. Additionally, $\hat{\theta}$ lies within the conditional identified set under the shock-sign restriction $\varepsilon_{1k} \geq 0$. To see this, first note that $\hat{\theta}$ lies in the range $(0,\pi/2)$, since the argument entering the $\arctan(.)$ function is positive by assumption. The conditional identified set for $\theta$ under the shock-sign restriction in this case is given by (\ref{eq:thetaset1}). The lower bound of the conditional identified set is bounded above by zero, while the upper bound is bounded below by $\pi/2$, so $\hat{\theta}$ necessarily lies within this conditional identified set.

How does this estimator relate to the true value of $\theta$? Assume that the data are generated by a process with parameter $\theta_{0} \in (0,\frac{\pi}{2})$ (with $\mathbf{Q}$ equal to the rotation matrix). Replacing $y_{1k}$ and $y_{2k}$ in (\ref{eq:thetahat}) using $\mathbf{y}_{k} = \mathbf{A}_{0}^{-1}\bm{\varepsilon}_{k}$ yields an expression for $\hat{\theta}$ in terms of the true parameters and the underlying structural shocks:
\begin{multline}
  \hat{\theta} = \arctan\Bigg(\left(\sigma_{22}\left(\sigma_{11}\cos\theta_{0}\varepsilon_{1k} - \sigma_{11}\sin\theta_{0}\varepsilon_{2k}\right)\right)^{-1}\Big[\sigma_{11}\Big[\left(\sigma_{21}\cos\theta_{0}+\sigma_{22}\sin\theta_{0}\right)\varepsilon_{1k}+ \\
  \left(\sigma_{22}\cos\theta_{0}-\sigma_{21}\sin\theta_{0}\right)\varepsilon_{2k}\Big] - \sigma_{21}\left(\sigma_{11}\cos\theta_{0}\varepsilon_{1k}-\sigma_{11}\sin\theta_{0}\varepsilon_{2k}\right)\Big]\Bigg).
\end{multline}
If $\varepsilon_{2k} = 0$, we have that $\hat{\theta} = \theta_{0}$. Otherwise, $\hat{\theta}$ will not in general coincide with $\theta_{0}$. For example, for $\varepsilon_{2k} \neq 0$ and $\varepsilon_{1k} \approx 0$, $\hat{\theta} \approx \arctan(\cot\theta_{0}) = \frac{\pi}{2} - \theta_{0}$.\footnote{This follows from the fact that $\arctan(x) + \arctan\left(\frac{1}{x}\right) = \frac{\pi}{2}$ for $x > 0$.} In this case, the impulse-response estimator is
\begin{equation}
    \hat{\eta} = \sigma_{11}\cos\hat{\theta} \approx \sigma_{11}\cos\left(\frac{\pi}{2} - \theta_{0}\right) = -\sigma_{11}\sin\theta_{0},
\end{equation}
which is the true impulse response of the first variable to the second shock, rather than the first shock. In general, the estimator of the impulse response may be sensitive to the value of the second shock in period $k$, since it is based solely on the data in period $k$.

The impulse response considered above is to a standard-deviation shock in $\varepsilon_{1t}$ (i.e., an \textit{absolute} impulse response). In the literature that uses proxies to identify the effects of macroeconomic shocks, it is common to use the \textit{relative} impulse response, which is the impulse response to a shock that raises a particular variable by one unit. For example, the impulse response of $y_{2t}$ to a shock that raises the first variable by one unit, $\tilde{\eta}_{2}$, is the ratio of the absolute impulse response of the second variable to the absolute impulse response of the first variable. The analogue estimator of the absolute impulse response of the first variable is
\begin{align}
  \hat{\eta} &= \sigma_{11}\cos\hat{\theta} \notag \\
  &= \sigma_{11}\cos\left(\arctan\left(\frac{\sigma_{11}y_{2k}-\sigma_{21}y_{1k}}{\sigma_{22}y_{1k}}\right)\right) \notag \\
  &= \frac{\sigma_{11}\sigma_{22}y_{1k}}{\sqrt{\sigma_{22}^{2}y_{1k}^{2} + (\sigma_{11}y_{2k} - \sigma_{21}y_{1k})^{2}}},
\end{align}
where the last line follows from the fact that $\cos(\arctan(x)) = (1+x^{2})^{-1/2}$. The estimator for the absolute impulse response of the second variable is
\begin{align}
  \hat{\eta}_{2} &= \sigma_{21}\cos\hat{\theta} + \sigma_{22}\sin\hat{\theta} \notag \\
  &= \sqrt{\sigma_{21}^{2} + \sigma_{22}^{2}}\cos\left(\hat{\theta} - \arctan\left(\frac{\sigma_{22}}{\sigma_{21}}\right)\right) \notag \\
  &= \sqrt{\sigma_{21}^{2} + \sigma_{22}^{2}}\left[\cos\hat{\theta}\cos\left(\arctan\left(\frac{\sigma_{22}}{\sigma_{21}}\right)\right) + \sin\hat{\theta}\sin\left( \arctan\left(\frac{\sigma_{22}}{\sigma_{21}}\right)\right)\right] \notag \\
  &= \frac{\sigma_{11}\sigma_{22}y_{2k}}{\sqrt{\sigma_{22}^{2}y_{1k}^{2} + (\sigma_{11}y_{2k} - \sigma_{21}y_{1k})^{2}}},
\end{align}
where we have used that $a\cos x + b\sin x = \sqrt{a^{2} + b^{2}}\cos(x - \alpha)$ with $\tan\alpha = b/a$, $\cos(x - y) = \cos x \cos y + \sin x \sin y$, $\cos(\arctan(x)) = (1+x^{2})^{-1/2}$ and $\sin(\arctan(x)) = x(1+x^{2})^{-1/2}$. Consequently,
\begin{equation}
  \tilde{\eta}_{2} = \frac{\hat{\eta}_{2}}{\hat{\eta}} = \frac{y_{2k}}{y_{1k}}.
\end{equation}
The estimator of the relative impulse response will clearly also be sensitive to the realizations of the structural shocks in period $k$. Similar to above, if $\varepsilon_{2k} = 0$, then $\tilde{\eta}_{2}$ will be equal to the true relative impulse response of the second variable to the first shock. If $\varepsilon_{1k} \approx 0$ and $\varepsilon_{2k} \neq 0$, then $\tilde{\eta}_{2}$ will be approximately equal to the true relative impulse response of the second variable to the second shock.

\end{appendices}

\newpage
\bibliographystyle{ecta}
\bibliography{NarrativeRestrictions_Bibliography}
\addcontentsline{toc}{section}{References}

\end{document}